\documentclass{article}

\usepackage{amsmath, amssymb}
\usepackage{fullpage}
\usepackage{amsthm, mathrsfs}
\usepackage{color}
\usepackage[affil-it]{authblk}
\usepackage{dsfont}
\usepackage{hyperref}
\usepackage{graphicx}
\usepackage{bbm}

\newtheorem{theorem}{Theorem}[section]
\newtheorem{proposition}[theorem]{Proposition}
\newtheorem{lemma}[theorem]{Lemma}

\newtheorem{remark}[theorem]{Remark}

\numberwithin{equation}{section}

\newcommand{\supp}{\mathrm{supp}}

\begin{document}

\title{Phase mixing and the Vlasov equation in cosmology}

\author{Martin Taylor}
\author{Renato Velozo Ruiz}

\affil{\small Imperial College London,
Department of Mathematics,
South~Kensington~Campus,~London~SW7~2AZ,~United~Kingdom
\vskip.3pc
martin.taylor@imperial.ac.uk
\vskip.2pc 
r.velozoruiz@imperial.ac.uk
\ }

\date{December 3, 2025}

\maketitle

\begin{abstract}
	We consider the Vlasov equation on slowly expanding isotropic homogeneous tori, described by the Friedmann--Lema\^itre--Robertson--Walker cosmological spacetimes.  For expansion rate $t^q$, with $0< q<\frac{1}{2}$ (excluding certain exceptional values), we show that the spatial density decays at the rate $t^{-6q}$ and that, when the spatial average is removed, the density decays at an enhanced rate due to a phase mixing effect.  This enhancement is polynomial for Sobolev initial data and super-polynomial, but sub-exponential, for real analytic initial data.  We further show that, when the expansion rate is the borderline $t^{\frac{1}{2}}$ --- the rate which describes a radiation filled universe --- a degenerate phase mixing effect results in a logarithmic enhancement for Sobolev initial data and a super-logarithmic enhancement (in fact, a gain of $\exp(-\mu(\log t)^{\epsilon})$ for some $\mu,\epsilon>0$) for analytic initial data.  The proof is based on a collection of commuting vector fields, and certain combinatorial properties of an associated collection of differential operators.  The vector fields are not explicit, but are shown to have good properties when $t$ is large with respect to the momentum support of the solution.  A physical space dyadic localisation is employed to treat non-compactly supported (in particular, non-trivial real analytic) but suitably decaying solutions.
\end{abstract}

\tableofcontents

\section{Introduction}

Phase mixing effects for kinetic equations, and the phenomenon of Landau damping, have been famously exhibited in homogeneous settings in plasma physics \cite{Lan}, and around inhomogeneous isolated systems in gravitational physics \cite{Lyn62, Lyn67}.  This article concerns a related mixing effect which occurs, in the gravitational context, in homogeneous cosmology.  For more on the Vlasov equation in cosmology, see \cite{Ehl, ElMaMa, EllEls, RySh, Wein08}.  The main results of this article concern the Vlasov equation on a Friedmann--Lema\^itre--Robertson--Walker (FLRW) family of spacetimes.  Each member of this family describes a homogeneous, isotropic cosmology, evolving from a big bang singularity at $t = 0$ and expanding indefinitely as $t \to \infty$.  We exhibit a phase mixing effect when the expansion rate is suitably slow.

\subsection{The main results}

For each given $q > 0$, the FLRW spacetime $(\mathcal{M},g)$, with expansion rate $t^q$, is given by
\begin{equation} \label{eq:FLRWgeneral}
	\mathcal{M} = (0,\infty) \times \mathbb{T}^3,
	\qquad
	g = -dt^2 + t^{2q} \big( (dx^1)^2 + (dx^2)^2 + (dx^3)^2 \big),
\end{equation}
where $(x^1,x^2,x^3)$ denote Cartesian coordinates on $\mathbb{T}^3$.  Each such spacetime describes a homogeneous isotropic universe.
The Vlasov equation on \eqref{eq:FLRWgeneral} takes the form
\begin{equation} \label{eq:Vlasovgeneral}
	\partial_t f
	+
	\frac{p^i}{p^0} \partial_{x^i} f
	-
	\frac{2q}{t} p^{i} \partial_{p^i} f
	=
	0,
	\qquad
	p^0 = \sqrt{1+ t^{2q} \vert p \vert^2},
\end{equation}
where $\vert p \vert^2 = (p^1)^2+(p^2)^2+(p^3)^2$.  The summation convention is adopted throughout, whereby repeated Latin indices indicate a summation over $1,2,3$, and repeated Greek indices indicate a summation over $0,1,2,3$.  For given $f\colon (0,\infty) \times \mathbb{T}^3 \times \mathbb{R}^3 \to [0,\infty)$, the corresponding spatial density $\rho \colon (0,\infty) \times \mathbb{T}^3 \to [0,\infty)$ is defined by
\[
	\rho(t,x) = \int_{\mathbb{R}^3} f(t,x,p) dp.
\]
Each solution of \eqref{eq:Vlasovgeneral} describes an ensemble of collisionless, unit mass particles, evolving with respect to the fixed FLRW gravitational background \eqref{eq:FLRWgeneral}.

\subsubsection{Phase mixing on slowly expanding FLRW spacetimes}

The first main result of this article concerns a mixing effect in solutions of \eqref{eq:Vlasovgeneral} when $0 < q<\frac{1}{2}$.

\begin{theorem}[Phase mixing for the Vlasov equation on slowly expanding FLRW spacetimes] \label{thm:main1}
	Consider $0 < q < \frac{1}{2}$ such that $\frac{1}{2q}$ is not an integer (i.\@e.\@ $q \neq \frac{1}{4}, \frac{1}{6}, \frac{1}{8},\ldots$), $k \geq 2$, and some $f_1 \in H^k_q(\mathbb{T}^3 \times \mathbb{R}^3)$.  Let $f$ be the unique solution of equation \eqref{eq:Vlasovgeneral} on $[1,\infty) \times \mathbb{T}^3\times \mathbb{R}^3$ such that $f(1,x,p) = f_1(x,p)$.  The spatial density  of $f$ satisfies, for all $t \geq 1$,
	\begin{equation} \label{eq:main2}
		\sup_{x \in \mathbb{T}^3}\big\vert \rho(t,x) - \overline{\rho}(t) \big\vert
		\lesssim
  		\frac{\Vert f_1 \Vert_{H^k_q}}{t^{6q + k(1-2q)}},
		\qquad
		\text{where}
		\quad
		\overline{\rho}(t)
		=
		\frac{1}{t^{6q}} \int_{\mathbb{T}^3} \int_{\mathbb{R}^3} f_1(x,p) dp dx.
	\end{equation}
	If $f_1$ lies in the analytic space $f_1 \in H^{\omega}_q(\mathbb{T}^3 \times \mathbb{R}^3)$ (see Section \ref{subsec:functionspaces}) then the spatial density satisfies, for all $t \geq 1$,
	\begin{equation} \label{eq:main3}
		\sup_{x \in \mathbb{T}^3}\big\vert \rho(t,x) - \overline{\rho}(t) \big\vert
		\lesssim
		\frac{\Vert f_1 \Vert_{H^{\omega}_q}}{t^{6q}}
		e^{- \mu \, t^{\frac{1-2q}{10}}}
		,
	\end{equation}
	for some $\mu >0$.
\end{theorem}

The spaces $H^k_q(\mathbb{T}^3\times \mathbb{R}^3)$ are $p$-weighted versions of the standard Sobolev spaces (with $p$ weights depending on $q$), and $H^{\omega}_q(\mathbb{T}^3\times \mathbb{R}^3)$ is an associated space of analytic functions.  The constant $\mu$ is related to a radius of convergence for the analytic norm.  See Section \ref{subsec:functionspaces} below.

\subsubsection{Degenerate phase mixing on the radiation FLRW spacetime}

The spacetime \eqref{eq:FLRWgeneral} with $q =\frac{1}{2}$ is known as the \emph{radiation FLRW spacetime}.  The second main result concerns a degenerate mixing effect for solutions of \eqref{eq:Vlasovgeneral} when $q =\frac{1}{2}$.

\begin{theorem}[Degenerate phase mixing for the Vlasov equation on the radiation FLRW spacetime] \label{thm:main2}
	Consider some $k \geq 2$ and $f_1 \in H^k_{\log}(\mathbb{T}^3 \times \mathbb{R}^3)$.  Let $f$ be the unique solution of equation \eqref{eq:Vlasovgeneral}, with $q=\frac{1}{2}$, on $[1,\infty) \times \mathbb{T}^3\times \mathbb{R}^3$ such that $f(1,x,p) = f_1(x,p)$.  The spatial density of $f$ satisfies, for all $t \geq 1$,
	\begin{equation} \label{eq:mainradiation2}
		\sup_{x \in \mathbb{T}^3}\big\vert \rho(t,x) - \overline{\rho}(t) \big\vert
		\lesssim
  		\frac{\Vert f_1 \Vert_{H^k_{\log}}}{t^3 (\log (1+t))^k},
		\qquad
		\text{where}
		\quad
		\overline{\rho}(t)
		=
		\frac{1}{t^3} \int_{\mathbb{T}^3} \int_{\mathbb{R}^3} f_1(x,p) dp dx.
	\end{equation}
	If $f_1$ lies in the analytic space $f_1 \in H^{\omega}_{\log}(\mathbb{T}^3 \times \mathbb{R}^3)$ then the spatial density satisfies, for all $t \geq 1$,
	\begin{equation} \label{eq:mainradiation3}
		\sup_{x \in \mathbb{T}^3}\big\vert \rho(t,x) - \overline{\rho}(t) \big\vert
		\lesssim
		\frac{\Vert f_1 \Vert_{H^{\omega}_{\log}}}{t^3} e^{-\mu(\log (1+t))^{\frac{1}{58}}}
		,
	\end{equation}
	for $\mu = (\lambda(f_1)/2)^{\frac{1}{58}}$ (with $\lambda(f_1)$ defined in Section \ref{subsec:functionspaces} below).
\end{theorem}

The spaces $H^k_{\log}(\mathbb{T}^3\times \mathbb{R}^3)$ are, again, $p$-weighted versions of the standard Sobolev spaces and $H^{\omega}_{\log}(\mathbb{T}^3\times \mathbb{R}^3)$ is an associated space of analytic functions, in which every $\partial_x$ derivative is now weighted by an additional $\log (2+\vert p \vert)$ factor.  See Section \ref{subsec:functionspaces} below.

\begin{figure}
	\centering
	\includegraphics[angle=270,scale=0.4]{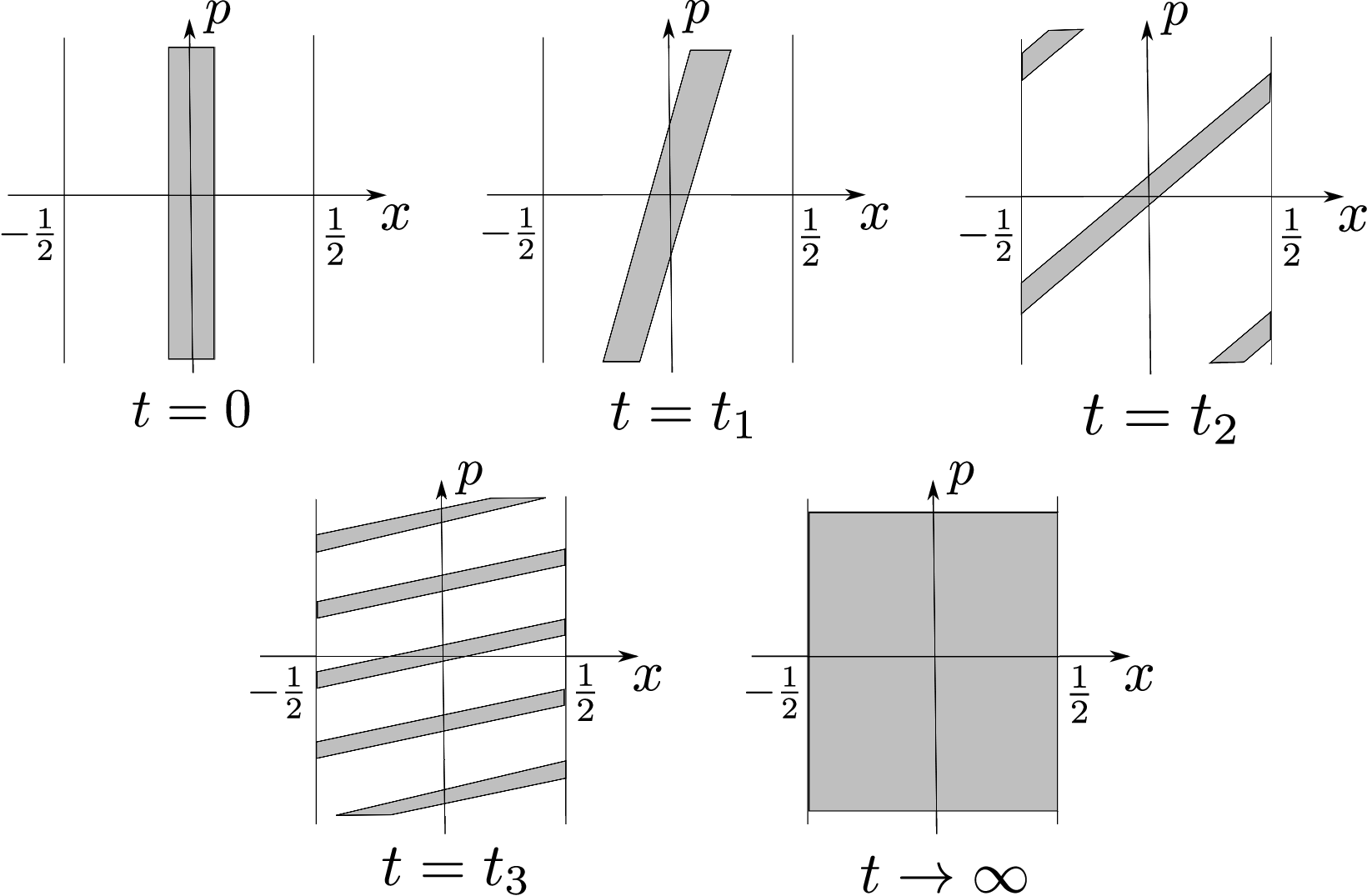}
	\caption{The phase mixing effect in the absence of expansion.}
	\label{fig:mixing1}
\end{figure}

\subsubsection{Remarks on the main theorems}

See Figure \ref{fig:mixing1} for an illustration of the phase mixing phenomena in the $q=0$ (i.\@e.\@ the standard, non-expanding) case.  Particles initially at the origin, with momentum $p$, evolve to position $t p$ after time $t$, and thus eventually, by periodicity, re-enter the picture on the left and right of the torus.  The support of a solution $f$ of \eqref{eq:Vlasovgeneral}, with $q=0$, which is initially localised around $\{x=0\}$ will therefore ``mix'' in evolution, as depicted in Figure \ref{fig:mixing1}.  For general $q>0$, the vectors $t^{-q} \partial_{x^i}$ and $t^{-q} \partial_{p^i}$, for $i=1,2,3$, have unit length and so particles initially at the origin, with momentum $p$, evolve to position $(t-1) \cdot t^{-q} p$ at time $t$.  Since the the torus is now expanding at rate $t^q$, a version of the above effect may occur if there are $t >1$ such that $t^{q} \leq t^{1-q}$, i.\@e.\@ if $q \leq \frac{1}{2}$.  See Figure \ref{fig:mixing2}.

For each $q> 0$, \eqref{eq:FLRWgeneral} describes a cosmological spacetime which evolves from a big bang singularity as $t \to 0^+$ --- at which the Kretschmann scalar
\[
	R^{\alpha \beta \gamma \delta} R_{\alpha \beta \gamma \delta}
	=
	\frac{12}{t^4} q^2(2q^2-2q+1),
\]
blows up at rate $t^{-4}$ --- and expands indefinitely as $t\to \infty$.  The spacetimes \eqref{eq:FLRWgeneral} arise as solutions of the Einstein equations
\begin{equation} \label{eq:Einstein}
	Ric(g)_{\mu \nu} - \frac{1}{2} R(g) g_{\mu \nu} = T_{\mu \nu},
\end{equation}
for various matter models.  Some notable examples include
\begin{itemize}
	\item
		Einstein--Euler: $\frac{1}{3} \leq q \leq \frac{2}{3}$;
	\item
		Einstein--scalar field: $q = \frac{1}{3}$;
	\item
		Einstein--nonlinear scalar field $q>0$;
	\item
		Einstein--massless Vlasov: $q = \frac{1}{2}$;
	\item
		Einstein--massless Boltzmann: $q = \frac{1}{2}$;
	\item
		Einstein--massive Vlasov: $q \sim \frac{2}{3}$,
\end{itemize}
which are described in detail in Appendix \ref{subsec:FLRW}.  Many of these examples fall within the range of $q$ covered by Theorem \ref{thm:main1} and Theorem \ref{thm:main2}.  In particular the case $q=\frac{1}{2}$, considered in Theorem \ref{thm:main2}, describes a radiation filled universe and arises as a solution of Einstein--Euler with \emph{radiation equation of state} and also of Einstein--massless Vlasov.  Accordingly, the case $q=\frac{1}{2}$ is typically used to describe the \emph{radiation dominated epoch} of the early universe \cite{ElMaMa, KoTu}.

\begin{figure}
	\centering
	\includegraphics[angle=270,scale=0.45]{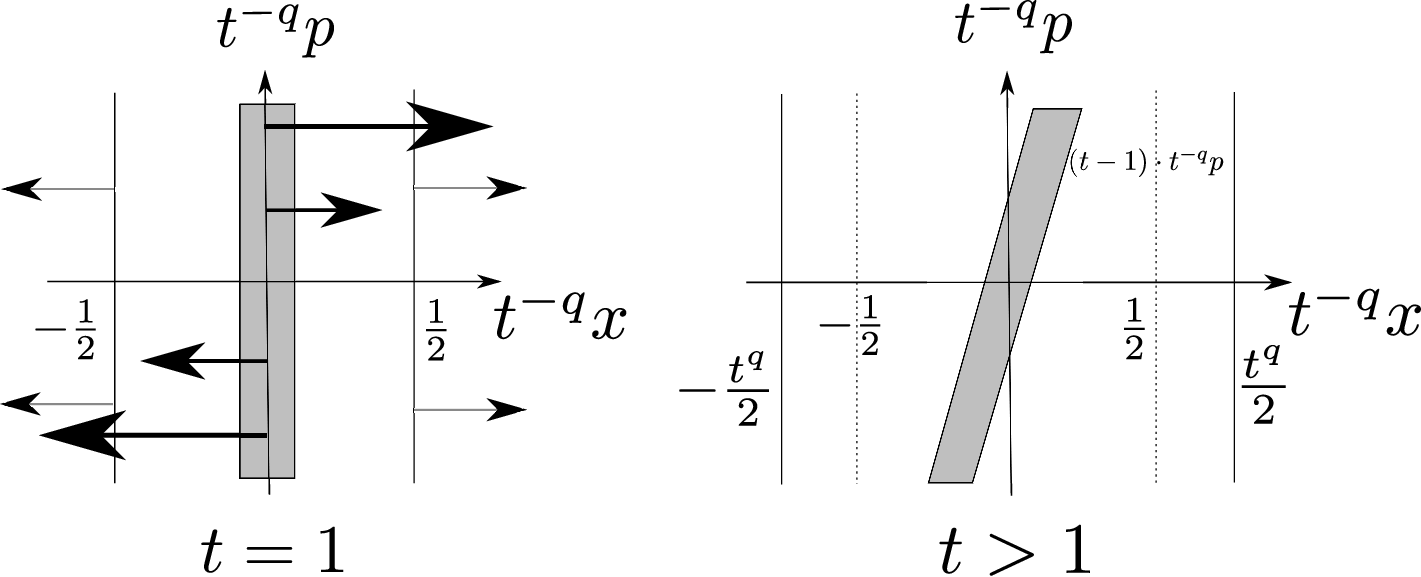}
	\caption{The phase mixing effect in the presence of expansion.  The support of a configuration initially localised around the origin may catch the expanding boundary if $0 \leq q \leq \frac{1}{2}$.}
	\label{fig:mixing2}
\end{figure}

\begin{remark}[Decay due to expansion and the case $q > \frac{1}{2}$]
	Theorem \ref{thm:main1} in particular implies that $\rho(t,x)$ decays at the rate $t^{-6q}$.  This $t^{-6q}$ decay rate is purely due to the expansion of the FLRW spacetime \eqref{eq:FLRWgeneral}.  It is the enhanced decay rate for for $\rho(t,x) - \overline{\rho}(t)$ of Theorem \ref{thm:main1} and Theorem \ref{thm:main2}, which is due to the phase mixing effect discussed above.  Though the mixing effect of Theorem \ref{thm:main1} is not present in the case $q > \frac{1}{2}$, the $t^{-6q}$ decay due to expansion remains present.
\end{remark}

\begin{remark}[The cases $q = \frac{1}{4}, \frac{1}{6}, \frac{1}{8},\ldots$]
	Theorem \ref{thm:main1} contains the assumption that the expansion rate $q$ is such that $\frac{1}{2q}$ is not an integer (i.\@e.\@ $q \neq \frac{1}{4}, \frac{1}{6}, \frac{1}{8},\ldots$).  This assumption is made for a technical reason in order to simplify the proof.  Indeed, the proof of Theorem \ref{thm:main1} relies on a collection of vector fields which are used to commute the Vlasov equation \eqref{eq:Vlasovgeneral} (see the discussion in Section \ref{subsec:overview} below).  The vector fields used are smooth (in fact real analytic) if $\frac{1}{2q}$ is not an integer, but not in $C^{\frac{1}{2q}+1}$ if $\frac{1}{2q}$ is an integer.  There is another collection of smooth vector fields (see Remark \ref{rmk:repgeneral}) which can be used, in particular in these remaining cases, but, in order to simply the proof, these vector fields and these exceptional values of $q$ are not considered here.  The vector fields used do, however, give a proof of a version of \eqref{eq:main2} in these exceptional cases (provided \eqref{eq:defofX} below is suitably defined) for $k \leq \frac{1}{2q} $.
\end{remark}

\begin{remark}[The relativistic free transport equation $q = 0$]
	The case $q=0$, known as the relativistic free transport equation, is also excluded from Theorem \ref{thm:main1}.  One, of course, still has a mixing effect and similar estimates \eqref{eq:main2}--\eqref{eq:main3} however.  See Theorem \ref{thm:relft} below.
\end{remark}

\begin{remark}[Loss in analytic case compared to non-relativistic equation] \label{rmk:analyticloss}
	In Section \ref{subsec:nonrel1} and Section \ref{subsec:nonrel2} below, equation \eqref{eq:Vlasovgeneral} will be compared to its non-relativistic counterpart.  In the case of analytic initial data $f_1$, there is a loss in the rates of Theorem \ref{thm:main1} and Theorem \ref{thm:main2} compared to these non-relativistic equations ($\exp(- \mu \, t^{\frac{1-2q}{10}})$ compared to $\exp(- \lambda (1-2q)^{-1} t^{1-2q})$ in Theorem \ref{thm:main1}, and $\exp(-\mu(\log (1+t))^{\frac{1}{58}})$ compared to $t^{-\lambda}$ in Theorem \ref{thm:main2}).  Though the rates \eqref{eq:main3} and \eqref{eq:mainradiation3} of Theorem \ref{thm:main1} and Theorem \ref{thm:main2} may not be sharp, it seems that there is indeed some loss compared to the non-relativistic case.  See the article \cite{You15} of Young where a related fact is discussed in the context of the relativistic Vlasov--Poisson system, compared to the classical Vlasov--Poisson system.
\end{remark}

\begin{remark}[Gevrey regularity] \label{rmk:Gevrey}
	Versions of Theorem \ref{thm:main1} and Theorem \ref{thm:main2} also hold for $f_1$ lying in a suitable \emph{Gevrey space}.  Indeed it also follows from the proof that, for $q$ as in Theorem \ref{thm:main1}, if $f_1 \in \mathcal{G}^s_q(\mathbb{T}^3\times \mathbb{R}^3)$ for some $s \geq 1$, then
	\[
		\sup_{x \in \mathbb{T}^3}\big\vert \rho(t,x) - \overline{\rho}(t) \big\vert
		\lesssim
		\frac{\Vert f_1 \Vert_{\mathcal{G}^s_q}}{t^{6q}}
		e^{- \mu \, t^{\frac{1-2q}{9+s}}}
		,
	\]
	for some $\mu >0$, where the Gevrey space $\mathcal{G}^s_q(\mathbb{T}^3\times \mathbb{R}^3)$ and associated norm are defined in Section \ref{subsec:functionspaces} below.  Similarly, for $q=\frac{1}{2}$, it follows from the proof of Theorem \ref{thm:main2} that, for $f_1 \in \mathcal{G}^s_{\log}(\mathbb{T}^3\times \mathbb{R}^3)$ for some $s \geq 1$, then
	\[
		\sup_{x \in \mathbb{T}^3}\big\vert \rho(t,x) - \overline{\rho}(t) \big\vert
		\lesssim
		\frac{\Vert f_1 \Vert_{\mathcal{G}^{s}_{\log}}}{t^3} e^{-\mu(\log (1+t))^{\frac{1}{57+s}}}
		.
	\]
\end{remark}

\begin{remark}[Estimates for derivatives]
	As will be apparent from the proof, the derivatives of solutions of \eqref{eq:Vlasovgeneral} also admit similar behaviour.  For example, in the case $0 < q<\frac{1}{2}$, solutions of \eqref{eq:Vlasovgeneral} satisfy, for any $k,l\geq 0$,
	\[
		\sup_{x \in \mathbb{T}^3}
		\sum_{\vert I \vert = l}
		\big\vert \partial_x^I \big( \rho(t,x) - \overline{\rho}(t) \big) \big\vert
		\lesssim
  		\frac{\Vert f_1 \Vert_{H^{k+l}}}{t^{6q + (k+l)(1-2q)}}.
	\]
\end{remark}

\begin{remark}[Weak convergence to spatial average] \label{rmk:weakconv}
	If $f$ is a suitably regular solution of equation \eqref{eq:Vlasovgeneral}, then the quantity
	\[
		\int_{\mathbb{R}^3} f(t,x,t^{-2q} p)dx = \overline{f_1}(p),
	\]
	is independent of $t$.  It follows from the proofs that, in the context of Theorem \ref{thm:main1} and Theorem \ref{thm:main2}, the rescaling of the solution $f(t,x,t^{-2q} p)$ of \eqref{eq:Vlasovgeneral} converges weakly to $\overline{f_1}(p)$,
	\[
		f(t,x,t^{-2q} p) \rightharpoonup \overline{f_1}(p)
		\quad
		\text{ as }
		\quad
		t \to \infty.
	\]
	See Remark \ref{rmk:weakconvagain} below.  In particular, the statements of Theorem \ref{thm:main1} and Theorem \ref{thm:main2} extend to corresponding statements for higher moments of $f$ in place of $\rho$.
\end{remark}

\subsubsection{Related works}

The phase mixing effect for the non-relativistic free transport equation \eqref{eq:nonrelft} was exhibited in nonlinear perturbations of homogeneous equilibria for the Vlasov--Poisson system by Mouhot--Villani \cite{MoVi}, following many previous works at the linearised level, for example \cite{Lan, Pen}.  Since then there have been a number of subsequent developments \cite{Bed20, BeMaMo, BCGIR, GrNgRo, GrNgRo2, IPWW}.  See the recent notes of Bedrossian \cite{Bed22} for further discussion.  Note also the works of Young \cite{You15}, \cite{You16} on the \emph{relativistic Vlasov--Poisson system}.  There have also been related recent works on kinetic equations around homogeneous equilibria on $\mathbb{R}_x^3 \times \mathbb{R}_p^3$ \cite{BBM18, BBM22, HNR21a, HNR21, HNR25, IPPW23, IPWW24, Ng23, Ng24}.  The phase mixing effect for \eqref{eq:nonrelft} has also been exhibited around shear flows for the incompressible Euler equations.  See for example \cite{Kel, BedMas, Zil}, along with many other recent works.

Related mixing effects have also been exhibited in the Newtonian gravitational setting, around isolated bodies at the linearised level \cite{HRSS, HaSc}, and also nonlinearly \cite{ChLu}.  Such effects have also been discussed in asymptotically flat settings in general relativity, for example by Rioseco--Sarbach \cite{RiSa20, RiSa24}.

A related phenomenon to that of Theorem \ref{thm:main2}, whereby each extra degree of differentiability assumed for initial data results in an extra logarithmic factor of decay in the resulting solution, occurs in general relativity for the Klein--Gordon equation (cf.\@ equation \eqref{eq:wave}) on Kerr--AdS black holes \cite{HoSm13, HoSm14}.  It has also been shown by Benomio \cite{Ben} that, subject to an assumed energy boundedness statement, a similar effect occurs for solutions of the wave equation on a class of \emph{black ring} spacetimes.

There has been study of the Euler equations (see Appendix \ref{subsec:FLRW}) on FLRW spacetimes of the form \eqref{eq:FLRWgeneral} \cite{FMOOW, FOOW, Spe}.  In particular, in \cite{FMOOW} it is suggested that shock waves should form in much of the very slowly expanding regime considered here.

There have been many works on the linear wave equation \eqref{eq:wave} of the spacetime \eqref{eq:FLRWgeneral}.  For the case where the spatial slices are flat copies of $\mathbb{R}^3$, and the combined effects of dispersion and expansion are relevant, see, for example, the recent articles of Nat{\'a}rio--Rossetti \cite{NaRo} and Haghshenas \cite{Hag} and references within.

Some future-global results for the coupled Einstein--Vlasov system (see Appendix \ref{subsec:FLRW}) in cosmological settings are known in certain symmetry classes --- see the recent \cite{Tay24} for an example in spherical symmetry, and \cite{AnReWe, Wea} for examples in $T^2$ symmetry --- and in the presence of a cosmological constant \cite{Rin13}.  FLRW solutions of the coupled Einstein--Vlasov--scalar field system have recently been shown to be stable in the past direction (towards the big bang singularity) \cite{FaUr}.  See also \cite{AnHeSh}.

Commuting vector field techniques, such as those of the present article (see Section \ref{subsec:overview} below), have recently been used in many works for kinetic and transport equations.  See, for example, \cite{BZZ, BiVe24, BiVe25, BiVeVe, Cha, ChLuNg, CoZe, FJS17, FJS21, LiTa, ScTa, Tay, WZZ20, Won}.

\subsection{Overview of the proof}
\label{subsec:overview}

In this section the main steps of the proofs of Theorem \ref{thm:main1} and Theorem \ref{thm:main2} are overviewed.  In Section \ref{subsec:nonrel1}, the non-relativistic analogue of equation \eqref{eq:Vlasovgeneral} is considered.  This non-relativistic equation serves as a toy problem, which illustrates the basic strategy of the proof of Theorem \ref{thm:main1} and Theorem \ref{thm:main2} in a much simpler context.  The $q=\frac{1}{2}$ case is discussed in Section \ref{subsec:nonrel2}.  Some relevant additional considerations arise for the relativistic free transport ($q=0$), which is discussed in Section \ref{subsec:relft}.  The main key steps in the proofs of Theorem \ref{thm:main1} and Theorem \ref{thm:main2} are discussed in Section \ref{subsec:intoproofmain}.

\subsubsection{A non-relativistic toy problem for $0\leq q < \frac{1}{2}$}
\label{subsec:nonrel1}

When $q=0$, the spacetime \eqref{eq:FLRWgeneral} describes a static torus, locally isometric to Minkowski space, and equation \eqref{eq:Vlasovgeneral} reduces to the well known \emph{relativistic free transport} equation.  Even more well known is the non-relativistic free transport equation,
\begin{equation} \label{eq:nonrelft}
	\partial_t f + p^i \partial_{x^i} f = 0,
\end{equation}
which is obtained by replacing $p^0$ with $1$ in the $q=0$ case of equation \eqref{eq:Vlasovgeneral}.  The phase mixing phenomena for equation \eqref{eq:nonrelft} has been studied extensively (see \cite{Vil} for a nice discussion).  For all $0\leq q \leq \frac{1}{2}$, the equation 
\begin{equation} \label{eq:Vlasovtoy}
	\partial_t f
	+
	p^i \partial_{x^i} f
	-
	\frac{2q}{t} p^{i} \partial_{p^i} f
	=
	0.
\end{equation}
obtained by replacing $p^0$ with $1$ in equation \eqref{eq:Vlasovgeneral}, is the non-relativistic analogue of equation \eqref{eq:Vlasovgeneral} and provides a good toy model to present the phenomena exhibited in Theorem \ref{thm:main1} and Theorem \ref{thm:main2} in a much simpler context.\footnote{The fact that equation \eqref{eq:Vlasovtoy} displays some of the phenomena seen in Theorem \ref{thm:main1} and Theorem \ref{thm:main2} is indicated by the fact that, as will be seen later, for all $0\leq q \leq \frac{1}{2}$, $p^0 \sim 1$ as $t\to \infty$ along the characteristics of equation \eqref{eq:Vlasovgeneral}.  Furthermore, equation \eqref{eq:Vlasovtoy} can independently be obtained by considering the nonlinear ordinary differential equations for the characteristics of equation \eqref{eq:Vlasovgeneral} and linearising around solutions with vanishing velocity.}

\begin{remark}[Non-relativistic limit]
	Equation \eqref{eq:Vlasovtoy} arises as a non-relativistic analogue of equation \eqref{eq:Vlasovgeneral}.  Indeed, in equation \eqref{eq:FLRWgeneral} the speed of light has been normalised to unity.  Restoring the speed of light $c>0$, the FLRW spacetimes take the form
	\[
		\mathcal{M} = (0,\infty) \times \mathbb{T}^3,
		\qquad
		g = -c^2dt^2 + t^{2q} \big( (dx^1)^2 + (dx^2)^2 + (dx^3)^2 \big).
	\]
	The appropriately normalised \emph{mass shell} takes the form
	\[
		\mathcal{P} = \{ (t,x,p) \in T \mathcal{M} \mid g(p,p) = -c^2\},
	\]
	and thus the Vlasov equation, for functions $f\colon \mathcal{P} \to [0,\infty)$, takes the form
	\[
		\partial_t f
		+
		\frac{p^i}{\sqrt{1+ c^{-2} t^{2q} \vert p \vert^2}} \partial_{x^i} f
		-
		\frac{2q}{t} p^{i} \partial_{p^i} f
		=
		0.
	\]
	Equation \eqref{eq:Vlasovtoy} thus arises in the formal limit $c\to \infty$.
\end{remark}

The analogue of Theorem \ref{thm:main1} for equation \eqref{eq:Vlasovtoy} takes the following form.

\begin{theorem}[Phase mixing for the non-relativistic equation for $0 \leq q<\frac{1}{2}$] \label{thm:toy1}
	Consider some $0 \leq q<\frac{1}{2}$, some $k \geq 2$, and some $f_1 \in H_{\circ}^k(\mathbb{T}^3 \times \mathbb{R}^3)$.  Let $f$ be the unique solution of equation \eqref{eq:Vlasovtoy} on $[1,\infty) \times \mathbb{T}^3\times \mathbb{R}^3$ such that $f(1,x,p) = f_1(x,p)$.  The spatial density of $f$ satisfies, for all $t \geq 1$,
	\begin{equation} \label{eq:maintoy2}
		\sup_{x \in \mathbb{T}^3}\big\vert \rho(t,x) - \overline{\rho}(t) \big\vert
		\lesssim
  		\frac{\Vert f_1 \Vert_{H_{\circ}^k}}{t^{6q + k(1-2q)}},
		\qquad
		\text{where}
		\quad
		\overline{\rho}(t)
		=
		\frac{1}{t^{6q}} \int_{\mathbb{T}^3} \int_{\mathbb{R}^3} f_1(x,p) dp dx.
	\end{equation}
	If $f_1$ lies in the analytic space $f_1 \in H_{\circ}^{\omega}(\mathbb{T}^3 \times \mathbb{R}^3)$ then the spatial density satisfies, for all $t \geq 1$,
	\begin{equation} \label{eq:maintoy3}
		\sup_{x \in \mathbb{T}^3}\big\vert \rho(t,x) - \overline{\rho}(t) \big\vert
		\lesssim
		\frac{\Vert f_1 \Vert_{H_{\circ}^{\omega}}}{t^{6q}}
		e^{- \lambda (1-2q)^{-1} t^{1-2q}}
		,
	\end{equation}
	for some $\lambda>0$.
\end{theorem}

The spaces $H_{\circ}^k(\mathbb{T}^3\times \mathbb{R}^3)$ are $p$-weighted versions of the standard Sobolev spaces, and $H_{\circ}^{\omega}(\mathbb{T}^3\times \mathbb{R}^3)$ is an associated space of analytic functions.  The constant $\lambda$ should be viewed as a radius of convergence for the analytic norm.  See Section \ref{subsec:functionspaces} for definitions.

\begin{remark}[Representation formula]
	Equation \eqref{eq:Vlasovtoy} admits a representation formula which, for $q\neq 1/2$, takes the form
	\[
		f(t,x,p) = f_1 \Big( x-\frac{1}{1-2q}(t^{2q} - t) p, t^{2q} p \Big).
	\]
	For $q=1/2$, the representation formula takes the form
	\[
		f(t,x,p) = f_1(x - t \log t \, p, tp).
	\]
	Though it can be helpful for gaining intuition, this representation formula does not feature in the proof of Theorem \ref{thm:toy1}, given below.
\end{remark}

The main ingredients in the proof of Theorem \ref{thm:toy1} are:
\begin{itemize}
	\item
		\textbf{Conservation laws:}
		For any solution $f$ of equation \eqref{eq:Vlasovtoy}, suitably decaying in $p$, and any $s \geq 0$, the quantity
		\begin{equation} \label{eq:introconservation}
			t^{6q} \int_{\mathbb{T}^3} \int_{\mathbb{R}^3} \vert t^{2q} p \vert^{s} \vert f(t,x,p) \vert^2 dp dx,
		\end{equation}
		is independent of $t$.  (See Proposition \ref{prop:conservationL2} below, for the analogue of this fact for equation \eqref{eq:Vlasovgeneral}.)
	\item
		\textbf{Commutation vector fields:}
		For all $0 \leq q < \frac{1}{2}$, the vector fields
		\begin{equation} \label{eq:toyvectors}
			L_i =\frac{t^{1-2q}}{1-2q} \partial_{x^i} + t^{-2q} \partial_{p^i}, \qquad i=1,2,3,
		\end{equation}
		commute with equation \eqref{eq:Vlasovtoy}.  Thus, if $f$ solves \eqref{eq:Vlasovtoy}, then $L^If$ also solves \eqref{eq:Vlasovtoy} for each multi-index $I$.
	\item
		\textbf{Sobolev inequality:}
		There is a constant $C>0$ such that, for any function $\rho$ and any $k \geq 2$,
		\begin{equation} \label{eq:Sobolevintro}
			\sup_{x\in \mathbb{T}^3} \vert \rho(t,x) - \overline{\rho}(t) \vert
			\leq
			\frac{C}{\sqrt{k}}
			\sum_{\vert I \vert = k}
			\Vert \partial_x^I \rho(t,\cdot) \Vert_{L^2(\mathbb{T}^3)}.
		\end{equation}
		See Proposition \ref{prop:Sobolev}.
	\item
		\textbf{Derivative relation:}
		For each $i=1,2,3$, spatial derivatives of $\rho$ are related to the vector fields \eqref{eq:toyvectors} applied to $f$ via
		\begin{equation} \label{eq:introderiv}
			\partial_{x^i} \rho(t,x) = \frac{1-2q}{t^{1-2q}} \int L_i f(t,x,p) dp.
		\end{equation}
\end{itemize}
Suppose thus that $f$ is a solution of equation \eqref{eq:Vlasovtoy} on $[1,\infty) \times \mathbb{T}^3\times \mathbb{R}^3$ such that $f(1,\cdot, \cdot) = f_1 \in H^k_{\circ}(\mathbb{T}^3 \times \mathbb{R}^3)$, for some $k \geq 2$.  It follows from the Sobolev inequality \eqref{eq:Sobolevintro} and the derivative relation \eqref{eq:introderiv} that
\begin{equation} \label{eq:introtoySobderiv}
	\sup_{x\in \mathbb{T}^3} \vert \rho(t,x) - \overline{\rho}(t) \vert
	\lesssim
	\frac{1}{\sqrt{k}}
	\Big( \frac{1-2q}{t^{1-2q}} \Big)^k
	\sum_{\vert I \vert = k}
	\Big\Vert \int L^I f(t,\cdot,p) dp \Big\Vert_{L^2(\mathbb{T}^3)}.
\end{equation}
For any multi-index $I$ it follows from a general functional inequality (see Proposition \ref{prop:interp}) that
\begin{equation*}
	\Big\Vert \int L^I f(t,\cdot,p) dp \Big\Vert_{L^2(\mathbb{T}^3)}
	\lesssim
	\left( \int_{\mathbb{T}^3} \int_{\mathbb{R}^3} \vert p \vert^2 \vert L^I f(t,x,p) \vert^2 dp dx \right)^{\frac{1}{4}}
	\left( \int_{\mathbb{T}^3} \int_{\mathbb{R}^3} \vert p \vert^4 \vert L^I f(t,x,p) \vert^2 dp dx \right)^{\frac{1}{4}}.
\end{equation*}
Now since the vector fields $L_i$ commute with equation \eqref{eq:Vlasovtoy}, the conservation laws \eqref{eq:introconservation} hold with $L^I f$ in place of $f$ so that
\begin{equation} \label{eq:introtoyfinite}
	\sup_{x\in \mathbb{T}^3} \vert \rho(t,x) - \overline{\rho}(t) \vert
	\lesssim
	\frac{(1-2q)^k}{t^{6q+k(1-2q)}\sqrt{k}}
	\sum_{\vert I \vert = k}
	\left( \int_{\mathbb{T}^3} \int_{\mathbb{R}^3} (\vert p \vert^2 + \vert p \vert^4)\vert L^I f_1(x,p) \vert^2 dp dx \right)^{\frac{1}{2}},
\end{equation}
yielding \eqref{eq:maintoy2}.

If, moreover, $f_1$ lies in the analytic space $H_{\circ}^{\omega}$ (see again the definition in Section \ref{subsec:notation}) then \eqref{eq:introtoyfinite} holds for all $k \geq 2$ and thus, by definition of the analytic norm, there is $\lambda >0$ such that
\begin{equation} \label{eq:introtoyanalyticdecay}
	\sup_{x \in \mathbb{T}^3}\big\vert \rho(t,x) - \overline{\rho}(t) \big\vert
	\lesssim
  	\frac{(1-2q)^k k! }{(\lambda t^{(1-2q)})^k\sqrt{k}}
	\frac{\Vert f_1 \Vert_{H_{\circ}^{\omega}}}{t^{6q}},
	\qquad
	\text{for all }
	k \geq 2.
\end{equation}
Thus, for all times $t \geq 1$, \eqref{eq:maintoy3} follows from setting $k= \frac{\lambda t^{1-2q}}{1-2q}$ in \eqref{eq:introtoyanalyticdecay} (or the floor of this quantity, if it is not an integer) and recalling that $n! \, e^n \lesssim \sqrt{n} \, n^n$ for all $n \in \mathbb{N}$ (see Proposition \ref{prop:Stirling} below).

\subsubsection{A non-relativistic toy problem for the degenerate $q = \frac{1}{2}$ case}
\label{subsec:nonrel2}

When $q=\frac{1}{2}$, the non-relativistic problem \eqref{eq:Vlasovtoy} takes the form
\begin{equation} \label{eq:Vlasovtoy12}
	\partial_t f
	+
	p^i \partial_{x^i} f
	-
	\frac{1}{t} p^{i} \partial_{p^i} f
	=
	0.
\end{equation}
The analogue of Theorem \ref{thm:main2} for \eqref{eq:Vlasovtoy12} takes the following form.  (For convenience, for this non-relativistic problem, it is helpful to consider initial data at some $t_0>1$ rather than at $t=1$.)

\begin{theorem}[Degenerate phase mixing for the non-relativistic equation with $q = \frac{1}{2}$] \label{thm:toy2}
	Consider some $t_0 >1$, some $k \geq 2$, and some $f_0 \in H_{\circ}^k(\mathbb{T}^3 \times \mathbb{R}^3)$.  Let $f$ be the unique solution of equation \eqref{eq:Vlasovtoy12} on $[t_0,\infty) \times \mathbb{T}^3\times \mathbb{R}^3$ such that $f(t_0,x,p) = f_0(x,p)$.  The spatial density of $f$ satisfies, for all $t \geq t_0$,
	\begin{equation} \label{eq:maintoy12}
		\sup_{x \in \mathbb{T}^3}\big\vert \rho(t,x) - \overline{\rho}(t) \big\vert
		\lesssim
  		\frac{\Vert f_0 \Vert_{H_{\circ}^k}}{t^3 (\log t)^k},
		\qquad
		\text{where}
		\quad
		\overline{\rho}(t)
		=
		\frac{t_0^3}{t^3} \int_{\mathbb{T}^3} \int_{\mathbb{R}^3} f_0(x,p) dp dx.
	\end{equation}
	If $f_0$ lies in the analytic space $f_0 \in H^{\omega}_{\circ}(\mathbb{T}^3 \times \mathbb{R}^3)$ then the spatial density satisfies, for all $t \geq t_0$,
	\begin{equation} \label{eq:maintoy122}
		\sup_{x \in \mathbb{T}^3}\big\vert \rho(t,x) - \overline{\rho}(t) \big\vert
		\lesssim
		\frac{\Vert f_0 \Vert_{H_{\circ}^{\omega}}}{t^{3+\lambda}}
		,
	\end{equation}
	where $\lambda = \lambda(f_0)$ (see Section \ref{subsec:functionspaces} below).
\end{theorem}

When $q = \frac{1}{2}$, the treatment of the non-relativistic problem \eqref{eq:Vlasovtoy} is similar to the $0\leq q < \frac{1}{2}$ case discussed above.  The main difference is that the commuting vector fields now take the form
\[
	L_k = \log t \, \partial_{x^k} + \frac{1}{t} \partial_{p^k}, \qquad k = 1,2,3,
\]
and these vector fields are used in place of \eqref{eq:toyvectors}.  The analogue of \eqref{eq:introtoyfinite} then takes the form
\begin{equation} \label{eq:introtoyfinite12}
	\sup_{x\in \mathbb{T}^3} \vert \rho(t,x) - \overline{\rho}(t) \vert
	\lesssim
	\frac{1}{t^3 (\log t)^k\sqrt{k}}
	\sum_{\vert I \vert = k}
	\left( \int_{\mathbb{T}^3} \int_{\mathbb{R}^3} (\vert p \vert^2 + \vert p \vert^4)\vert L^I f(t,x,p) \vert^2 dp dx \right)^{\frac{1}{2}},
\end{equation}
yielding \eqref{eq:maintoy12}.

Suppose now that $f_0$ lies in the analytic space $f_0 \in H_{\circ}^{\omega}(\mathbb{T}^3 \times \mathbb{R}^3)$.  Then, by \eqref{eq:introtoyfinite12}, there exists $\lambda>0$ such that
\begin{equation} \label{eq:introtoyanalyticdecay12}
	\sup_{x \in \mathbb{T}^3}\big\vert \rho(t,x) - \overline{\rho}(t) \big\vert
	\lesssim
  	\frac{k! }{(\lambda \log t)^k\sqrt{k}}
	\frac{\Vert f_0 \Vert_{H_{\circ}^{\omega}}}{t^3},
	\qquad
	\text{for all }
	k \geq 2.
\end{equation}
The inequality \eqref{eq:maintoy122} then follows by setting $k =\lfloor \lambda \log t \rfloor$ in \eqref{eq:introtoyanalyticdecay12} and recalling again that $n! \, e^n \lesssim \sqrt{n} \, n^n$ for all $n \in \mathbb{N}$.

\begin{remark}[The non-relativistic problem \eqref{eq:Vlasovtoy} as $t\to 0$]
	For $q > \frac{1}{2}$, there is no improvement in the behaviour of $\rho(t,x) - \overline{\rho}(t)$, for solutions of \eqref{eq:Vlasovtoy}, as $t \to \infty$.  There is an improvement in the behaviour of $\rho(t,x) - \overline{\rho}(t)$ as $t\to 0$ when $q > \frac{1}{2}$, however.  Contrast with the case $q < \frac{1}{2}$ where there is no improvement as $t\to 0$.  When $q=\frac{1}{2}$ there is a small improvement in $\rho - \overline{\rho}$ both as $t\to \infty$ and as $t \to 0$.  Note however that equation \eqref{eq:Vlasovtoy} only seems to be a reasonable model for \eqref{eq:Vlasovgeneral} as $t \to \infty$, and not as $t \to 0$.
\end{remark}

\subsubsection{The relativistic free transport equation}
\label{subsec:relft}

The main difference between the non-relativistic problems discussed in Section \ref{subsec:nonrel1} and Section \ref{subsec:nonrel2} and the problem \eqref{eq:Vlasovgeneral} is the fact that the vector fields which commute with the equation \eqref{eq:Vlasovgeneral} are more complicated than \eqref{eq:toyvectors}.  This fact means that the derivative relation \eqref{eq:introderiv} is no longer so simple.

In this section, in order to illustrate some of these new difficulties in a considerably simpler setting, the relativistic free transport equation is considered (which is obtained by setting $q=0$ in \eqref{eq:Vlasovgeneral}).  In order to make the notation easier to navigate, we restrict here to the case of one spatial dimension (though the proof, of course, generalises to higher dimensions).  Thus, for $f\colon \mathbb{R} \times \mathbb{T} \times \mathbb{R} \to [0,\infty)$, the equation takes the form
\begin{equation} \label{eq:relft}
	\partial_tf + \frac{p}{\sqrt{1+p^2}} \partial_x f = 0,
	\qquad
	\rho(t,x) = \int_{\mathbb{R}} f(t,x,p) dp.
\end{equation}

\begin{theorem}[Phase mixing for the relativistic free transport equation] \label{thm:relft}
	Consider some $k \geq 2$ and $f_1 \in H_{\mathrm{exp}}^k(\mathbb{T} \times \mathbb{R})$.  Let $f$ be the unique solution of equation \eqref{eq:relft} on $[1,\infty) \times \mathbb{T}\times \mathbb{R}$ such that $f(1,x,p) = f_1(x,p)$.  The spatial density satisfies, for all $t \geq 1$,
	\begin{equation} \label{eq:relft1}
		\sup_{x \in \mathbb{T}}\big\vert \rho(t,x) - \overline{\rho} \big\vert
		\lesssim
  		\frac{\Vert f_1 \Vert_{H^k_{\mathrm{exp}}}}{t^k},
		\qquad
		\text{where}
		\quad
		\overline{\rho}
		=
		\int_{\mathbb{T}} \int_{\mathbb{R}} f_1(x,p) dp dx.
	\end{equation}
	If $f_1$ lies in the analytic space $f_1 \in H_{\mathrm{exp}}^{\omega}(\mathbb{T} \times \mathbb{R})$ then the spatial density satisfies, for all $t \geq 1$,
	\begin{equation} \label{eq:relft2}
		\sup_{x \in \mathbb{T}}\big\vert \rho(t,x) - \overline{\rho}(t) \big\vert
		\lesssim
		\Vert f_1 \Vert_{H^{\omega}_{\mathrm{exp}}}
		e^{- \mu \, t^{\frac{1}{11}}}
		,
	\end{equation}
	for some $\mu>0$.
\end{theorem}

The spaces $H_{\mathrm{exp}}^k$ and $H_{\mathrm{exp}}^{\omega}$ are exponentially weighted versions of the standard Sobolev spaces.  See Section \ref{subsec:functionspaces} for definitions.

The main ingredients in the proof of Theorem \ref{thm:relft} are the conservations laws \eqref{eq:introconservation}, the Sobolev inequality \eqref{eq:Sobolevintro}, together with the following analogues of \eqref{eq:toyvectors} and \eqref{eq:introderiv}:
\begin{itemize}
	\item
		\textbf{Commutation vector fields:}
		The vector field
		\begin{equation} \label{eq:relftvectors}
			M = t \partial_{x} + (1+p^2)^{\frac{3}{2}} \partial_{p},
		\end{equation}
		commutes with equation \eqref{eq:relft}.
	\item
		\textbf{Derivative relation:}
		Spatial derivatives of $\rho$ are related to the vector fields \eqref{eq:relftvectors} applied to $f$ via
		\begin{equation} \label{eq:relftderiv}
			\partial_{x} \rho(t,x) = \frac{1}{t} \int M f(t,x,p) + 3p(1+p^2)^{\frac{1}{2}} f(t,x,p) dp.
		\end{equation}
\end{itemize}

The extra, zeroth order, term in \eqref{eq:relftderiv} (compared to \eqref{eq:introderiv}) arises from the $p$ dependent factor multiplying the $\partial_p$ derivative in \eqref{eq:relftvectors} (which has been integrated by parts).  Applying repeatedly, it follows that, for any $i \geq 1$,
\[
	\partial_{x}^i \rho(t,x) = \frac{1}{t^i} \int (M + 3p(1+p^2)^{\frac{1}{2}})^i f(t,x,p) dp.
\]
Thus, in order to repeat the step \eqref{eq:introtoySobderiv} above, one is lead to considering combinatorial properties of the operator $(M + 3p(1+p^2)^{\frac{1}{2}})^i$.  The main such property is the following Binomial Theorem-type result, which relates the operator to combinations of the commutation vector fields.

\begin{proposition}[A binomial theorem for the commutation vector fields] \label{prop:relcomb}
	For each $i \geq 1$, there are integers $C^i_{jkm} \in \mathbb{Z}$, for $j,k,m\in \mathbb{N}_0$, such that
	\begin{equation} \label{eq:relftcomb}
		\big(
		M
		+
		3p(1+p^2)^{\frac{1}{2}}
		\big)^i
		=
		\sum_{j,k,m} C^i_{jkm} p^j (1+p^2)^{\frac{m}{2}} M^k.
	\end{equation}
	Moreover
	\begin{itemize}
		\item
			Each $C^i_{jkm}$ satisfies
			\begin{equation} \label{eq:relftcomb2}
				0
				\leq
				C^i_{jkm}
				\leq
				(4 i)!
				\qquad
				\text{for all}
				\quad
				j,k,m\in \mathbb{N}_0.
			\end{equation}
		\item
			The non-vanishing $C^i_{jkm}$ satisfy
			\begin{equation} \label{eq:relftcomb3}
				C^i_{jkm} \neq 0
				\qquad
				\Rightarrow
				\quad
				0 \leq j \leq i,
				\quad
				0 \leq k \leq i,
				\quad
				0 \leq m \leq 3i.
			\end{equation}
		\item
			The non-vanishing $C^i_{jkm}$ moreover satisfy
			\begin{equation} \label{eq:relftcomb4}
				C^i_{jkm} \neq 0
				\qquad
				\Rightarrow
				\qquad
				j+m+2k \leq 2i.
			\end{equation}
	\end{itemize}
\end{proposition}

The fact that the non-vanishing terms in \eqref{eq:relftcomb} satisfy the property \eqref{eq:relftcomb4} will be used to track the $p$ weight of each term in the summation \eqref{eq:relftcomb}.  In order to treat the case of analytic $f_1$ in Theorem \ref{thm:relft}, it is important to characterise the behaviour of the constants as the number of derivatives $i \to \infty$.  The property \eqref{eq:relftcomb2} is used to estimate the constant for each such term, and the property \eqref{eq:relftcomb3} is used to estimate the number of terms.

\begin{proof}[Proof of Proposition \ref{prop:relcomb}]
	Clearly \eqref{eq:relftcomb} holds for $i=1$ with
	\[
		C^1_{010} = 1,
		\qquad
		C^1_{101} = 3,
		\qquad
		C^1_{jkm} = 0 \text{ otherwise}.
	\]
	Suppose \eqref{eq:relftcomb} holds for some $i$.  Then
	\begin{multline*}
		\big(
		M
		+
		3p(1+p^2)^{\frac{1}{2}}
		\big)
		\Big(
		\sum_{j,k,m} C^i_{jkm} p^j (1+p^2)^{\frac{m}{2}} M^k
		\Big)
		=
		\sum_{j,k,m} C^i_{jkm} 
		\Big(
		j p^{j-1} (1+p^2)^{\frac{m+3}{2}} M^k
		\\
		+
		m p^{j+1} (1+p^2)^{\frac{m+1}{2}} M^k
		+
		p^j (1+p^2)^{\frac{m}{2}} M^{k+1}
		+
		3 p^{j+1} (1+p^2)^{\frac{m+1}{2}} M^k
		\Big),
	\end{multline*}
	and thus \eqref{eq:relftcomb} holds for $i+1$ with
	\begin{align} \label{eq:introrelftinduc}
		C^{i+1}_{jkm}
		=
		(j+1) C^i_{j+1,k,m-3}
		+
		(m+2) C^i_{j-1,k,m-1}
		+
		C^i_{j,k-1,m},
	\end{align}
	where $C^i_{jkm}:=0$ if $j<0$, $k<0$ or $m<0$.  Thus, by induction, \eqref{eq:relftcomb} holds for all $i$.
	
	Consider now \eqref{eq:relftcomb3}.  Clearly \eqref{eq:relftcomb3} holds for $i=1$.  Suppose now \eqref{eq:relftcomb3} holds for some $i$.  If $C^{i+1}_{jkm} \neq 0$ then it must be the case that at least one of the terms on the right hand side of \eqref{eq:introrelftinduc} is non-vanishing.  Thus $j-1\leq i$, $k-1 \leq i$ and $m-3 \leq i$, i.\@e.\@ \eqref{eq:relftcomb3} holds for $i+1$, and thus for all $i$ by induction.
	
	Note now that the property \eqref{eq:relftcomb2} holds for $i=1$.  If \eqref{eq:relftcomb3} holds for some $i \geq 1$ then, by \eqref{eq:introrelftinduc} and the property \eqref{eq:relftcomb3},
	\[
		0
		\leq
		C^{i+1}_{jkm}
		\leq
		i (4 i)! + (3 i + 3) (4 i)! + (4 i)!
		=
		(4(i+1))!,
	\]
	and so, by induction, \eqref{eq:relftcomb2} holds for all $i$.
	
	Consider finally \eqref{eq:relftcomb4}.  Clearly \eqref{eq:relftcomb4} holds for $i=1$.  Suppose \eqref{eq:relftcomb4} holds for some $i$.  If $C^i_{jkm} \neq 0$ then at least one of the terms on the right hand side of \eqref{eq:introrelftinduc} must be non-vanishing.  If the first term is non-vanishing, then $j+2k+m-2\leq 2i$.  Similarly if the second or third terms are non-vanishing.  Thus \eqref{eq:relftcomb4} holds for $i+1$ and, by induction, for all $i$.
\end{proof}

Proposition \ref{prop:relcomb} --- in particular the fact \eqref{eq:relftcomb}, together with the properties \eqref{eq:relftcomb2} and \eqref{eq:relftcomb4} --- implies that, for each $i \geq 1$,
\begin{equation} \label{eq:relftbinominequ}
	\big\vert (M + 3p(1+p^2)^{\frac{1}{2}})^i f(t,x,p) \big\vert
	\leq
	3 (i+1)^3 (4i)!
	\sum_{k=0}^i
	(1+\vert p\vert)^{2(i-k)}
	\vert M^k f(t,x,p) \vert,
\end{equation}
where the property \eqref{eq:relftcomb3} is used to ensure that,
\[
	\# \{ (j,k,m) \mid C^i_{jkm} \neq 0 \}
	\leq
	3 (i+1)^3.
\]
The Sobolev inequality \eqref{eq:Sobolevintro} then implies that
\[
	\sup_{x\in \mathbb{T}} \vert \rho(t,x) - \overline{\rho}(t) \vert
	\lesssim
	\frac{1}{\sqrt{k}}
	\frac{k^3 (4k)!}{t^k}
	\sum_{i = 0}^k
	\Big\Vert \int (1+\vert p\vert)^{2(k-i)} \vert M^i f(t,\cdot,p) \vert dp \Big\Vert_{L^2(\mathbb{T})}.
\]
Now since the vector field $M$ commutes with the equation, the conservation laws \eqref{eq:introconservation} hold with $M^i f$ in place of $f$ so that
\begin{multline} \label{eq:introrelftfinite}
	\sup_{x\in \mathbb{T}} \vert \rho(t,x) - \overline{\rho}(t) \vert
	\lesssim
	\frac{k^3 (4k)!}{t^k \sqrt{k}}
	\sum_{i_1+i_2 \leq k}
	\left(
	\int_{\mathbb{T}} \int_{\mathbb{R}}
	(1+\vert p\vert^{2(k-i_1-i_2)+6i_2+2}) \vert \partial_x^{i_1} \partial_p^{i_2} f_1(x,p) \vert^2
	dp dx
	\right)^{\frac{1}{2}}
	\\
	\lesssim
	\frac{k^3 (4k)!(6k+2)!}{t^k \sqrt{k}}
	\Vert f_1 \Vert_{H^k_{\mathrm{exp}}},
\end{multline}
where the fact that $\vert M^if_1 \vert \leq \sum_{i_1+i_2 \leq i} (1+\vert p\vert^{3i_2}) \vert \partial_x^{i_1} \partial_p^{i_2} f_1\vert$ has been used (along with the property \eqref{eq:expweightconv} of the norm $\Vert f_1 \Vert_{H^k_{\mathrm{exp}}}$), yielding \eqref{eq:relft1}.

If, moreover, $f_1$ lies in the analytic space $H^{\omega}_{\mathrm{exp}}$ then \eqref{eq:introrelftfinite} holds for all $k \geq 2$ and thus, by definition of the analytic norm, there is $\lambda >0$ such that
\begin{equation} \label{eq:introrelftanalyticdecay}
	\sup_{x \in \mathbb{T}}\big\vert \rho(t,x) - \overline{\rho}(t) \big\vert
	\lesssim
  	\frac{k^7 (4k)! (6k)! k! }{(\lambda t)^k\sqrt{k}}
	\Vert f_1 \Vert_{H^{\omega}_{\mathrm{exp}}}
	\lesssim
  	\frac{2^k(11k)!}{(\lambda^{\frac{1}{11}} t^{\frac{1}{11}})^{11k}\sqrt{11k}}
	\Vert f_1 \Vert_{H^{\omega}_{\mathrm{exp}}},
	\qquad
	\text{for all }
	k \geq 2,
\end{equation}
where Proposition \ref{prop:factorials} below has been used in the final inequality, along with the fact that $k^7 \lesssim 2^k$ for all $k$.  Thus, for all times $t \geq 1$, \eqref{eq:relft2} follows from setting $k= \frac{1}{11}\mu t^{\frac{1}{11}}$ in \eqref{eq:introrelftanalyticdecay} (or the floor of this quantity, if it is not an integer), with $\mu = \lambda^{\frac{1}{11}} 2^{-\frac{1}{11}}$, and recalling that $n! \, e^n \lesssim \sqrt{n} \, n^n$ for all $n \in \mathbb{N}$ (see Proposition \ref{prop:Stirling} below).

\subsubsection{The proof of Theorem \ref{thm:main1} and Theorem \ref{thm:main2}}
\label{subsec:intoproofmain}

The main difference between the discussion of the toy problem \eqref{eq:Vlasovtoy} above, or of the relativistic free transport equation \eqref{eq:relft}, with the proof of Theorem \ref{thm:main1} and Theorem \ref{thm:main2} is the fact the commuting vector fields are considerably more complicated in the latter cases.  Indeed, the commuting vector fields used in the proof of Theorem \ref{thm:main1} take the form
\begin{equation} \label{eq:introvf1}
	L_k = {A_k}^i(t,p) \partial_{x^i} + \frac{1}{t^{2q}} \partial_{p^k},
	\qquad
	k=1,2,3,
\end{equation}
where
\begin{equation} \label{eq:introvf2}
	{A_k}^i(t,p)
	=
	G_q \big( t \, \vert t^{2q} p \vert^{-\frac{1}{q}} \big)
	\vert t^{2q} p \vert^{\frac{1-2q}{q}}
	\bigg(
	\delta_k^i
	+
	\frac{1-2q}{q}
	\frac{t^{4q} p^i p^k}{\vert t^{2q} p \vert^2}
	\bigg)
	-
	\frac{1}{q}
	\frac{t^{1-q}}{\sqrt{t^{2q} + \vert t^{2q} p \vert^2}}
	\frac{t^{4q} p^i p^k}{\vert t^{2q} p \vert^2},
\end{equation}
and $G_q\colon [0,\infty) \to \mathbb{R}$ is defined by
\begin{equation} \label{eq:introvf3}
	G_q(s)
	=
	\int_0^s \frac{1}{\sqrt{\tilde{s}^{4q}+\tilde{s}^{2q}}} d \tilde{s} + X_q.
\end{equation}
Here $X_q$ is a constant for each $q$ (see \eqref{eq:defofX} below).  Note that the integral \eqref{eq:introvf3} cannot, in general, be explicitly evaluated, and hence also the vector fields \eqref{eq:introvf1} are not explicit in general.

Suppose that $t \geq 1$ is such that the matrix ${A_k}^i(t,p)$ is invertible for all $p$ (though it is not clear, a priori, whether there is any such $t$).  For such $t$, the analogue of the derivative relation \eqref{eq:introderiv} is
\[
	\partial_{x^i} \rho(t,x)	
	=
	\frac{1}{t^{1-2q}}
	\int_{\mathbb{R}^3}
	\Big(
	M_i
	+
	\eta(t,p)_i
	\Big)
	f(t,x,p)
	dp,
\]
where
\[
	M_i = t^{1-2q} {(A^{-1})_i}^j(t,p) L_j
	\qquad
	\eta(t,p)_i = t^{1-4q} \partial_{p^j} {(A^{-1})_i}^j(t,p),
\]
Iterating gives, for any $k \geq 0$ and any multi-index $I$ with $\vert I \vert = k$,
\[
	\partial_{x}^I \rho(t,x)
	=
	\frac{1}{t^{k(1-2q)}}
	\int_{\mathbb{R}^3}
	\big(
	M
	+
	\eta(t,p)
	\big)^I
	f(t,x,p)
	dp.
\]
By analogy with the discussion for the non-relativistic problem discussed in Section \ref{subsec:nonrel1}, one would hope to relate $\big( M + \eta(t,p) \big)^I$ to $L^I$, so that the conservation laws \eqref{eq:introconservation} can be applied (which remain valid for equation \eqref{eq:Vlasovgeneral}).\footnote{Note a simplification which occurs when $q=0$ (i.\@e.\@ for the relativistic free transport equation, discussed in Section \ref{subsec:relft}).  There (with $d=1$), the vector fields factorise as
\[
	L = (1+p^2)^{-\frac{3}{2}} M,
\]
where $M$ is as in \eqref{eq:relftvectors}.  The function $(1+p^2)^{-\frac{3}{2}}$ is conserved by the operator in \eqref{eq:relft}, and so the vector fields $M$ also commute with equation \eqref{eq:relft}.  It is thus unnecessary, for $q=0$, to directly consider the $L$ vector fields at all.}  It is thus desirable to establish an inequality of the form
\begin{equation} \label{eq:introinequwouldlike}
	\sum_{\vert I \vert = k}
	\big\vert
	\big(
	M
	+
	\eta(t,p)
	\big)^I
	f(t,x,p)
	\big\vert
	\leq
	C_k
	\sum_{\vert I \vert \leq k}
	\big\vert
	L^I f(t,x,p)
	\big\vert,
\end{equation}
for some constants $C_k \geq 1$ (cf.\@ the inequality \eqref{eq:relftbinominequ} for the $q=0$ case).  Indeed, after showing that the matrix ${A_k}^i(t,p)$ is invertible and establishing such an inequality, the proof proceeds much as in the the non-relativistic problem.\footnote{A further step, essentially not present in the non-relativistic problem, involves relating $L^I\vert_{t=1}$ to the standard $\partial_x$ and $\partial_p$ derivatives appearing in the spaces $H^k_q$ (see, for example Proposition \ref{prop:sfcomb2} in the case of $q=\frac{1}{3}$).
\label{footnote:furtherstep}}

An inequality of the form \eqref{eq:introinequwouldlike}, in which there are no growing $t$ factors on the right hand side, in particular encodes the fact $t^{-k(1-2q)}$ is the correct behaviour for $f\in H^k_q$.  Further, in order to treat the case of analytic solutions, it is important to characterise the rate at which the constants $C_k$ grow as $k \to \infty$.  The presence of the losses for analytic solutions, as compared to the non-relativistic problem (see Remark \ref{rmk:analyticloss}), is partly due to the presence of such growing constants (the other reason is due to further such growing constants in the relation discussed in footnote \ref{footnote:furtherstep}).

After establishing the invertibility of ${A_k}^i(t,p)$, the inequality \eqref{eq:introinequwouldlike} is established via another Binomial Theorem-type result for the operator $M+\eta$, as in Proposition \ref{prop:relcomb} for the $q=0$ case.  Such a statement is indeed the main part of the analysis, and one of the main challenges is in finding a suitable analogue of the ansatz \eqref{eq:relftcomb} (which is necessarily more complex than that of Proposition \ref{prop:relcomb}).

\subsubsection*{The case $q=\frac{1}{3}$}

When $q=\frac{1}{3}$ the vector fields are explicit.  Indeed, the function $G_{\frac{1}{3}}$ can be expressed explicitly, and the matrix ${A_k}^i(t,p)$ takes the form
\[
	{A_k}^i(t,p)
	=
	3 (t^{\frac{2}{3}}+\vert t^{\frac{2}{3}} p \vert^2)^{\frac{1}{2}}
	\bigg(
	\delta_k^i
	+
	\frac{t^{\frac{4}{3}} p^i p^k}{t^{\frac{2}{3}}+\vert t^{\frac{2}{3}} p \vert^2}
	\bigg).
\]
It can directly be checked that ${A_k}^i(t,p)$ is invertible for all $t \geq 1$ and $p \in \mathbb{R}^3$, with explicit inverse (see \eqref{eq:sfinverseofA} below) and thus $M_i$ and $\eta_i$ take the form
\[
	M_i
	=
	\frac{
	t^{\frac{1}{3}}
	}{
	3(t^{\frac{2}{3}}+\vert t^{\frac{2}{3}} p \vert^2)^{\frac{1}{2}}
	}
	\bigg(
	\delta^i_k
	-
	\frac{
	t^{\frac{2}{3}}p^i t^{\frac{2}{3}}p^k
	}{
	t^{\frac{2}{3}}+2\vert t^{\frac{2}{3}} p \vert^2
	}
	\bigg)
	L_k,
	\qquad
	\eta(t,p)_i
	=
	\frac{t^{\frac{1}{3}} t^{\frac{2}{3}} p^i}{3 (t^{\frac{2}{3}} + \vert t^{\frac{2}{3}} p \vert^2)^{\frac{1}{2}}}
	\Big(
	\frac{4 \vert t^{\frac{2}{3}} p \vert^2 }{t^{\frac{2}{3}}+2\vert t^{\frac{2}{3}} p \vert^2}
	-
	5
	\Big).
\]
It can be shown (see Proposition \ref{prop:sfcomb}) that the operator $\big( M + \eta(t,p) \big)^I$ can be expressed, for each multi-index $I$, as a linear combination of terms of the form
\[
	t^{\frac{l}{3}} 
	(t^{\frac{2}{3}}+\vert t^{\frac{2}{3}} p \vert^2)^{-\frac{m}{2}}
	(t^{\frac{2}{3}}+2 \vert t^{\frac{2}{3}} p \vert^2)^{-\frac{n}{2}}
	(t^{\frac{2}{3}} p)^J L^K,
\]
for $l,m,n\in \mathbb{N}_0$, multi-indices $K$, and $J\in (\mathbb{N}_0)^3$.  Furthermore, there are restrictions on $l,m,n,K,J$ which ensure that the inequality \eqref{eq:introinequwouldlike} holds for all $k\geq0$.  The constants $C_k$ can moreover be shown to satisfy
\[
	C_k \lesssim 2^k (146k)!,
\]
yielding also a result in the analytic case (which is non-sharp, and improved in the treatment of general $q$).  Since the vector fields are explicit for $q=\frac{1}{3}$, a separate proof in this case is included in Section \ref{subsec:sfproof}.  See Section \ref{subsec:sfproof} for more details.

\subsubsection*{The case $q=\frac{1}{2}$}

When $q=\frac{1}{2}$ the function $G_q$ defined by \eqref{eq:introvf3} is also explicit.  There is again a collection of vector fields of the form \eqref{eq:introvf1}, where the analogue of the matrix \eqref{eq:introvf2} (see \eqref{eq:radiationvectorfields2} below) is again explicitly invertible for all $t \geq 1$ and $p \in \mathbb{R}^3$.  As such, the proof of Theorem \ref{thm:main2} can be treated in a similar manner as the $q=\frac{1}{3}$ case of Theorem \ref{thm:main1} discussed above, though most of the expressions, and thus the corresponding computations, are more complicated.  See Section \ref{subsec:radiationproof} for details.

\subsubsection*{The case of general $0< q < \frac{1}{2}$ with $q \neq \frac{1}{4}, \frac{1}{6},\ldots$}

For general $0< q < \frac{1}{2}$ the function $G_q$ defined by \eqref{eq:introvf3} cannot be evaluated explicitly and, in particular, it is difficult to determine good properties, such as invertibility, of the matrix \eqref{eq:introvf2}.  The quantities are more accessible, however, when $t^q \gg \vert t^{2q} p\vert$.  Indeed, a computation reveals that the matrix \eqref{eq:introvf2} can be expressed as
\[
	{A_k}^i(t,p)
	=
	t^{1-2q}
	\Big[
	H_q \big( t^{-2q} \vert t^{2q} p \vert^2 \big)
	\delta_k^i
	+
	2 t^{-2q}H_q'(t^{-2q} \vert t^{2q} p \vert^2)
	t^{4q} p^i p^k
	\Big],
\]
where $H_q(s^2) = G_q(s^{-\frac{1}{q}}) s^{\frac{1 - 2q}{q}}$.  Provided $q \neq \frac{1}{4}, \frac{1}{6},\ldots$, the function $H_q$ is real analytic and has a convergent power series around $0$ (see Proposition \ref{prop:hanalytic} below).  It follows from properties of $H_q$ around $0$ that ${A_k}^i(t,p)$ is invertible when $t^q \gg \vert t^{2q} p\vert$.  Moreover $M_i$ and $\eta_i$ take the form
\[
	M_i
	=
	\frac{1}{H_q(t^{-2q} \vert t^{2q} p \vert^2)}
	\Big(
	\delta^k_i
	-
	t^{-2q}B_q(t^{-2q} \vert t^{2q} p \vert^2) t^{2q} p^k t^{2q} p^i
	\Big)
	L_k,
	\qquad
	\eta(t,p)_i
	=
	\Phi_q(t^{-2q}\vert t^{2q} p \vert^2) t^{-2q} t^{2q} p^k,
\]
for some functions $B_q$ and $\Phi_q$ which are real analytic function around $0$, provided again that $q \neq \frac{1}{4}, \frac{1}{6},\ldots$ (see Proposition \ref{prop:analyticdivision} below).  Using the power series expansions for $H_q$, $B_q$, and $\Phi_q$ around $0$, it can be shown (see Proposition \ref{prop:comb}) that, when $t^q \gg \vert t^{2q} p\vert$, the operator $\big( M + \eta(t,p) \big)^I$ can be expressed, for each multi-index $I$, as a power series with each term taking the form
\[
	t^{-lq} (t^{-q} t^{2q} p)^J L^K.
\]
for $l\in \mathbb{N}_0$, multi-indices $K$, and $J\in (\mathbb{N}_0)^3$.  The summation is over infinitely many $J$, with summable coefficients, but, for each fixed $I$, only finitely $l$ and $K$.  Furthermore, there are restrictions on $l,J,K$ which ensure that, for $t^q \gg \vert t^{2q} p\vert$, the inequality \eqref{eq:introinequwouldlike} holds for all $k\geq0$.  The constants $C_k$ can moreover be shown to satisfy
\[
	C_k \lesssim 2^k (4k)!,
\]
which can be used to yield a result in the analytic case.

As noted, the above discussion applies only under the condition that
\begin{equation} \label{eq:introtpcondition}
	t^q \gg \vert t^{2q} p\vert.
\end{equation}
The weight $t^{2q} p$ is preserved by the equation and so, if $f_1$ is assumed to be compactly supported, there is a constant $R>0$ such that the solution at all later times has the support property
\[
	\mathrm{supp}(f(t,x,\cdot)) \subset \{ \vert t^{2q} p \vert \leq R\}.
\]
Thus, under this compact support assumption, \eqref{eq:introtpcondition} is guaranteed to hold in the support of $f$ for $t \gg R^{\frac{1}{q}}$.  Since Theorem \ref{thm:main1} is a statement about the long time behaviour of solutions, it is only necessary to consider the $L$ vector fields (and thus the above discussion) for such large $t$, and it suffices to consider only the coordinate derivatives $\partial_{x^i}$ as commutators prior to these late times.

In order to treat $f_1$ which are not necessarily compactly supported (in particular, to treat non-trivial examples of analytic $f_1$) a dyadic physical space localisation of the solution $f$ is considered.  The solution is written as $f = \sum_{n=0}^{\infty} f^n$ where each $f^n$ solves the Vlasov equation \eqref{eq:Vlasovgeneral} and satisfies the support property
\[
	\supp(f^n) \subset \{ (t,x,p) \mid 2^{n-1} \leq \vert t^{2q} p \vert \leq 2^{n+1} \}.
\]
Each dyadic piece $f^n$ of the solution can then be treated as above --- using the coordinate $\partial_{x^i}$ as commutators for $t < T_n$, and the $L_i$ vector fields for times $t \geq T_n$, where $\{T_n\}_{n=0}^{\infty}$ is a suitable sequence of times with $T_n \to \infty$ as $n \to \infty$.  See Section \ref{subsec:generalproof} for further details.

\subsubsection*{The case of $q = \frac{1}{4}, \frac{1}{6},\ldots$}
The values $q = \frac{1}{4}, \frac{1}{6},\ldots$ are excluded from Theorem \ref{thm:main1} as the function $s \mapsto G_q \big( s^{-\frac{1}{q}} \big) s^{\frac{1-2q}{q}}$, with $G_q$ defined by \eqref{eq:introvf3} --- and hence the vector fields \eqref{eq:introvf1} --- are not smooth but slightly singular in these cases.  In Remark \ref{rmk:repgeneral} below a different collection of vector fields, which are regular for all $q$, are given.  We expect that a version of Theorem \ref{thm:main1}, for all $0<q<\frac{1}{2}$, can be given using these vector fields but, in order to the simplify the proof, do not do so here.

\subsection{Outline of the paper}

Section \ref{section:preliminaries} contains certain preliminaries.  The notation used throughout is introduced, various functional inequalities are presented, and some conservation laws for equation \eqref{eq:Vlasovgeneral} are given.  In Section \ref{section:proof}, the proofs of Theorem \ref{thm:main1} and Theorem \ref{thm:main2} are given.

\subsection*{Acknowledgements}

We acknowledge support through Royal Society Tata University Research Fellowship URF\textbackslash R1\textbackslash 191409.  We are grateful to M.\@ Tiba for helpful discussions.

\section{Preliminaries}
\label{section:preliminaries}

This section contains certain preliminaries to the proof of Theorem \ref{thm:main1} and Theorem \ref{thm:main2}.  In Section \ref{subsec:notation} some notation is introduced, which will be used throughout.  In Section \ref{subsec:inequalities} certain functional inequalities and combinatorial statments are presented.  Section \ref{subsec:GandH} concerns some basic properties of the functions $G_q$ and $H_q$, which feature later in the vector fields of Section \ref{section:proof}.  In Section \ref{subsec:conservation} some conservation laws for equation \eqref{eq:Vlasovgeneral} are stated.

\subsection{Notation} \label{subsec:notation}
In this section the notation used throughout the paper is introduced.

\subsubsection{Spatial average}
For any function $\phi \colon \mathbb{T}^3 \to \mathbb{R}$, the spatial average is denoted
\[
	\overline{\phi} = \int_{\mathbb{T}^3} \phi(x) dx.
\]

\subsubsection{Multi-indices} \label{subsec:multiindices}
In what follows a collection of vector fields $\{L_1,L_2,L_3\}$ will be introduced.  Given $k \in \mathbb{N}_0$, a \emph{multi-index of length $k$} is defined to be a collection $I=(i_1,i_2\ldots,i_k)$, with $i_j \in \{ (1,0,0),(0,1,0),(0,0,1) \}$ for each $j =1,\ldots,k$.  Given such a multi-index $I$, define
\[
	\partial_x^I = \partial_{x}^{i_1} \ldots \partial_{x}^{i_k},
	\qquad
	\partial_p^I = \partial_{p}^{i_1} \ldots \partial_{p}^{i_k},
	\qquad
	L^I = L^{i_1} \ldots L^{i_k},
\]
where
\[
	\partial_x^{(1,0,0)} = \partial_{x^1},
	\quad
	\partial_x^{(0,1,0)} = \partial_{x^2},
	\quad
	\partial_x^{(0,0,1)} = \partial_{x^3},
	\qquad
	L^{(1,0,0)} = L_1,
	\quad
	L^{(0,1,0)} = L_2,
	\quad
	L^{(0,0,1)} = L_3,
\]
etc.
Define also, for $p \in \mathbb{R}^3$,
\[
	p^I = p^{i_1} p^{i_2} \ldots p^{i_k},
\]
where
\[
	p^{(1,0,0)} = p^1,
	\qquad
	p^{(0,1,0)} = p^2,
	\qquad
	p^{(0,0,1)} = p^3.
\]
For any function $h(t,x,p)$, define
\[
	\big( L+h(t,x,p) \, p \big)^I
	=
	\big( L^{i_1}+h(t,x,p) \, p^{i_1} \big) \big( L^{i_2}+h(t,x,p) \, p^{i_2} \big) \ldots \big( L^{i_k}+h(t,x,p) \, p^{i_k} \big).
\]
For such $I$, define also $\vert I \vert = k$.

For $J = (j_1,j_2,j_3) \in (\mathbb{N}_0)^3$, define $\vert J \vert = j_1+j_2+j_3$ and, for $p \in \mathbb{R}^3$,
\[	
	p^J = (p^1)^{j_1} (p^2)^{j_2} (p^3)^{j_3}.
\]

Given a multi-index $K = (i_1,\ldots,i_k)$ with $i_1,\ldots,i_k \in \{ (1,0,0), (0,1,0),(0,0,1)\}$, define
\[
	K-e_i
	=
	\begin{cases}
		(i_2,\ldots,i_{k})
		&
		\text{if } i_1 = e_i \text { and } k \geq 2,
		\\
		0
		&
		\text{if } i_1 = e_i \text { and } k = 1,
		\\
		\emptyset
		&
		\text{otherwise}.
	\end{cases}
\]
Define also
\[
	K+e_i = (e_i,i_1,\ldots,i_k).
\]
If $0=(0,0,0)$ is the zero vector, then define
\begin{equation} \label{eq:sfzeroconvention}
	(p^1,p^2,p^3)^0 = 1, \qquad L^0=1.
\end{equation}
Thus, for example, if $K = (e_i,e_j)$, then $L^{K-e_i} = L_j$.  If $K = (e_i)$, then $L^{K-e_i} = 1$.

\subsubsection{Function spaces} \label{subsec:functionspaces}
For $k \in \mathbb{N}$ and $0<q< \frac{1}{2}$, define the weighted Sobolev norms on functions $h \colon \mathbb{T}^3\times \mathbb{R}^3 \to \mathbb{R}$ by
\begin{align*}
	\Vert h \Vert_{H_{\circ}^k(\mathbb{T}^3\times \mathbb{R}^3)}^2
	&
	=
	\sum_{\vert I \vert + \vert J \vert \leq k} \int_{\mathbb{T}^3} \int_{\mathbb{R}^3} 
	(
	\vert p \vert^2 + \vert p \vert^4
	)
	\vert \partial_x^I \partial_p^J h(x,p) \vert^2 dp dx,
	\\
	\Vert h \Vert_{H_{q}^k(\mathbb{T}^3\times \mathbb{R}^3)}^2
	&
	=
	\sum_{\vert I \vert + \vert J \vert \leq k} \int_{\mathbb{T}^3} \int_{\mathbb{R}^3} 
	(
	\vert p \vert^2 + \vert p \vert^4
	)
	(
	1+\vert p \vert
	)^{\vert I \vert \frac{1-2q}{q}}
	\vert \partial_x^I \partial_p^J h(x,p) \vert^2 dp dx,
	\\
	\Vert h \Vert_{H_{\log}^k(\mathbb{T}^3\times \mathbb{R}^3)}^2
	&
	=
	\sum_{\vert I \vert + \vert J \vert \leq k} \int_{\mathbb{T}^3} \int_{\mathbb{R}^3} 
	(
	\vert p \vert^2 + \vert p \vert^4
	)
	\big(\log (
	2+\vert p \vert
	) \big)^{\vert I \vert}
	\vert \partial_x^I \partial_p^J h(x,p) \vert^2 dp dx,
	\\
	\Vert h \Vert_{H^k_{\mathrm{exp}}(\mathbb{T}^3\times \mathbb{R}^3)}^2
	&
	=
	\sum_{\vert I \vert + \vert J \vert \leq k} \int_{\mathbb{T}^3} \int_{\mathbb{R}^3} 
	e^{2\vert p \vert}
	\vert \partial_x^I \partial_p^J h(x,p) \vert^2 dp dx,
\end{align*}
along with the corresponding homogeneous semi-norms
\begin{align*}
	\Vert h \Vert_{\mathring{H}_{\circ}^k(\mathbb{T}^3\times \mathbb{R}^3)}^2
	&
	=
	\sum_{\vert I \vert + \vert J \vert = k} \int_{\mathbb{T}^3} \int_{\mathbb{R}^3} 
	(
	\vert p \vert^2 + \vert p \vert^4
	)
	\vert \partial_x^I \partial_p^J h(x,p) \vert^2 dp dx,
	\\
	\Vert h \Vert_{\mathring{H}_{q}^k(\mathbb{T}^3\times \mathbb{R}^3)}^2
	&
	=
	\sum_{\vert I \vert + \vert J \vert = k} \int_{\mathbb{T}^3} \int_{\mathbb{R}^3} 
	(
	\vert p \vert^2 + \vert p \vert^4
	)
	(
	1+\vert p \vert
	)^{\vert I \vert \frac{1-2q}{q}}
	\vert \partial_x^I \partial_p^J h(x,p) \vert^2 dp dx,
	\\
	\Vert h \Vert_{\mathring{H}_{\log}^k(\mathbb{T}^3\times \mathbb{R}^3)}^2
	&
	=
	\sum_{\vert I \vert + \vert J \vert = k} \int_{\mathbb{T}^3} \int_{\mathbb{R}^3} 
	(
	\vert p \vert^2 + \vert p \vert^4
	)
	\big(\log (
	2+\vert p \vert^2
	) \big)^{\vert I \vert}
	\vert \partial_x^I \partial_p^J h(x,p) \vert^2 dp dx,
	\\
	\Vert h \Vert_{\mathring{H}^k_{\mathrm{exp}}(\mathbb{T}^3\times \mathbb{R}^3)}^2
	&
	=
	\sum_{\vert I \vert + \vert J \vert = k} \int_{\mathbb{T}^3} \int_{\mathbb{R}^3} 
	e^{2\vert p \vert}
	\vert \partial_x^I \partial_p^J h(x,p) \vert^2 dp dx,
\end{align*}
and the Sobolev spaces
\begin{align*}
	H_{\circ}^k(\mathbb{T}^3\times \mathbb{R}^3)
	&
	=
	\Big\{
	h \colon \mathbb{T}^3\times \mathbb{R}^3 \to \mathbb{R}
	\, \Big\vert \,
	\Vert h \Vert_{H_{\circ}^k(\mathbb{T}^3\times \mathbb{R}^3)} < \infty
	\Big\},
	\\
	H^k_q(\mathbb{T}^3\times \mathbb{R}^3)
	&
	=
	\Big\{
	h \colon \mathbb{T}^3\times \mathbb{R}^3 \to \mathbb{R}
	\, \Big\vert \,
	\Vert h \Vert_{H^k_q(\mathbb{T}^3\times \mathbb{R}^3)} < \infty
	\Big\},
	\\
	H^k_{\log}(\mathbb{T}^3\times \mathbb{R}^3)
	&
	=
	\Big\{
	h \colon \mathbb{T}^3\times \mathbb{R}^3 \to \mathbb{R}
	\, \Big\vert \,
	\Vert h \Vert_{H^k_{\log}(\mathbb{T}^3\times \mathbb{R}^3)} < \infty
	\Big\},
	\\
	H^k_{\mathrm{exp}}(\mathbb{T}^3\times \mathbb{R}^3)
	&
	=
	\Big\{
	h \colon \mathbb{T}^3\times \mathbb{R}^3 \to \mathbb{R}
	\, \Big\vert \,
	\Vert h \Vert_{H^k_{\mathrm{exp}}(\mathbb{T}^3\times \mathbb{R}^3)} < \infty
	\Big\}.
\end{align*}
Define also the analytic spaces of smooth functions
\begin{align*}
	H_{\circ}^{\omega}(\mathbb{T}^3\times \mathbb{R}^3)
	&
	=
	\Big\{
	h \colon \mathbb{T}^3\times \mathbb{R}^3 \to \mathbb{R}
	\, \Big\vert \,
	\exists \lambda>0 \text{ such that }
	\Vert h \Vert_{\mathring{H}_{\circ}^k(\mathbb{T}^3\times \mathbb{R}^3)} \leq \frac{k!}{\lambda^k}
	\text{ for all } k \geq 0
	\Big\},
	\\
	H_q^{\omega}(\mathbb{T}^3\times \mathbb{R}^3)
	&
	=
	\Big\{
	h \colon \mathbb{T}^3\times \mathbb{R}^3 \to \mathbb{R}
	\, \Big\vert \,
	\exists \lambda>0 \text{ such that }
	\Vert h \Vert_{\mathring{H}_q^k(\mathbb{T}^3\times \mathbb{R}^3)} \leq \frac{k!}{\lambda^k}
	\text{ for all } k \geq 0
	\Big\},
	\\
	H_{\log}^{\omega}(\mathbb{T}^3\times \mathbb{R}^3)
	&
	=
	\Big\{
	h \colon \mathbb{T}^3\times \mathbb{R}^3 \to \mathbb{R}
	\, \Big\vert \,
	\exists \lambda>0 \text{ such that }
	\Vert h \Vert_{\mathring{H}_{\log}^k(\mathbb{T}^3\times \mathbb{R}^3)} \leq \frac{k!}{\lambda^k}
	\text{ for all } k \geq 0
	\Big\},
	\\
	H^{\omega}_{\mathrm{exp}}(\mathbb{T}^3\times \mathbb{R}^3)
	&
	=
	\Big\{
	h \colon \mathbb{T}^3\times \mathbb{R}^3 \to \mathbb{R}
	\, \Big\vert \,
	\exists \lambda>0 \text{ such that }
	\Vert h \Vert_{\mathring{H}^k_{\mathrm{exp}}(\mathbb{T}^3\times \mathbb{R}^3)} \leq \frac{k!}{\lambda^k}
	\text{ for all } k \geq 0
	\Big\}.
\end{align*}
For $h \in H^{\omega}_q(\mathbb{T}^3\times \mathbb{R}^3)$, define
\[
	\lambda(h)
	=
	\sup
	\Big\{
	\lambda>0
	\, \Big\vert \,
	\Vert h \Vert_{\mathring{H}^k_q(\mathbb{T}^3\times \mathbb{R}^3)} \leq \frac{k!}{(2\lambda)^k}
	\text{ for all } k \geq 0
	\Big\},
\]
and similarly for $h \in H_{\circ}^{\omega}$, $h \in H_{\log}^{\omega}$, and $h \in H_{\mathrm{exp}}^{\omega}$.  Define also the analytic norm
\[
	\Vert h \Vert_{H^{\omega}_q(\mathbb{T}^3\times \mathbb{R}^3)}
	=
	\sum_{k=0}^{\infty}
	\frac{\lambda(h)^k}{k!}
	\Vert h \Vert_{\mathring{H}^k_q(\mathbb{T}^3\times \mathbb{R}^3)},
\]
so that
\[
	\Vert h \Vert_{\mathring{H}^k_q(\mathbb{T}^3\times \mathbb{R}^3)}
	\leq
	\frac{k!}{\lambda(h)^k}
	\Vert h \Vert_{H^{\omega}_q(\mathbb{T}^3\times \mathbb{R}^3)}
	,
	\qquad
	\text{for all}
	\quad
	k \geq 0.
\]
Similarly for $h \in H_{\circ}^{\omega}(\mathbb{T}^3\times \mathbb{R}^3)$, $h \in H_{\log}^{\omega}(\mathbb{T}^3\times \mathbb{R}^3)$, and $h \in H_{\mathrm{exp}}^{\omega}(\mathbb{T}^3\times \mathbb{R}^3)$.

Recall that, for any $l \in \mathbb{N}$,
\[
	(1+\vert p\vert^{l})^2
	\leq
	(l !)^2 \, e^{2\vert p \vert},
\]
and so, for any function $h \colon \mathbb{T}^3\times \mathbb{R}^3 \to \mathbb{R}$,
\begin{equation} \label{eq:expweightconv}
	\sum_{\vert I \vert + \vert J \vert = k} 
	\Big(
	\int_{\mathbb{T}^3} \int_{\mathbb{R}^3} 
	(1+\vert p\vert^{2l})
	\vert \partial_x^I \partial_p^J h(x,p) \vert^2 dp dx
	\Big)^{\frac{1}{2}}
	\leq
	l! \,
	\Vert h \Vert_{\mathring{H}^k_{\mathrm{exp}}(\mathbb{T}^3\times \mathbb{R}^3)}.
\end{equation}
This fact was used in the proof of Theorem \ref{thm:relft}.  (Some form of super-polynomial weight seems to be necessary, for the the proof of Theorem \ref{thm:relft}, in the definition of the space $H^{\omega}_{\mathrm{exp}}$.  In the case of finite $k$, however, the results of Theorem \ref{thm:relft} still hold if this weight is suitably relaxed.)

Though the main results of the present work do not feature Gevrey spaces, the following spaces are defined for the purposes of Remark \ref{rmk:Gevrey}.  Given $s \geq 1$, define the Gevrey space of smooth functions
\[
	\mathcal{G}^s_q (\mathbb{T}^3\times \mathbb{R}^3)
	=
	\Big\{
	h \colon \mathbb{T}^3\times \mathbb{R}^3 \to \mathbb{R}
	\, \Big\vert \,
	\exists \lambda>0 \text{ such that }
	\Vert h \Vert_{\mathring{H}_q^k(\mathbb{T}^3\times \mathbb{R}^3)} \leq \frac{(k!)^s}{\lambda^k}
	\text{ for all } k \geq 0
	\Big\},
\]
along with the norm, for $h \in \mathcal{G}^s_q (\mathbb{T}^3\times \mathbb{R}^3)$,
\[
	\Vert h \Vert_{\mathcal{G}^s_q(\mathbb{T}^3\times \mathbb{R}^3)}
	=
	\sum_{k=0}^{\infty}
	\frac{\lambda_s(h)^k}{(k!)^s}
	\Vert h \Vert_{\mathring{H}_q^k(\mathbb{T}^3\times \mathbb{R}^3)},
	\qquad
	\lambda_s(h)
	=
	\sup
	\Big\{
	\lambda>0
	\, \Big\vert \,
	\Vert h \Vert_{\mathring{H}_q^k(\mathbb{T}^3\times \mathbb{R}^3)} \leq \frac{(k!)^s}{(2\lambda)^k}
	\text{ for all } k \geq 0
	\Big\}.
\]
The space $\mathcal{G}^s_{\log}$, along with its associated norm, is defined analogously.

\subsubsection{Constants}
The notation
\[
	A \lesssim B,
\]
will be used when there is a universal constant $C$ such that
\[
	A \leq CB.
\]
When such notation is used, the constant $C$ may depend on the value of $q$ under consideration (and may blow up as $q$ approaches $\frac{1}{4}, \frac{1}{6}, \ldots$), but never on the value of the number of derivatives being considered.

\subsection{Functional inequalities}
\label{subsec:inequalities}

In this section some functional inequalities are collected, which will be used in the sequel.

\subsubsection{Sobolev inequality}
The first functional inequality is a standard $L^{\infty}$--$L^2$ Sobolev inequality.

\begin{proposition}[$L^{\infty}$--$L^2$ Sobolev inequality] \label{prop:Sobolev}
	For any smooth function $\phi \colon \mathbb{T}^3 \to \mathbb{R}$, and any $k \geq 2$,
	\[
		\sup_{x\in \mathbb{T}^3} \vert \phi(x) - \overline{\phi} \vert
		\lesssim
		\frac{1}{\sqrt{k}}
		\sum_{\vert I \vert = k}
		\Vert \partial_x^I \phi \Vert_{L^2(\mathbb{T}^3)}.
	\]
\end{proposition}

\begin{proof}
	The function $\phi$ can be Fourier decomposed
	\[
		\phi(x) = \sum_{\xi \in \mathbb{Z}^3} \hat{\phi}(\xi) e^{2\pi i\xi\cdot x},
		\qquad
		\hat{\phi}(\xi) = \int_{\mathbb{T}^3} \phi(x) e^{-2\pi i\xi\cdot x} dx,
	\]
	and so
	\[
		\vert \phi(x) - \overline{\phi} \vert
		\leq
		\sum_{\vert \xi \vert \neq 0} \vert \hat{\phi}(\xi) \vert
		\leq
		\Big( \sum_{\vert \xi \vert \neq 0} \vert \xi \vert^{-2k} \Big)^{\frac{1}{2}}
		\Big( \sum_{\xi \in \mathbb{Z}^3} \vert \xi \vert^{2k} \vert \hat{\phi}(\xi) \vert^2 \Big)^{\frac{1}{2}}
		=
		\Big( \sum_{\vert \xi \vert \neq 0} \vert \xi \vert^{-2k} \Big)^{\frac{1}{2}}
		\sum_{\vert I \vert = k}
		\Vert \partial_x^I \phi \Vert_{L^2(\mathbb{T}^3)},
	\]
	by the Cauchy--Schwarz Inequality and the Plancherel Theorem.  The proof follows from the fact that $\vert \xi \vert^{-2k}$ is summable for $k \geq 2$ and satisfies
	\[
		\sum_{\xi \in \mathbb{Z}^3 \smallsetminus \{ 0\}} \vert \xi \vert^{-2k}
		\lesssim
		\frac{1}{k}.
	\]
\end{proof}

\subsubsection{Interpolation inequality}

The next functional inequality is an interpolation inequality concerning the $L^2$ norm of the momentum average of a function $h\colon \mathbb{T}^3 \times \mathbb{R}^3 \to \mathbb{R}$ and certain weighted $L^2$ norms of $h$.

\begin{proposition}[Interpolation inequality] \label{prop:interp}
	For any suitably decaying function $h \colon \mathbb{T}^3 \times \mathbb{R}^3 \to \mathbb{R}$,
	\[
		\int_{\mathbb{T}^3} \vert \rho_h(x) \vert^2 dx
		\lesssim
		\left( \int_{\mathbb{T}^3} \int_{\mathbb{R}^3} \vert p \vert^2 \vert h(x,p) \vert^2 dp dx \right)^{\frac{1}{2}}
		\left( \int_{\mathbb{T}^3} \int_{\mathbb{R}^3} \vert p \vert^4 \vert h(x,p) \vert^2 dp dx \right)^{\frac{1}{2}},
	\]
	where
	\[
		\rho_h(x) = \int_{\mathbb{R}^3} h(x,p) dp.
	\]
\end{proposition}

\begin{proof}
	Note first that the integral
	\[
		\int_{\mathbb{R}^3} \frac{1}{\vert p \vert^2 + \vert p \vert^4} dp,
	\]
	is finite.  Thus, for any suitably decaying function $h \colon \mathbb{T}^3 \times \mathbb{R}^3 \to \mathbb{R}$,
	\begin{equation} \label{eq:rhohnoninterp}
		\int_{\mathbb{T}^3} \vert \rho_h(x) \vert^2 dx
		\lesssim
		\int_{\mathbb{T}^3} \int_{\mathbb{R}^3} (\vert p \vert^2 + \vert p \vert^4) \vert h(x,p) \vert^2 dp dx.
	\end{equation}
	Defining now
	\[
		h_{\lambda}(x,p) = h(x,\lambda p),
	\]
	it follows that
	\[
		\int_{\mathbb{T}^3} \vert \rho_{h_{\lambda}}(x) \vert^2 dx
		=
		\lambda^{-6} \int_{\mathbb{T}^3} \vert \rho_{h}(x) \vert^2 dx,
		\qquad
		\int_{\mathbb{T}^3} \int_{\mathbb{R}^3} \vert p \vert^2 \vert h_{\lambda}(x,p) \vert^2 dp dx
		=
		\lambda^{-5} \int_{\mathbb{T}^3} \int_{\mathbb{R}^3} \vert p \vert^2 \vert h (x,p) \vert^2 dp dx,	
	\]
	\[
		\int_{\mathbb{T}^3} \int_{\mathbb{R}^3} \vert p \vert^4 \vert h_{\lambda}(x,p) \vert^2 dp dx
		=
		\lambda^{-7} \int_{\mathbb{T}^3} \int_{\mathbb{R}^3} \vert p \vert^4 \vert h(x,p) \vert^2 dp dx.
	\]
	Applying \eqref{eq:rhohnoninterp} to $h_{\lambda}$ then gives
	\[
		\int_{\mathbb{T}^3} \vert \rho_h(x) \vert^2 dx
		\lesssim
		\lambda \int_{\mathbb{T}^3} \int_{\mathbb{R}^3} \vert p \vert^2 \vert h(x,p) \vert^2 dp dx
		+
		\lambda^{-1} \int_{\mathbb{T}^3} \int_{\mathbb{R}^3} \vert p \vert^4 \vert h(x,p) \vert^2 dp dx.
	\]
	The proof follows from setting
	\[
		\lambda = \left( \int_{\mathbb{T}^3} \int_{\mathbb{R}^3} \vert p \vert^2 \vert h(x,p) \vert^2 dp dx \right)^{-\frac{1}{2}}
		\left( \int_{\mathbb{T}^3} \int_{\mathbb{R}^3} \vert p \vert^4 \vert h(x,p) \vert^2 dp dx \right)^{\frac{1}{2}}.
	\]
\end{proof}

\subsubsection{Combinatorial results}

This section contains some combinatorial results which will be used in the proof of the main results.

\begin{proposition}[Stirling's formula] \label{prop:Stirling}
	There is a sequence $\{r_n\}_{n=1}^{\infty}$ such that $r_n \to 0$ as $n \to \infty$ and, for any $n \in \mathbb{N}$,
	\[
		\frac{n! e^n}{n^n \sqrt{2\pi n}}
		=
		e^{r_n}.
	\]
\end{proposition}

\begin{proof}
	There are many proofs of this well known fact --- see \cite{Rob}, for example.
\end{proof}

\begin{proposition}[Products of factorials and polynomials] \label{prop:factorials}
	For any $n_1,n_2 \in \mathbb{N}$,
	\[
		(n_1k)! (n_2 k)! \leq ((n_1+n_2)k)!,
	\]
	for all $k \in \mathbb{N}$.
\end{proposition}

\begin{proof}
	Suppose, without loss of generality, that $n_1 \leq n_2$.  Then
	\[
		\frac{((n_1+n_2)k)!}{(n_2 k)!}
		=
		(n_2k+1) \times (n_2k+2) \times \ldots \times (n_2k+k) \times (n_2k+k+1)\times \ldots \times (n_2k +n_1k)
		\geq
		(n_1 k)!,
	\]
	and the proof follows.
\end{proof}

\begin{lemma}[Summation Lemma] \label{prop:summationlemma}
	For any integer $k \geq 1$,
	\[
		\sum_{j = 0}^{\infty}
		(1+j)^{k}
		e^{-j}
		\lesssim
		k! \cdot k^{\frac{1}{2}}.
	\]
\end{lemma}

\begin{proof}
	Note that the the function $f(x) = (1+x)^k e^{-x}$ agrees with the summand whenever $x$ is an integer, and has a global maximum at $x=k-1$.  It follows that the summation is uniformly bounded by
	\[
		\sum_{j = 0}^{\infty}
		(1+j)^{k}
		e^{-j}
		\lesssim
		\int_0^k (1+x)^k e^{-x}dx
		+
		\int_k^{\infty} x^k e^{1-x}dx.
	\]
	Now Proposition \ref{prop:Stirling} implies that
	\[
		\int_0^k (1+x)^k e^{-x}dx
		\leq
		k k^k e^{1-k}
		\lesssim
		k! \cdot k^{\frac{1}{2}},
	\]
	and, after integrating by parts $k$ times,
	\[
		\int_k^{\infty} x^k e^{-x}dx
		=
		k^k e^{-k} \sum_{l=0}^k \frac{k!}{(k-l)! k^l}
		\lesssim
		k! \cdot k^{\frac{1}{2}},
	\]
	and the result follows.
\end{proof}

\subsection{The functions $G_q$ and $H_q$}
\label{subsec:GandH}

In this section, functions $G_q$ and $H_q$ are introduced, which will appear in the commutation vector fields introduced in Section \ref{section:proof} below.  Various analyticity properties of these functions are also collected.

Consider $0 < q < \frac{1}{2}$ such that $\frac{1}{2q}$ is not an integer (i.\@e.\@ $q \neq \frac{1}{4}, \frac{1}{6}, \frac{1}{8},\ldots$).  Let $N_q$ be the unique positive integer such that
\begin{equation} \label{eq:propertyofNq}
	\frac{1}{2(N_q+2)} < q < \frac{1}{2(N_q+1)}.
\end{equation}
Define $G_q\colon [0,\infty) \to \mathbb{R}$ by
\begin{equation} \label{eq:Gqprime}
	G_q(s) = \int_0^s \frac{1}{\sqrt{\tilde{s}^{4q}+\tilde{s}^{2q}}} d \tilde{s} + X_q,
\end{equation}
where
\begin{equation} \label{eq:defofX}
	X_q
	=
	-
	\int_0^1 \frac{1}{\sqrt{s^{4q}+s^{2q}}} ds
	+
	\sum_{n=0}^{N_q} \frac{a_n}{1-2q(n+1)}
	-
	\int_1^{\infty}\frac{1}{s^{2q}} \Big( \frac{1}{\sqrt{1+s^{-2q}}} - \sum_{n=0}^{N_q} a_n s^{-2qn} \Big) ds,
\end{equation}
and
\[
	a_0 = 1,
	\qquad
	a_n = \frac{(-1)^n \prod_{l=1}^n (2l-1)}{2^n \, n!}
	\text{ for }
	n\geq1.
\]
The coefficients $a_n$ have the property that
\begin{equation} \label{eq:1sqrt1xexpansion}
	\frac{1}{\sqrt{1+x}} = \sum_{n=0}^{\infty} a_n x^n,
	\qquad
	\text{for }
	\vert x \vert <1,
\end{equation}
and so it follows that the integrals in \eqref{eq:defofX} are convergent (note, in particular, that the property \eqref{eq:propertyofNq} of $N_q$ implies that $2q(n+1) > 1$ for all $n \geq N_q+1$).  This precise form of $X_q$ is chosen so that $G_q(s)$ takes the form \eqref{eq:Gtakestheform} below, and thus admits the expansion \eqref{eq:Gslargeexpansion}.

\begin{remark}[Explicit expressions for $G_q$]
	For certain values of $q$ the function $G_q$ can be expressed explicitly.  For example when $q = \frac{1}{3}$, $G_q$ takes the explicit form
	\begin{equation} \label{eq:G13}
		G_{\frac{1}{3}}(s) = 3\sqrt{1+s^{\frac{2}{3}}},
	\end{equation}
	and when $q = \frac{1}{4}$ (a case otherwise not considered in the present work), $G_q$ takes the explicit form
	\[
		G_{\frac{1}{4}}(s)
		=
		X_{\frac{1}{4}}
		+
		2 \sqrt{s+ s^{\frac{1}{2}}}
		-
		\log
		\Big(
		1 + 2 s^{\frac{1}{2}} + 2\sqrt{s+ s^{\frac{1}{2}}}
		\Big).
	\]
	For general $q$, however, $G_q$ cannot be expressed explicitly.
\end{remark}

Define also $H_q \colon [0,\infty) \to \mathbb{R}$ by,
\begin{equation} \label{eq:defofH}
	H_q(0) = \frac{1}{1-2q},
	\qquad
	H_q(s) = \frac{G_q( s^{-\frac{1}{2q}})}{s^{1-\frac{1}{2q}}}
	\text{ for } s\neq 0.
\end{equation}
so that
\[
	H_q(s^2) = \frac{G_q(s^{-\frac{1}{q}})}{s^{2-\frac{1}{q}}} = G_q(s^{-\frac{1}{q}}) s^{\frac{1 - 2q}{q}}.
\]

\begin{proposition}[Analyticity of $H_q$] \label{prop:hanalytic}
	Consider $0 < q < \frac{1}{2}$ such that $\frac{1}{2q}$ is not an integer (i.\@e.\@ $q \neq \frac{1}{4}, \frac{1}{6}, \frac{1}{8},\ldots$).  The function $H_q \colon [0,\infty) \to \mathbb{R}$ defined by \eqref{eq:defofH} is real analytic (in fact, $H_q$ extends to a real analytic function $H_q \colon (-1,\infty) \to \mathbb{R}$).  Moreover, for any $0\leq s < 1$,
	\begin{equation} \label{eq:analyticseriesforH}
		H_q(s) = \sum_{n=0}^{\infty} b_n s^n,
		\qquad
		b_0 = \frac{1}{1-2q},
		\qquad
		b_n = \frac{a_n}{1-2q(n+1)}
		\text{ for }
		n\geq1.
	\end{equation}
\end{proposition}

\begin{proof}
	Note that $H_q$ satisfies the ordinary differential equation
	\[
		\frac{d}{ds} H_q(s) + \frac{2q-1}{2q} \frac{H_q(s)}{s} = - \frac{1}{2q s\sqrt{1+s}}.
	\]
	It follows from the Cauchy--Kovalevskaya Theorem that $H_q$ is real analytic on $(0,\infty)$.  Suppose now that $0< s <1$.  Note first that, for $s>1$, using the expression \eqref{eq:defofX} for $X_q$,
	\begin{equation} \label{eq:Gtakestheform}
		G_q(s)
		=
		\sum_{n=0}^{N_q} \frac{a_ns^{1-2q(n+1)}}{1-2q(n+1)}
		-
		\int_s^{\infty}
		\frac{1}{\tilde{s}^{2q}} \Big( \frac{1}{\sqrt{1+\tilde{s}^{-2q}}} - \sum_{n=0}^{N_q} a_n \tilde{s}^{-2qn} \Big)
		d\tilde{s},
	\end{equation}
	where the fact that that the property \eqref{eq:propertyofNq} implies that $2q(n+1) < 1$ for all $0 \leq n \leq N_q$ has been used.  It follows that, for $0<s<1$,
	\[
		H_q(s)
		=
		\frac{G_q( s^{-\frac{1}{2q}})}{s^{1-\frac{1}{2q}}}
		=
		\sum_{n=0}^{N_q} \frac{a_ns^n}{1-2q(n+1)}
		-
		\frac{1}{s^{1-\frac{1}{2q}}}
		\int_{s^{-\frac{1}{2q}}}^{\infty}
		\frac{1}{\tilde{s}^{2q}} \bigg[ \frac{1}{\sqrt{1+\tilde{s}^{-2q}}} - \sum_{n=0}^{N_q} a_n \tilde{s}^{-2qn} \bigg]
		d \tilde{s},
	\]
	Consider some $N\geq N_q$.  The property \eqref{eq:1sqrt1xexpansion} implies that
	\[
		\bigg\vert H_q(s) - \sum_{n=0}^N b_n s^n \bigg\vert
		\leq
		\frac{s^{N+1}}{2q(N+2)-1},
		\qquad
		b_0 = \frac{1}{1-2q},
		\qquad
		b_n = \frac{a_n}{1-2q(n+1)}
		\text{ for }
		n\geq1.
	\]
	Thus, letting $N \to \infty$, the expansion \eqref{eq:analyticseriesforH} holds for any $0\leq s < 1$ and so $H_q$ is analytic around $0$ (in fact, $H_q$ extends to a real analytic function $H_q \colon (-1,\infty) \to \mathbb{R}$).
\end{proof}

\begin{remark}[Expansion for $G_q$]
	It follows from the proof of Proposition \ref{prop:hanalytic} that $G_q$ admits the expansion, for large $s$,
	\begin{equation} \label{eq:Gslargeexpansion}
		G_q(s) = \sum_{n=0}^{\infty} \frac{a_ns^{1-2q(n+1)}}{1-2q(n+1)} = \sum_{n=0}^{\infty} b_n s^{1-2(n+1)q},
		\qquad
		\text{for }
		s > 1.
	\end{equation}
\end{remark}

The following proposition is used later when relating combinations of vector fields at $t=1$ to the standard $\partial_x$ and $\partial_p$ derivatives appearing in the spaces $H^k_q$.

\begin{proposition}[$H_q$ for large $\vert p \vert$] \label{prop:hanalyticlargep}
	There are real analytic functions $\psi_q, \phi_q \colon (-1,1) \to \mathbb{R}$ such that, for all $\vert p \vert >1$,
	\[
		H_q(\vert p \vert^2)
		=
		X_q \vert p\vert^{\frac{1-2q}{q}}
		+
		\psi_q(\vert p \vert^{-1}),
		\qquad
		2 H_q'(\vert p \vert^2)
		=
		Y_q \vert p\vert^{\frac{1-4q}{q}} + \phi_q(\vert p\vert^{-1}),
	\]
	where $Y_q = \frac{1-2q}{q} X_q$ and
	\begin{equation} \label{eq:defofcn}
		\psi_q(s) = \sum_{n=0}^{\infty} c_n s^{2n+1},
		\qquad
		\phi_q(s) = - \sum_{n=0}^{\infty} (2n+1) c_n s^{2n+3},
		\qquad
		c_n = \frac{a_n}{1+(2n-1)q},
		\quad
		\text{ for }
		n\geq0.
	\end{equation}
\end{proposition}

\begin{proof}
	Recall that
	\[
		G_q(s) = X_q + \int_0^s \frac{1}{\tilde{s}^q} \frac{1}{\sqrt{1+\tilde{s}^{2q}}} d \tilde{s},
	\]
	and, recalling the expansion \eqref{eq:1sqrt1xexpansion}, it follows as in the proof of Proposition \ref{prop:hanalytic} that $G_q$ admits the expansion, for small $s$,
	\[
		G_q(s)
		=
		X_q
		+ 
		\sum_{n=0}^{\infty} c_n s^{1+(2n-1)q},
		\qquad
		\text{ for }
		\vert s \vert <1,
	\]
	with $c_n$ defined by \eqref{eq:defofcn}.  Thus
	\[
		H_q(s^{2}) = X_q s^{\frac{1-2q}{q}} + \sum_{n=0}^{\infty} c_n s^{-(2n+1)},
		\qquad
		\text{ for }
		\vert s \vert <1,
	\]
	and the proof for $H_q$ follows.  Differentiating gives
	\[
		2 s H_q' (s^{2}) = \frac{1-2q}{q} X_q s^{\frac{1-3q}{q}} - \sum_{n=0}^{\infty} (2n+1) c_n s^{-(2n+2)},
		\qquad
		\text{ for }
		\vert s \vert <1,
	\]
	and the proof for $2H_q'$ also follows.
\end{proof}

The following proposition contains various functions, defined in terms of $H_q$, which also appear in the proof of Theorem \ref{thm:main1}.

\begin{proposition}[Analyticity of $1/H_q$, $B_q/H_q$, and $\Phi_q$] \label{prop:analyticdivision}
	There exists $\delta>0$ such that the functions
	\[
		\frac{1}{H_q}, \, \frac{1}{H_q'}, \, \frac{B_q}{H_q}, \, \Phi_q \colon [0,\delta) \to \mathbb{R},
	\]
	where the latter two are defined by
	\begin{equation} \label{eq:analyticdivision}
		B_q(s) = \frac{2H_q'(s)}{2 s H_q'(s) + H_q(s)},
		\qquad
		\Phi_q(s)
		=
		-
		\frac{2}{H_q(s)}
		\Big(
		\frac{H_q'(s)}{H_q(s)}
		-
		s B_q(s) \frac{H_q'(s)}{H_q(s)}
		+
		s B_q'(s)
		+
		2 B_q(s)
		\Big),
	\end{equation}
	are well defined and real analytic (or, rather, extend to an real analytic functions on $(-\delta,\delta)$).
\end{proposition}

\begin{proof}
	Note that $H_q'$ is real analytic and
	\[
		H_q'(s) = \sum_{n=0}^{\infty} (n+1)b_{n+1}s^n.
	\]
	Since $H_q(0) = \frac{1}{1-2q}$ and $H_q'(0) = b_1 = - \frac{1}{2(1-4q)}$, it follows that there exists $\delta>0$ such that $H_q(s) \neq 0$ and $H_q'(s) \neq 0$ for all $(-\delta,\delta)$.  The fact that $\frac{1}{H_q}$, $\frac{1}{H_q'}$, $\frac{B_q}{H_q}$, and $\Phi_q$ are real analytic on $(-\delta,\delta)$ then follows from standard theory of analytic functions.  Moreover, the coefficients of the corresponding power series expansions around $0$ can be expressed in terms of $\{b_n\}$.  See, for example, the textbook of Krantz--Parks \cite{KraPar}.
\end{proof}

\subsection{Conservation laws}
\label{subsec:conservation}

In this section some conservation laws satisfied by solutions of \eqref{eq:Vlasovgeneral} are given.
The first conservation law concerns the spatial average of $\rho$.

\begin{proposition}[Conservation law for spatial average of $\rho$] \label{prop:conservation}
	For any $q>0$, if $f$ solves \eqref{eq:Vlasovgeneral}, then the average of $\rho$ satisfies the conservation law
	\[
		t^{6q} \int_{\mathbb{T}^3} \rho(t,x) dx = \int_{\mathbb{T}^3} \int_{\mathbb{R}^3} f_1(x,p) dp dx,
	\]
	for all $t \in (0,\infty)$, where $f_1(x,p) = f(1,x,p)$.
\end{proposition}

\begin{proof}
	After integrating equation \eqref{eq:Vlasovgeneral} with respect to both $x$ and $p$, the second term on the left hand side vanishes and the third term can be integrated by parts in $p$ to give
	\[
		\partial_t \int_{\mathbb{T}^3} \int_{\mathbb{R}^3} f(t,x,p) dp dx
		+
		\frac{6q}{t}
		\int_{\mathbb{T}^3} \int_{\mathbb{R}^3} f(t,x,p) dp dx
		=
		0,
	\]
	from which the result follows.
\end{proof}

It follows from Proposition \ref{prop:conservation} that, for any solution $f$ of \eqref{eq:Vlasovgeneral},
\[
	\overline{\rho}(t)
	=
	\frac{1}{t^{6q}} \int_{\mathbb{T}^3} \int_{\mathbb{R}^3} f_1(x,p) dp dx.
\]

The next conservation law concerns $p$-weighted $L^2$ norms of solutions.

\begin{proposition}[Weighted $L^2$ conservation law for $f$] \label{prop:conservationL2}
	For any $q>0$, if $f$ solves \eqref{eq:Vlasovgeneral}, then, for any $s\geq 0$,
	\[
		t^{2q(s+3)} \int_{\mathbb{T}^3} \int_{\mathbb{R}^3} \vert p \vert^s \vert f(t,x,p) \vert^2 dp dx
		=
		\int_{\mathbb{T}^3} \int_{\mathbb{R}^3} \vert p \vert^s \vert f_1(x,p) \vert^2 dp dx,
	\]
	for all $t \in (0,\infty)$, where $f_1(x,p) = f(1,x,p)$.
\end{proposition}

\begin{proof}
	A simple computation gives
	\begin{multline*}
		\partial_t \int_{\mathbb{T}^3} \int_{\mathbb{R}^3} \vert p \vert^s \vert f(t,x,p) \vert^2 dp dx
		=
		-2 \int_{\mathbb{T}^3} \int_{\mathbb{R}^3} \vert p \vert^s f \Big( \frac{p^i}{p^0} \partial_{x^i} f - \frac{2q}{t} p^i \partial_{p^i} f \Big) dp dx
		\\
		=
		- \frac{2q(s+3)}{t} \int_{\mathbb{T}^3} \int_{\mathbb{R}^3} \vert p \vert^s \vert f(t,x,p) \vert^2 dp dx,
	\end{multline*}
	from which the result immediately follows.
\end{proof}

\section{The proof of the main results}
\label{section:proof}

In this section the proofs of Theorem \ref{thm:main1} and Theorem \ref{thm:main2} are given.  Many parts of the proof of Theorem \ref{thm:main1} can be made more explicit when $q=\frac{1}{3}$.  A more explicit proof of this special case is therefore first given in Section \ref{subsec:sfproof}.  In Section \ref{subsec:radiationproof} the proof of Theorem \ref{thm:main2} is given.  Finally, in Section \ref{subsec:generalproof} the proof of Theorem \ref{thm:main1} in full generality is given.  

\subsection{The $q=\frac{1}{3}$ FLRW spacetime}
\label{subsec:sfproof}

When $q=\frac{1}{3}$, the Vlasov equation \eqref{eq:Vlasovgeneral} takes the form
\begin{equation} \label{eq:Vlasovsf}
	\partial_t f
	+
	\frac{p^i}{p^0} \partial_{x^i} f
	-
	\frac{2}{3t} p^{i} \partial_{p^i} f
	=
	0,
	\qquad
	p^0 = \sqrt{1+ t^{\frac{2}{3}} \vert p \vert^2}.
\end{equation}
This section concerns the following version of Theorem \ref{thm:main1} for equation \eqref{eq:Vlasovsf}.

\begin{theorem}[Phase mixing for the Vlasov equation on the $q=\frac{1}{3}$ FLRW spacetime] \label{thm:mainsf}
	Let $f$ be a solution of equation \eqref{eq:Vlasovsf} on $[1,\infty) \times \mathbb{T}^3\times \mathbb{R}^3$ such that $f(1,x,p) = f_1(x,p)$ for some $f_1 \in H^k_{1/3}(\mathbb{T}^3 \times \mathbb{R}^3)$ for some $k \geq 2$.  The spatial density satisfies, for all $t\geq 1$,
	\begin{equation} \label{eq:mainsf2}
		\sup_{x\in \mathbb{T}^3} \big\vert \rho(t,x) - \overline{\rho}(t) \big\vert
		\lesssim
		\frac{\Vert f_1 \Vert_{H^k_{1/3}}}{t^{2 + \frac{k}{3}}}.
	\end{equation}
	If $f_1$ lies in the analytic space $f_1 \in H^{\omega}_{1/3}(\mathbb{T}^3 \times \mathbb{R}^3)$ then the spatial density satisfies, for all $t\geq 1$,
	\begin{equation} \label{eq:mainsf3}
		\sup_{x \in \mathbb{T}^3}\big\vert \rho(t,x) - \overline{\rho}(t) \big\vert
		\lesssim
		\frac{\Vert f_1 \Vert_{H^{\omega}_{1/3}}}{t^{2}}
		e^{- \mu \, t^{\frac{1}{456}}},
	\end{equation}
	where $\mu = \big( \lambda(f_1) \big)^{\frac{1}{152}}/2$.
\end{theorem}

\subsubsection{Commutation vector fields}

The proof of Theorem \ref{thm:mainsf} is based on the following set of vector fields which commute with equation \eqref{eq:Vlasovsf}.
For $k =1,2,3$, define
\begin{equation} \label{eq:sfvectorfields1}
	L_k = {A_k}^i(t,p) \partial_{x^i} + \frac{1}{t^{\frac{2}{3}}} \partial_{p^k},
\end{equation}
where
\begin{equation} \label{eq:sfvectorfields2}
	{A_k}^i(t,p)
	=
	3\sqrt{t^{\frac{2}{3}}+\vert t^{\frac{2}{3}} p \vert^2}
	\bigg(
	\delta_k^i
	+
	\frac{t^{\frac{4}{3}} p^i p^k}{t^{\frac{2}{3}}+\vert t^{\frac{2}{3}} p \vert^2}
	\bigg).
\end{equation}

\begin{proposition}[Commuting vector fields for equation \eqref{eq:Vlasovsf}] \label{prop:sfcommutators}
	The vector fields $L_k$, for $k=1,2,3$, defined by \eqref{eq:sfvectorfields1}--\eqref{eq:sfvectorfields2}, satisfy
	\[
		\Big[
		\partial_t
		+
		\frac{p^i}{\sqrt{1+ t^{\frac{2}{3}} \vert p \vert^2}} \partial_{x^i}
		-
		\frac{2}{3t} p^{i} \partial_{p^i}
		,
		L_k
		\Big]
		=
		0, \qquad k=1,2,3.
	\]
\end{proposition}

\begin{proof}
	The proof is a direct computation.
\end{proof}

\begin{proposition}[Inverse of the matrix $A$] \label{prop:sfmatrixA}
	For all $t \geq 1$ and all $p \in \mathbb{R}^3$, the matrix ${A_k}^i(t,p)$, defined by \eqref{eq:sfvectorfields2}, is invertible.  The inverse takes the form
	\begin{align} \label{eq:sfinverseofA}
		{(A^{-1})_k}^i(t,p)
		=
		\
		&
		\frac{
		1
		}{
		3\sqrt{t^{\frac{2}{3}}+\vert t^{\frac{2}{3}} p \vert^2}
		}
		\bigg(
		\delta_k^i
		-
		\frac{
		t^{\frac{2}{3}}p^i t^{\frac{2}{3}}p^k
		}{
		t^{\frac{2}{3}}+2\vert t^{\frac{2}{3}} p \vert^2
		}
		\bigg),
	\end{align}
\end{proposition}

\begin{proof}
	The proof is a direct computation.
\end{proof}

\begin{remark}[Representation formula] \label{rmk:repsf}
	Equation \eqref{eq:Vlasovsf} admits the following representation formula for solutions.  For any solution, there is a function $F \colon \mathbb{R}^3 \times \mathbb{R}^3 \to [0,\infty)$ such that
	\[
		f(t,x,p)
		=
		F \Big(x - 3 (t^{\frac{2}{3}}+\vert t^{\frac{2}{3}} p \vert^2)^{\frac{1}{2}} t^{\frac{2}{3}} p, t^{\frac{2}{3}} p \Big).
	\]
	It follows that
	\begin{equation} \label{eq:sfrep}
		F(x,p)
		=
		f\big(t, x + 3(t^{\frac{2}{3}}+\vert p \vert^2)^{\frac{1}{2}} p, t^{-\frac{2}{3}} p \big).
	\end{equation}
	The expression for the vector fields \eqref{eq:sfvectorfields1}--\eqref{eq:sfvectorfields2} can be obtained by applying $\partial_{p^k}$ to \eqref{eq:sfrep}.
	
	There is a related representation formula for equation \eqref{eq:Vlasovsf},
	\[
		f(t,x,p)
		=
		f_1\Big(x - 3\big( (t^{\frac{2}{3}}+\vert t^{\frac{2}{3}} p \vert^2)^{\frac{1}{2}} - (1+\vert t^{\frac{2}{3}} p \vert^2 )^{\frac{1}{2}} \big) t^{\frac{2}{3}} p, t^{\frac{2}{3}} p \Big).
	\]
	where $f_1(x,p) = f(1,x,p)$.  It follows that
	\begin{equation} \label{eq:sfrep2}
		f_1(x,p)
		=
		f\Big(t, x + 3\big( (t^{\frac{2}{3}}+\vert p \vert^2)^{\frac{1}{2}} - (1+\vert p \vert^2)^{\frac{1}{2}} \big) p, t^{-\frac{2}{3}} p \Big).
	\end{equation}
	Applying $\partial_{p^k}$ to \eqref{eq:sfrep2} gives rise to a related set of vector fields, again of the form \eqref{eq:sfvectorfields1} but now with
	\[
		{A_k}^i(t,p)
		=
		3\bigg( \sqrt{t^{\frac{2}{3}}+\vert t^{\frac{2}{3}} p \vert^2} - \sqrt{1+\vert t^{\frac{2}{3}} p \vert^2} \bigg)
		\Bigg(
		\delta_k^i
		-
		\frac{t^{\frac{4}{3}} p^i p^k}{\sqrt{(t^{\frac{2}{3}}+\vert t^{\frac{2}{3}} p \vert^2)(1+\vert t^{\frac{2}{3}} p \vert^2)}}
		\Bigg).
	\]
	This related collection of vector fields will not be used in the present work, but will be commented on again in Remark \ref{rmk:repgeneral}.
\end{remark}

\subsubsection{A binomial theorem for the commutation vector fields}
\label{subsec:sfbinom}

In this section a Binomial Theorem-type result is given for certain operators related to the commutation vector fields which arise when relating derivatives of $\rho$ to vector field derivatives of $f$.

To ease notation, define
\begin{equation} \label{eq:sfnotation}
	r^k = t^{\frac{2}{3}} p^k,
	\qquad
	\partial_{r^k} = \frac{1}{t^{\frac{2}{3}}} \partial_{p^k},
	\qquad
	h(t,r) = 3\sqrt{t^{\frac{2}{3}}+\vert r \vert^2},
	\qquad
	b(t,r) = t^{\frac{2}{3}}+2\vert r \vert^2,
\end{equation}
so that
\begin{equation} \label{eq:sfbinom1}
	L_k
	=
	h(t,r)
	\bigg(
	\delta_k^i
	+
	\frac{9 r^i r^k}{h(t,r)^2}
	\bigg)
	\partial_{x^i}
	+
	\partial_{r^k},
\qquad
	{(A^{-1})_k}^i(t,p)
	=
	\frac{
	1
	}{
	h(t,r)
	}
	\Big(
	\delta_k^i
	-
	\frac{r^i r^k}{b(t,r)}
	\Big),
	\qquad
	i,k=1,2,3.
\end{equation}
Moreover, a computation gives
\begin{align} \label{eq:sfbinom2}
	\frac{1}{t^{\frac{1}{3}}} \partial_{p^i} {(A^{-1})_k}^i(t,p)
	=
	\frac{4 t^{\frac{1}{3}} \vert r \vert^2 r^k}{h(t,r) b(t,r)^2}
	-
	\frac{5 t^{\frac{1}{3}} \,r^k}{h(t,r) b(t,r)}.
\end{align}
Define the operators
\begin{equation} \label{eq:sfdefofM}
	M_i
	=
	t^{\frac{1}{3}} {(A^{-1})_i}^k(t,p) L_k
	=
	\frac{
	t^{\frac{1}{3}}
	}{
	h(t,r)
	}
	\Big(
	\delta_k^i
	-
	\frac{r^i r^k}{b(t,r)}
	\Big)
	L_k,
	\qquad
	i=1,2,3.
\end{equation}

As will be seen in the proof of Proposition \ref{prop:rhovectorssf} below, for any multi-index $I$,
\begin{equation} \label{eq:sfaswillbeseen}
	(t^{\frac{1}{3}} \partial_x)^I \rho(t,x)
	=
	\int
	\bigg(
	M
	+
	\frac{4 t^{\frac{1}{3}} \vert r \vert^2 \, r}{h(t,r) b(t,r)^2}
	-
	\frac{5 t^{\frac{1}{3}} \,r}{h(t,r) b(t,r)}
	\bigg)^I f(t,x,p) dp.
\end{equation}
The following proposition is a Binomial Theorem-type result which relates the operator appearing on the right side of \eqref{eq:sfaswillbeseen} to combinations of the commutation vector fields.  See \eqref{eq:sfcomb}.  In the proof of Proposition \ref{prop:rhovectorssf}, it is important to capture the fact that the coefficient of each such term is uniformly bounded in $t$, and also to keep track of the $r$ weight of each such term.  The fact that the non-vanishing terms in \eqref{eq:sfcomb} satisfy the property \eqref{eq:sfcomb4} will be used to establish these facts.  In order to treat the case of analytic $f_1$ in Theorem \ref{thm:mainsf}, it is also important to characterise the behaviour of the constants in such estimates as the number of derivatives $\vert I\vert \to \infty$.  The property \eqref{eq:sfcomb2} is used to estimate these constant for each such term, and the property \eqref{eq:sfcomb3} is used to estimate the number of terms.

The notation is somewhat easier to navigate for the corresponding statement in one spatial dimension.  The reader is therefore encouraged to first read the treatment of the relativistic free transport equation in one spatial dimension in Section \ref{subsec:relft}.

\begin{proposition}[A binomial theorem for the commutation vector fields] \label{prop:sfcomb}
	For each multi-index $\vert I \vert \geq 1$, there are integers $C^I_{JKlmn} \in \mathbb{Z}$, for $l,m,n\in \mathbb{N}_0$, multi-indices $K$, and $J\in (\mathbb{N}_0)^3$, such that
	\begin{equation} \label{eq:sfcomb}
		\bigg(
		M
		+
		\frac{4 t^{\frac{1}{3}} \vert r \vert^2 r^k}{h(t,r) b(t,r)^2}
		-
		\frac{5 t^{\frac{1}{3}} \,r^k}{h(t,r) b(t,r)}
		\bigg)^I
		=
		\sum_{J,K,l,m,n} C^I_{JKlmn} t^{\frac{l}{3}} h(t,r)^{-m} b(t,r)^{-n} r^J L^K.
	\end{equation}
	Moreover
	\begin{itemize}
		\item
			Each $C^I_{JKlmn}$ satisfies
			\begin{equation} \label{eq:sfcomb2}
				\vert C^I_{JKlmn} \vert \lesssim (146 \vert I \vert)!
				\qquad
				\text{for all}
				\quad
				l,m,n\in \mathbb{N}_0,
				\quad
				J \in (\mathbb{N}_0)^3,
				\quad
				\text{and multi-indices } K.
			\end{equation}
		\item
			The non-vanishing $C^I_{JKlmn}$ satisfy
			\begin{equation} \label{eq:sfcomb3}
				C^I_{JKlmn} \neq 0
				\qquad
				\Rightarrow
				\quad
				0 \leq \vert J \vert \leq 3\vert I \vert,
				\quad
				0 \leq \vert K \vert \leq \vert I \vert,
				\quad
				0 \leq l \leq \vert I \vert,
				\quad
				0 \leq m \leq 3\vert I \vert,
				\quad
				0 \leq n \leq 2\vert I \vert.
			\end{equation}
		\item
			The non-vanishing $C^I_{JKlmn}$ moreover satisfy
			\begin{equation} \label{eq:sfcomb4}
				C^I_{JKlmn} \neq 0
				\qquad
				\Rightarrow
				\qquad
				0 \leq l + \vert J \vert \leq m + 2n.
			\end{equation}
	\end{itemize}
\end{proposition}

\begin{proof}
	Clearly \eqref{eq:sfcomb} holds when $\vert I \vert=1$ with, for $e_1=(1,0,0)$, $e_2=(0,1,0)$, $e_3=(0,0,1)$,
	\[
		C^{e_i}_{0,e_k,1,1,0} = \delta_{ik},
		\qquad
		C^{e_i}_{e_{j_1} + e_{j_2},e_k,1,1,1} = - \delta^i_{j_1} \delta^k_{j_2},
		\qquad
		C^{e_i}_{2e_c+e_j,0,1,1,2} = 4 \delta_{ij},
		\qquad
		C^{e_i}_{e_j,0,1,1,1} = -5 \delta_{ij},
	\]
	where the third holds for $c=1,2,3$, and
	\[
		C^{e_i}_{JKlmn} = 0
		\quad
		\text{otherwise},
	\]
	for each $i = 1,2,3$.
	
	Note that, for $J = (j_1,j_2,j_3)$,
	\[
		M_i(h(t,r)^{-m})
		=
		\frac{-9\,m\, t^{\frac{1}{3}} r^k
		}{
		h(t,r)^{m+3}
		}
		\Big(
		\delta^k_i
		-
		\frac{r^i r^k}{b(t,r)}
		\Big)
		,
		\qquad
		M_i(b(t,r)^{-n})
		=
		\frac{-4\,n\, t^{\frac{1}{3}} r^k
		}{
		h(t,r) b(t,r)^{n+1}
		}
		\Big(
		\delta^k_i
		-
		\frac{r^i r^k}{b(t,r)}
		\Big)
		,
	\]
	\[
		M_i(r^J)
		=
		\frac{
		j_i
		t^{\frac{1}{3}}
		r^{J-e_i}
		}{
		h(t,r)
		}
		-
		\frac{
		\vert J \vert t^{\frac{1}{3}} r^{J+e_{i}}
		}{
		h(t,r) b(t,r)
		}
		.
	\]
	Suppose now that \eqref{eq:sfcomb} holds from some $I$.  It follows that
	\begin{align*}
		&
		\bigg(
		M_i + \frac{4 t^{\frac{1}{3}} \vert r \vert^2 r^k}{h(t,r) b(t,r)^2}
		-
		\frac{5 t^{\frac{1}{3}} \,r^k}{h(t,r) b(t,r)}
		\bigg)
		\sum_{J,K,l,m,n} C^I_{JKlmn} t^{\frac{l}{3}} h(t,r)^{-m} b(t,r)^{-n} r^J L^K
		\\
		&
		=
		\sum_{J,K,l,m,n} C^I_{JKlmn}
		\bigg(
		-9m t^{\frac{l+1}{3}} h^{-(m+3)} b^{-n} r^{J+e_i} L^K
		+
		9m t^{\frac{l+1}{3}} h^{-(m+3)} b^{-(n+1)} \sum_{c=1}^3 r^{J+e_i+2e_c} L^K
		\\
		&
		\quad
		-
		4n t^{\frac{l+1}{3}} h^{-(m+1)} b^{-(n+1)} r^{J+e_i} L^K
		+
		4n t^{\frac{l+1}{3}} h^{-(m+1)} b^{-(n+1)} \sum_{c=1}^3 r^{J+e_i+2e_c} L^K
		+
		j_i t^{\frac{l+1}{3}} h^{-(m+1)} b^{-n} r^{J-e_i} L^K
		\\
		&
		\quad
		-
		\vert J\vert t^{\frac{l+1}{3}} h^{-(m+1)} b^{-(n+1)} r^{J+e_i} L^K
		+
		t^{\frac{l+1}{3}} h^{-(m+1)} b^{-n} r^{J} L_i L^K
		-
		t^{\frac{l+1}{3}} h^{-(m+1)} b^{-(n+1)} \sum_{c=1}^3 r^{J+e_i+e_c} L_c L^K
		\\
		&
		\quad
		+
		4\sum_{c=1}^3
		t^{\frac{l+1}{3}} h^{-(m+1)} b^{-(n+2)} r^{J+2e_c+e_i} L^K
		-
		5 t^{\frac{l+1}{3}} h^{-(m+1)} b^{-(n+1)} r^{J+e_i} L^K
		\bigg),
	\end{align*}
	and so \eqref{eq:sfcomb} holds for $I+e_i$ with, for each $l,m,n\in \mathbb{N}_0$, multi-index $K$, and $J\in (\mathbb{N}_0)^3$,
	\begin{align} \label{eq:sfcombrecur}
		&
		C^{I+e_i}_{JKlmn}
		=
		-
		9(m-3) C^I_{J-e_i,K,l-1,m-3,n}
		+
		9(m-3) \sum_{c=1}^3 C^I_{J-e_i-2e_c,K,l-1,m-3,n-1}
		\\
		&
		\nonumber
		\qquad
		-
		4(n-1) C^I_{J-e_i,K,l-1,m-1,n-1}
		+
		4(n-1)\sum_{c=1}^3 C^I_{J-e_i-2e_c,K,l-1,m-1,n-2}
		+
		(j_i+1) C^I_{J+e_i,K,l-1,m-1,n}
		\\
		&
		\nonumber
		\qquad
		-
		(\vert J \vert - 1) C^I_{J-e_i,K,l-1,m-1,n-1}
		+
		C^I_{J,K-e_i,l-1,m-1,n}
		-
		\sum_{c=1}^3 C^I_{J-e_i-e_c,K-e_c,l-1,m-1,n-1}
		\\
		&
		\nonumber
		\qquad
		+
		4\sum_{c=1}^3
		C^I_{J-2e_c-e_i,K,l-1,m-1,n-2}
		-
		5 C^I_{J-e_i,K,l-1,m-1,n-1}
		,
	\end{align}
	where $C^{I}_{JKlmn} = 0$ if $j_1<0$, $j_2<0$, $j_3<0$, $K = \emptyset$, $l<0$, $m<0$ or $n<0$ (using the multi-index conventions of Section \ref{subsec:multiindices}).
	
	Consider now the property \eqref{eq:sfcomb3}.  Clearly \eqref{eq:sfcomb3} holds for $\vert I\vert=1$.  Suppose now that \eqref{eq:sfcomb3} holds for some $\vert I\vert\geq 1$.  If $C^{I+e_i}_{JKlmn} \neq 0$ then at least one of the terms of the right hand side of \eqref{eq:sfcombrecur} must be non-vanishing.  By the inductive hypothesis
	\[
		\vert J\vert-3 \leq 3\vert I\vert,
		\quad
		\vert K\vert-1\leq \vert I\vert,
		\quad
		l-1 \leq \vert I\vert,
		\quad
		m-3 \leq 3\vert I\vert,
		\quad
		\text{and}
		\quad
		n-2 \leq 2\vert I\vert,
	\]
	i.\@e.\@ \eqref{eq:sfcomb3} holds for each $\vert I\vert+1$.
	
	Note now that \eqref{eq:sfcomb2} holds for $\vert I\vert=1$.  Suppose \eqref{eq:sfcomb2} holds for some $\vert I \vert \geq1$.  It follows from \eqref{eq:sfcombrecur}, and the property \eqref{eq:sfcomb3}, that
	\[
		\vert C^{I+e_i}_{JKlmn} \vert
		\leq
		(146 \vert I \vert + 13) (146\vert I \vert)!
		\leq
		(146(\vert I \vert+1))!,
	\]
	i.\@e.\@ \eqref{eq:sfcomb2} holds for $\vert I \vert+1$.
	
	Finally, the property \eqref{eq:sfcomb4} follows from a similar induction argument using the relation \eqref{eq:sfcombrecur}.  Indeed, the \eqref{eq:sfcomb4} clearly holds for $\vert I \vert =1$.  Suppose \eqref{eq:sfcomb4} holds for some $\vert I\vert\geq 1$.  If $C^{I+e_i}_{JKlmn} \neq 0$ then at least one of the terms of the right hand side of \eqref{eq:sfcombrecur} must be non-vanishing.  If the first term on the right hand side of \eqref{eq:sfcombrecur}, namely $-9(m-3) C^I_{J-e_i,K,l-1,m-3,n}$, is non-vanishing then, by the inductive hypothesis
	\[
		l - 1 + \vert J \vert - 1 \leq m-3 + 2n.
	\]
	In particular, $J$, $l$, $m$, $n$ satisfy \eqref{eq:sfcomb4}.  It is similarly verified that $J$, $l$, $m$, $n$ satisfy \eqref{eq:sfcomb4} if any of the other terms on the right hand side of \eqref{eq:sfcombrecur} are non-vanishing, and thus \eqref{eq:sfcomb4} holds for each $\vert I \vert +1$.
	
\end{proof}

The following proposition is shown in a similar way to Proposition \ref{prop:sfcomb}, and is used to relate combinations of vector fields at $t=1$ to the standard $\partial_x$ and $\partial_p$ derivatives appearing in the spaces $H^k_{1/3}$.

\begin{proposition}[Vector fields at $t=1$] \label{prop:sfcomb2}
	For any $k \geq 0$,
	\[
		\sum_{\vert I \vert \leq k} \big\vert L^I(f(t,x,p)) \vert_{t=1} \big\vert
		\lesssim
		(5k)!
		\sum_{\vert I \vert + \vert J \vert \leq k}
		(1+\vert p \vert)^{\vert I \vert}
		\big\vert \partial_x^I \partial_p^Jf_1(x,p) \big\vert.
	\]
\end{proposition}

\begin{proof}
	Note first that, for any multi-index $I$, there are constants $C^I_{JKLm} \in \mathbb{Z}$ such that
	\begin{equation} \label{eq:sft1comb}
		L^I\vert_{t=1}
		=
		\sum
		C^I_{JKLm} (1+\vert p \vert^2)^{\frac{m}{2}} p^J \partial_x^K \partial_p^L.
	\end{equation}
	(Note that $m$ is now indexed by $\mathbb{Z}$ and so can take both positive and negative values.)  Indeed, \eqref{eq:sft1comb} clearly holds for $I = e_k$ with constants
	\[
		C^{e_k}_{0,e_i,0,1} = 3 \delta^k_i,
		\qquad
		C^{e_k}_{e_{j_1}+e_{j_2},e_i,0,-1} = 3 \delta^k_{j_1} \delta^i_{j_2},
		\qquad
		C^{e_k}_{0,0,e_i,0} = \delta^k_i.
	\]
	The fact that \eqref{eq:sft1comb} holds in general can be established by an induction argument, similar to that of the proof of Proposition \ref{prop:sfcomb}, by applying $L_i$ to \eqref{eq:sft1comb} and checking that the form is preserved.  Moreover, the recursion relation
	\begin{align*}
		C^{I+e_i}_{JKLm}
		=
		\
		&
		3 C^I_{J,K-e_i,L,m-1}
		+
		3 \sum_{k=1}^3 C^I_{J-e_i-e_k,K-e_k,L,m+1}
		+
		(m+2)C^I_{J-e_i,K,L,m+2}
		\\
		&
		+
		(j_i+1) C^I_{J+e_i,K,L,m}
		+
		C^I_{J,K,L-e_i,m},
	\end{align*}
	holds, where $C^{I}_{JKlm} = 0$ if $j_1<0$, $j_2<0$, $j_3<0$, $K = \emptyset$, or $L = \emptyset$ (note though that $C^{I}_{JKlm}$ may now be non-zero for $m<0$).  As in the proof of Proposition \ref{prop:sfcomb}, it is easily inductively shown that each $C^I_{JKLm}$ satisfies
	\[
		\vert C^I_{JKLm} \vert \lesssim (5 \vert I \vert)!
		\qquad
		\text{for all}
		\quad
		m \in \mathbb{Z},
		\quad
		J \in (\mathbb{N}_0)^3,
		\quad
		\text{and multi-indices } K,L.
	\]
	and the non-vanishing $C^I_{JKLm}$ satisfy
	\[
		C^I_{JKLm} \neq 0
		\qquad
		\Rightarrow
		\quad
		0 \leq \vert J \vert \leq 2\vert I \vert,
		\quad
		0 \leq \vert K \vert + \vert L \vert \leq \vert I \vert,
		\quad
		- 2 \vert I\vert \leq m \leq \vert I \vert,
		\quad
		\vert J \vert + m \leq \vert N \vert.
	\]
	The proof then follows.
\end{proof}

\subsubsection{Derivative relations}

The proof of Theorem \ref{thm:mainsf} involves taking derivatives of $\rho$.  These derivatives can be related to combinations of the above vector fields applied to $f$.

\begin{proposition}[Derivatives of $\rho$ and vector fields] \label{prop:rhovectorssf}
	For all $k \geq 0$,
	\[
		\sum_{\vert I \vert = k} 
		\big\vert \partial_x^I \rho(t,x) \vert
		\lesssim
		\frac{k^9(146k)!}{t^{\frac{k}{3}}}
		\sum_{\vert I \vert = 0}^k
		\int_{\mathbb{R}^3}
		\vert L^I f(t,x,p) \vert
		dp
		.
	\]
\end{proposition}

\begin{proof}
	Note first that, for any function $g$ and any $i=1,2,3$,
	\begin{align*}
		t^{\frac{1}{3}} \partial_{x^i} \int g(t,x,p) dp
		&
		=
		\int
		t^{\frac{1}{3}} {(A^{-1})_i}^j L_j g(t,x,p) - t^{\frac{1}{3}} {(A^{-1})_i}^j \partial_{p^j} g(t,x,p) dp
		\\
		&
		=
		\int
		\Big(
		t^{\frac{1}{3}} {(A^{-1})_i}^j L_j + t^{\frac{1}{3}} \partial_{p^j} {(A^{-1})_i}^j
		\Big) g(t,x,p) dp.
	\end{align*}
	Recall the notation introduced in Section \ref{subsec:sfbinom} and the facts \eqref{eq:sfbinom1}--\eqref{eq:sfbinom2}.  Applying the above repeatedly, it follows that
	\[
		(t^{\frac{1}{3}} \partial_x)^I \rho(t,x)
		=
		\int
		\bigg(
		M
		+
		\frac{4 t^{\frac{1}{3}} \vert r \vert^2 \, r}{h(t,r) b(t,r)^2}
		-
		\frac{5 t^{\frac{1}{3}} \,r}{h(t,r) b(t,r)}
		\bigg)^I f(t,x,p) dp,
	\]
	where the operators $M_i$ are defined by \eqref{eq:sfdefofM}.  Proposition \ref{prop:sfcomb} implies that
	\[
		\Bigg\vert
		\bigg(
		M + \frac{4 t^{\frac{1}{3}} \vert r \vert^2 \, r}{h(t,r) b(t,r)^2}
		-
		\frac{5 t^{\frac{1}{3}} \,r}{h(t,r) b(t,r)}
		\bigg)^I f(t,x,p)
		\Bigg\vert
		\lesssim
		\sum_{J,K,l,m,n} \vert C^I_{JKlmn} \vert t^{\frac{l}{3}} h(t,r)^{-m} b(t,r)^{-n} \vert r \vert^{\vert J\vert} \vert L^K f(t,x,p) \vert.
	\]
	Recall the property \eqref{eq:sfcomb4} of Proposition \ref{prop:sfcomb} and, accordingly, suppose that $0 \leq l + \vert J \vert \leq m + 2n$.  It follows that
	\[
		t^{\frac{l}{3}} \vert r \vert^{\vert J\vert} h(t,r)^{-m} b(t,r)^{-n}
		\leq
		\frac{t^{\frac{l}{3}} \vert r \vert^{\vert J\vert}}{(t^{\frac{2}{3}}+\vert r \vert^2)^{\frac{m+2n}{2}}}
		\leq
		1.
	\]
	The property \eqref{eq:sfcomb4} of Proposition \ref{prop:sfcomb} therefore implies that
	\[
		\Bigg\vert
		\bigg(
		M
		+
		\frac{4 t^{\frac{1}{3}} \vert r \vert^2 r}{h(t,r) b(t,r)^2}
		-
		\frac{5 t^{\frac{1}{3}} \,r}{h(t,r) b(t,r)}
		\bigg)^I f(t,x,p)
		\Bigg\vert
		\lesssim
		\sum_{J,K,l,m,n} \vert C^I_{JKlmn} \vert  \vert L^K f(t,x,p) \vert.
	\]
	If $\vert I\vert = k$, the property \eqref{eq:sfcomb3} in particular implies that
	\[
		\# \{ (J,K,l,m,n) \mid C^I_{JKlmn} \neq 0 \}
		\lesssim
		(3k)^3 \cdot k^3 \cdot k \cdot (3k) \cdot (2k)
		\lesssim
		k^9,
	\]
	and so the result follows from the property \eqref{eq:sfcomb2}.
\end{proof}

\subsubsection{The proof of Theorem \ref{thm:mainsf}}

The proof of Theorem \ref{thm:mainsf} can now be given.

\begin{proof}[Proof of Theorem \ref{thm:mainsf}]
	Consider some $k \in \mathbb{N}$ and suppose $f_1\in H^k_{1/3}$.  The Sobolev inequality of Proposition \ref{prop:Sobolev} gives
	\[
		\sup_{x\in \mathbb{T}^3} \vert \rho(t,x) - \overline{\rho}(t) \vert
		\lesssim
		\frac{1}{\sqrt{k}}
		\sum_{\vert I \vert = k} \Vert \partial_x^I \rho(t,\cdot) \Vert_{L^2(\mathbb{T}^3)}.
	\]
	Proposition \ref{prop:rhovectorssf} and Proposition \ref{prop:interp} give, for all $t \geq t_0$,
	\[
		\sum_{\vert I \vert = k} \Vert \partial_x^I \rho(t,\cdot) \Vert_{L^2(\mathbb{T}^3)}
		\lesssim
		\frac{k^9(146k)!}{t^{\frac{k}{3}}}
		\sum_{\vert I \vert \leq k}
		\left( \int_{\mathbb{T}^3} \int_{\mathbb{R}^3} \vert p \vert^2 \vert L^I f(t,x,p) \vert^2 dp dx \right)^{\frac{1}{4}}
		\left( \int_{\mathbb{T}^3} \int_{\mathbb{R}^3} \vert p \vert^4 \vert L^I f(t,x,p) \vert^2 dp dx \right)^{\frac{1}{4}}
		\!\!\! .
	\]
	By Proposition \ref{prop:sfcommutators} it follows that $L^If$ solves \eqref{eq:Vlasovsf} for all multi-indices $I$, and so it follows from Proposition \ref{prop:conservationL2}, the Cauchy--Schwarz inequality, and Proposition \ref{prop:sfcomb2} that, for any $k \geq 2$,
	\begin{equation} \label{eq:mainsfestimate}
		\sup_{x\in \mathbb{T}^3} \vert \rho(t,x) - \overline{\rho}(t) \vert
		\lesssim
		\frac{k^9(146k)!(5k)!}{t^{2+\frac{k}{3}}\sqrt{k}}
		\sum_{\vert I \vert + \vert J \vert \leq k}
		\left(
		\int_{\mathbb{T}^3} \int_{\mathbb{R}^3}
		(\vert p \vert^2 + \vert p \vert^4) 
		(1+\vert p \vert)^{\vert I \vert}
		\vert \partial_x^I \partial_p^J f_1(x,p) \vert^2 dp dx
		\right)^{\frac{1}{2}}.
	\end{equation}
	The proof of \eqref{eq:mainsf2} follows.
	
	Suppose now that $f_1 \in H^{\omega}_{1/3}$.  Note that
	\[
		\sum_{\vert I \vert + \vert J\vert \leq k}
		\left(
		\int_{\mathbb{T}^3} \int_{\mathbb{R}^3}
		(\vert p \vert^2 + \vert p \vert^4) 
		(1+\vert p \vert)^{\vert I \vert}
		\vert \partial_x^I \partial_p^J f_1(x,p) \vert^2 dp dx
		\right)^{\frac{1}{2}}
		\lesssim
		\frac{k!}{\lambda^k}
		\Vert f_1 \Vert_{H^{\omega}_{1/3}(\mathbb{T}^3\times \mathbb{R}^3)}
		,
	\]
	for all $k \geq 0$ and with $\lambda = \lambda(f_1)$.
	For all $t \geq 1$, $f$ satisfies \eqref{eq:mainsfestimate} for all $k \geq2$, and so, noting that
	\[
		k^9(146k)!(5k)!k! \lesssim 2^k(152k)!,
	\]
	it follows that
	\begin{equation} \label{eq:mainsfestimate2}
		\sup_{x\in \mathbb{T}^3} \vert \rho(t,x) - \overline{\rho}(t) \vert
		\lesssim
		\frac{1}{t^2} \frac{2^k (152k)!}{t^{\frac{k}{3}}\lambda^k\sqrt{k}} \Vert f_1 \Vert_{H^{\omega}_{1/3}}
		\lesssim
		\frac{1}{t^2} \frac{2^k(152k)!}{(\lambda t^{\frac{1}{3}})^k \sqrt{152k}} \Vert f_1 \Vert_{H^{\omega}_{1/3}}
		=
		\frac{1}{t^2}
		\frac{(152k)!}{(\mu t^{\frac{1}{456}})^{152k}\sqrt{152k}} \Vert f_1 \Vert_{H^{\omega}_{1/3}},
	\end{equation}
	for all $k \geq 2$, where $\mu = \lambda^{ \frac{1}{152}}/2$.  In particular, for $\hat{t}(t):=\lfloor \frac{1}{152} \lambda^{ \frac{1}{152}} t^{\frac{1}{456}} \rfloor$, where $\lfloor \cdot \rfloor$ is the floor function, \eqref{eq:mainsfestimate2} holds for $k=\hat{t}(t)$ and so (since $n! \,e^n \lesssim \sqrt{n} n^n$ for all $n$ --- see Proposition \ref{prop:Stirling} --- and $e^{-\lfloor s \rfloor} \leq e e^{-s}$ for all $s\geq 1$),
	\[
		\sup_{x\in \mathbb{T}^3} \vert \rho(t,x) - \overline{\rho}(t) \vert
		\lesssim
		\frac{1}{t^2} \frac{(152 \, \hat{t}(t))!}{(152 \, \hat{t}(t))^{152 \, \hat{t}(t)+\frac{1}{2}}} \Vert f_1 \Vert_{H^{\omega}_{1/3}}
		\lesssim
		\frac{1}{t^2} e^{-152 \, \hat{t}(t)} \Vert f_1 \Vert_{H^{\omega}_{1/3}}
		\lesssim
		\frac{1}{t^2} e^{- \mu t^{\frac{1}{456}}} \Vert f_1 \Vert_{H^{\omega}_{1/3}},
	\]
	which concludes the proof of \eqref{eq:mainsf2}.
\end{proof}

%

\subsection{The radiation $q=\frac{1}{2}$ FLRW spacetime}
\label{subsec:radiationproof}

In this section the proof of Theorem \ref{thm:main2} is given.  When $q=\frac{1}{2}$, the Vlasov equation \eqref{eq:Vlasovgeneral} takes the form
\begin{equation} \label{eq:Vlasovradiation}
	\partial_t f
	+
	\frac{p^i}{p^0} \partial_{x^i} f
	-
	\frac{1}{t} p^{i} \partial_{p^i} f
	=
	0,
	\qquad
	p^0 = \sqrt{1+ t \vert p \vert^2}.
\end{equation}

\subsubsection{Commutation vector fields}

The main difference between the proof of Theorem \ref{thm:main2} and that of Theorem \ref{thm:mainsf} is the form that the commuting vector fields take.  For $k=1,2,3$, define vector fields
\begin{equation} \label{eq:radiationvectorfields1}
	L_k = {A_k}^i(t,p) \partial_{x^i} + \frac{1}{t} \partial_{p^k},
\end{equation}
where
\begin{equation} \label{eq:radiationvectorfields2}
	{A_k}^i(t,p)
	=
	2 \log \Big(
	t^{\frac{1}{2}} + \sqrt{t + \vert t p \vert^2}
	\Big) \delta_k^i
	+
	\frac{2
	t^2 p^ip^k
	}{
	(t^{\frac{1}{2}}+ \sqrt{t+\vert t p \vert^2})\sqrt{t+\vert t p \vert^2}
	}.
\end{equation}

\begin{proposition}[Commuting vector fields for equation \eqref{eq:Vlasovradiation}] \label{prop:radiationcommutators}
	The vector fields $L_k$, for $k=1,2,3$, defined by \eqref{eq:radiationvectorfields1}--\eqref{eq:radiationvectorfields2}, satisfy
	\[
		\Big[
		\partial_t
		+
		\frac{p^i}{\sqrt{1+ t \vert p \vert^2}} \partial_{x^i}
		-
		\frac{1}{t} p^{i} \partial_{p^i}
		,
		L_k
		\Big]
		=
		0, \qquad k=1,2,3.
	\]
\end{proposition}

\begin{proof}
	The proof is a direct computation.
\end{proof}

\begin{proposition}[Inverse of the matrix $A$] \label{prop:matrixAradiation}
	For all $t \geq 1$ and $p \in \mathbb{R}^3$, the matrix ${A_k}^i(t,p)$, defined by \eqref{eq:radiationvectorfields2}, is invertible.  The inverse takes the form
	\begin{align*}
		{(A^{-1})_k}^i(t,p)
		=
		\
		&
		\frac{1}
		{
		2 \log (
		t^{\frac{1}{2}} + \sqrt{t + \vert t p \vert^2}
		)
		}
		\bigg(
		\delta_k^i
		-
		\frac{
		t^2 p^ip^k
		}{
		\vert t p\vert^2
		+
		(t^{\frac{1}{2}}+ \sqrt{t+\vert t p \vert^2})\sqrt{t+\vert t p \vert^2} \, 2 \log ( t^{\frac{1}{2}} + \sqrt{t + \vert t p \vert^2} )
		}
		\bigg),
	\end{align*}
	for $i,k=1,2,3$.
\end{proposition}

\begin{proof}
	The proof is a direct computation.
\end{proof}

\begin{remark}[Representation formula for equation \eqref{eq:Vlasovradiation}] \label{rmk:radiationrep}
	The expressions for the vector fields of Proposition \ref{prop:radiationcommutators} can be obtained from the following representation formula for solutions of equation \eqref{eq:Vlasovradiation}:
	\[
		f(t,x,p)
		= 
		F\Big(x
		-
		t p \,
		2
		\log
		\Big(
		t^{\frac{1}{2}} + \sqrt{t + \vert t p \vert^2}
		\Big),
		t p
		\Big).
	\]
	The representation formula for the solution in terms of the initial data $f_1$ takes the form
	\[
		f(t,x,p)
		= 
		f_1\Bigg(x
		-
		t p \,
		2
		\log
		\Bigg(
		\frac{t^{\frac{1}{2}} + \sqrt{t + \vert t p \vert^2}}{1 + \sqrt{1 + \vert t p \vert^2}}
		\Bigg),
		t p
		\Bigg),
	\]
	and can be used to obtain a related collection of commutation vector fields of the form \eqref{eq:radiationvectorfields1}, now with
	\[
		{A_k}^i(t,p)
		=
		\delta_k^i 2 \log \Bigg(
		\frac{t^{\frac{1}{2}} + \sqrt{t + \vert t p \vert^2}}{1 + \sqrt{1 + \vert t p \vert^2}}
		\Bigg)
		+
		\frac{
		2t^2 p^ip^k
		}{
		(t^{\frac{1}{2}}+ \sqrt{t+\vert t p \vert^2})\sqrt{1+\vert t p \vert^2}
		}
		-
		\frac{
		2t^2 p^ip^k
		}{
		(1+ \sqrt{1+\vert t p \vert^2})\sqrt{1+\vert t p \vert^2}
		}
		.
	\]
	These latter vector fields will not be used in the proof of Theorem \ref{thm:main2}.
\end{remark}

\subsubsection{A binomial theorem for the commutation vector fields}
\label{subsec:radbinom}

As in the previous section, the main step in the proof of Theorem \ref{thm:main2} is establishing the following Binomial Theorem-type result for an operator associated to the commuting vector fields.

Define
\[
	r^k = t p^k,
	\qquad
	\partial_{r^k} = \frac{1}{t} \partial_{p^k},
	\qquad
	h(t,r)
	=
	2
	\log \big(
	t^{\frac{1}{2}} + \sqrt{t + \vert r \vert^2}
	\big),
\]
\[
	b(t,r)
	=
	\vert r \vert^2
	+
	\big( t^{\frac{1}{2}}+ \sqrt{t+\vert r \vert^2} \big) \sqrt{t+\vert r \vert^2} \, \log ( t^{\frac{1}{2}} + \sqrt{t + \vert r \vert^2} ),
	\qquad
	c(t,r) = \sqrt{t + \vert r \vert^2}.
\]
It follows that
\begin{equation} \label{eq:radbinom1}
	L_k
	=
	h(t,r)
	\Big(
	\delta_k^i
	+
	\frac{r^i r^k}{b(t,r) - \vert r \vert^2}
	\Big)
	\partial_{x^i}
	+
	\partial_{r^k},
\qquad
	{(A^{-1})_k}^i(t,p)
	=
	\frac{
	1
	}{
	h(t,r)
	}
	\Big(
	\delta_k^i
	-
	\frac{r^i r^k}{b(t,r)}
	\Big),
	\qquad
	i,k=1,2,3.
\end{equation}
Moreover, a computation gives
\begin{align} \label{eq:radbinom2}
	\frac{\log (1+t)}{t} \partial_{p^i} {(A^{-1})_k}^i(t,p)
	=
	r^k
	\, \eta(t,r)
	,
\end{align}
where
\begin{equation} \label{eq:radeta}
	\eta(t,r)
	=
	\log (1+t)
	\bigg(
	\frac{3 \vert r \vert^2}{h(t,r) b(t,r)^2}
	-
	\frac{5}{h(t,r)b(t,r)}
	+
	\frac{\vert r \vert^2}{b(t,r)^2}
	+
	\frac{t^{\frac{1}{2}} \vert r \vert^2}{2b(t,r)^2c(t,r)}
	\bigg).
\end{equation}
Define the operators
\begin{equation} \label{eq:raddefofM}
	M_i
	=
	\log (1+t) {(A^{-1})_i}^k(t,p) L_k
	=
	\frac{
	\log (1+t)
	}{
	h(t,r)
	}
	\Big(
	\delta_k^i
	-
	\frac{r^i r^k}{b(t,r)}
	\Big)
	L_k,
	\qquad
	i=1,2,3.
\end{equation}

As will be seen in the proof of Proposition \ref{prop:rhovectorsrad} below, for any multi-index $I$,
\[
	\big( \log (1+t) \, \partial_x \big)^I \rho(t,x)
	=
	\int
	\big(
	M
	+
	\eta(t,r) \,
	r
	\big)^I f(t,x,p) dp.
\]
The following proposition is a Binomial Theorem-type result which relates the operator appearing on the right side to combinations of the commutation vector fields.

\begin{proposition}[A binomial theorem for the commutation vector fields] \label{prop:radcomb}
	For each multi-index $\vert I \vert \geq 1$, there are $C^I_{JKl_1l_2mn_1n_2} \in \mathbb{R}$, for $l_1,l_2,m,n_1,n_2\in \mathbb{N}_0$, multi-indices $K$, and $J\in (\mathbb{N}_0)^3$, such that
	\begin{equation} \label{eq:radcomb}
		\big(
		M
		+
		\eta(t,r) \,
		r
		\big)^I
		=
		\sum_{J,K,l_1,l_2,m,n_1,n_2}
		C^I_{JKl_1l_2mn_1n_2} \big( \log (1+t) \big)^{l_1} t^{\frac{l_2}{2}} h(t,r)^{-m} b(t,r)^{-n_1} c(t,r)^{-n_2} r^J L^K.
	\end{equation}
	Moreover
	\begin{itemize}
		\item
			Each $C^I_{JKlmn}$ satisfies
			\begin{equation} \label{eq:radcomb2}
				\vert C^I_{JKl_1l_2mn_1n_2} \vert \lesssim ( 51 \vert I \vert)!
				\qquad
				\text{for all}
				\quad
				l_1,l_2,m,n_1,n_2\in \mathbb{N}_0,
				\quad
				J \in (\mathbb{N}_0)^3,
				\quad
				\text{and multi-indices } K.
			\end{equation}
		\item
			The non-vanishing $C^I_{JKlmn}$ satisfy
			\begin{align} \label{eq:radcomb3}
				C^I_{JKl_1l_2mn_1n_2} \neq 0
				\qquad
				\Rightarrow
				\qquad
				&
				0 \leq \vert J \vert \leq 3\vert I \vert,
				\quad
				0 \leq \vert K \vert \leq \vert I \vert,
				\quad
				0 \leq l_1 \leq \vert I \vert,
				\quad
				0 \leq l_2 \leq \vert I \vert,
				\\
				&
				0 \leq m \leq \vert I \vert,
				\quad
				0 \leq n_1 \leq 2\vert I \vert,
				\quad
				0 \leq n_2 \leq 2\vert I \vert.
				\nonumber
			\end{align}
		\item
			The non-vanishing $C^I_{JKl_1l_2mn_1n_2}$ moreover satisfy
			\begin{equation} \label{eq:radcomb4}
				C^I_{JKl_1l_2mn_1n_2} \neq 0
				\qquad
				\Rightarrow
				\qquad
				l_2+\vert J \vert \leq 2n_1+n_2,
				\qquad
				l_1+l_2+\vert J \vert \leq m+2n_1+n_2.
			\end{equation}
	\end{itemize}
\end{proposition}

\begin{proof}
	Clearly \eqref{eq:radcomb} holds when $\vert I \vert=1$ with, for $e_1=(1,0,0)$, $e_2=(0,1,0)$, $e_3=(0,0,1)$,
	\[
		C^{e_i}_{0,e_k,1,0,1,0,0} = \delta^i_{k},
		\qquad
		C^{e_i}_{e_{j_1} + e_{j_2},e_k,1,0,1,1,0} = - \delta^i_{j_1} \delta^k_{j_2},
	\]
	\[
		C^{e_i}_{e_{j_1}+2e_{j_2},0,1,0,1,2,0} = 4\delta^i_{j_1},
		\quad
		C^{e_i}_{e_j,0,1,0,1,1,0} = -6 \delta^i_{j},
		\quad
		C^{e_i}_{e_{j_1}+2e_{j_2},0,1,0,0,2,0} = 2 \delta^i_{j_1},
		\quad
		C^{e_i}_{e_{j_1}+2e_{j_2},0,1,1,0,2,1} = \delta^i_{j_1},
	\]
	and
	\[
		C^{e_i}_{JKl_1l_2mn_1n_2} = 0
		\quad
		\text{otherwise},
	\]
	for each $i = 1,2,3$.
	
	Note that
	\begin{align*}
		M_i \big( h(t,r)^{-m} \big)
		&
		=
		-
		\frac{m \log (1+t) \, r^i}{b(t,r) h(t,r)^{m+1}},
		\\
		M_i \big( b(t,r)^{-n_1} \big)
		&
		=
		- n_1 \frac{\log (1+t) \big( b(t,r) - \vert r \vert^2 \big) r^i}{b(t,r)^{n_1+2}}
		\Big(
		1
		+
		\frac{3}{h(t,r)}
		+
		\frac{t^{\frac{1}{2}}}{2c(t,r)}
		\Big),
		\\
		M_i \big( c(t,r)^{-n_2} \big)
		&
		=
		n_2\frac{\log (1+t) \, \vert r\vert^2 r^i}{h(t,r) b(t,r) c(t,r)^{n_2+2}}
		-
		n_2 \frac{\log (1+t) \, r^i}{h(t,r) c(t,r)^{n_2+2}}
		\\
		M_i (r^J)
		&
		=
		j_i \frac{\log (1+t) \, r^{J-e_i}}{h(t,r)}
		-
		\vert J \vert
		\frac{\log (1+t) \, r^{J+e_i}}{h(t,r)b(t,r)}.
	\end{align*}
	Suppose now that \eqref{eq:radcomb} holds from some $I$.  Thus, it follows that
	\begin{align*}
		&
		\bigg(
		M_i
		+
		\eta(t,r)
		r^i
		\bigg)
		\Big(
		\sum
		C^I_{JKl_1l_2mn_1n_2} \big( \log (1+t) \big)^{l_1} t^{\frac{l_2}{2}} h^{-m} b^{-n_1} c^{-n_2} r^J L^K
		\Big)
		\\
		&
		\qquad
		=
		\sum C^I_{JKl_1l_2mn_1n_2}
		\big(\log (1+t) \big)^{l_1+1}
		\bigg(
		-
		m t^{\frac{l_2}{2}} h^{-(m+1)} b^{-(n_1+1)} c^{-n_2} r^{J+e_i} L^K
		\\
		&
		\qquad \qquad
		-
		n_1
		\Big[ 
		t^{\frac{l_2}{2}} h^{-m} b^{-(n_1+1)} c^{-n_2} r^{J+e_i}
		+
		3 t^{\frac{l_2}{2}} h^{-(m+1)} b^{-(n_1+1)} c^{-n_2} r^{J+e_i}
		\\
		&
		\qquad \qquad
		+
		\frac{1}{2} t^{\frac{l_2+1}{2}} h^{-m} b^{-(n_1+1)} c^{-(n_2+1)} r^{J+e_i}
		-
		\sum_{d=1}^3
		\Big(
		t^{\frac{l_2}{2}} h^{-m} b^{-(n_1+2)} c^{-n_2} r^{J+e_i+2e_d}
		\\
		&
		\qquad \qquad
		+
		3
		t^{\frac{l_2}{2}} h^{-(m+1)} b^{-(n_1+2)} c^{-n_2} r^{J+e_i+2e_d}
		+
		\frac{1}{2} t^{\frac{l_2+1}{2}} h^{-m} b^{-(n_1+2)} c^{-(n_2+1)} r^{J+e_i+2e_d}
		\Big)
		\Big]
		L^K
		\\
		&
		\qquad \qquad
		+
		n_2
		\sum_{d=1}^3
		t^{\frac{l_2}{2}} h^{-(m+1)} b^{-(n_1+1)} c^{-(n_2+2)} r^{J+e_i+2e_d} L^K
		-
		n_2
		t^{\frac{l_2}{2}} h^{-(m+1)} b^{-n_1} c^{-(n_2+2)} r^{J+e_i} L^K
		\\
		&
		\qquad \qquad
		+
		j_i
		t^{\frac{l_2}{2}} h^{-(m+1)} b^{-n_1} c^{-n_2} r^{J-e_i} L^K
		-
		\vert J \vert
		t^{\frac{l_2}{2}} h^{-(m+1)} b^{-(n_1+1)} c^{-n_2} r^{J+e_i} L^K
		\\
		&
		\qquad \qquad
		+
		t^{\frac{l_2}{2}} h^{-(m+1)} b^{-n_1} c^{-n_2} r^{J} L^{K+e_i}
		-
		\sum_{d=1}^3
		t^{\frac{l_2}{2}} h^{-(m+1)} b^{-(n_1+1)} c^{-n_2} r^{J+e_i+e_d} L^{K+e_d}
		\\
		&
		\qquad \qquad
		-
		5
		t^{\frac{l_2}{2}} h^{-(m+1)} b^{-(n_1+1)} c^{-n_2} r^{J+e_i} L^K
		+
		\sum_{d=1}^3 \Big[
		3
		t^{\frac{l_2}{2}} h^{-(m+1)} b^{-(n_1+2)} c^{-n_2} r^{J+e_i+2e_d} L^K
		\\
		&
		\qquad \qquad
		+
		t^{\frac{l_2}{2}} h^{-m} b^{-(n_1+2)} c^{-n_2} r^{J+e_i+2e_d} L^K
		+
		\frac{1}{2} t^{\frac{l_2+1}{2}} h^{-m} b^{-(n_1+2)} c^{-(n_2+1)} r^{J+e_i+2e_d} L^K
		\Big]
		\bigg)
		,
	\end{align*}
	and so \eqref{eq:radcomb} holds for $I+e_i$ with, for each $l,m,n\in \mathbb{N}_0$, multi-index $K$, and $J\in (\mathbb{N}_0)^3$,
	\begin{align} \label{eq:radcombrecur}
		&
		C^{I+e_i}_{JKl_1l_2mn_1n_2}
		=
		-
		(m-1) C^I_{J-e_i,K,l_1-1,l_2,m-1,n_1-1,n_2}
		-
		(n_1-1) 
		\Big[
		C^I_{J-e_i,K,l_1-1,l_2,m,n_1-1,n_2}
		\\
		&
		\nonumber
		+
		3 C^I_{J-e_i,K,l_1-1,l_2,m-1,n_1-1,n_2}
		+
		\frac{1}{2} C^I_{J-e_i,K,l_1-1,l_2-1,m,n_1-1,n_2-1}
		\Big]
		+
		(n_1-2) \sum_{d=1}^3
		\Big[
		\\
		&
		\nonumber
		\quad
		C^I_{J-e_i-2e_d,K,l_1-1,l_2,m,n_1-2,n_2}
		+
		3 C^I_{J-e_i-2e_d,K,l_1-1,l_2,m-1,n_1-2,n_2}
		+
		\frac{1}{2} C^I_{J-e_i-2e_d,K,l_1-1,l_2-1,m,n_1-2,n_2-1}
		\Big]
		\\
		&
		\nonumber
		+
		(n_2-2)
		\sum_{d=1}^3
		C^I_{J-e_i-2e_d,K,l_1-1,l_2,m-1,n_1-1,n_2-2}
		-
		(n_2-2)
		C^I_{J-e_i,K,l_1-1,l_2,m-1,n_1,n_2-2}
		\\
		&
		\nonumber
		+
		(j_i+1)
		C^I_{J+e_i,K,l_1-1,l_2,m-1,n_1,n_2}
		-
		(\vert J \vert - 1)
		C^I_{J-e_i,K,l_1-1,l_2,m-1,n_1-1,n_2}
		+
		C^I_{J,K-e_i,l_1-1,l_2,m-1,n_1,n_2}
		\\
		&
		\nonumber
		-
		5
		C^I_{J-e_i,K,l_1-1,l_2,m-1,n_1-1,n_2}
		+
		\sum_{d=1}^3
		\Big[
		- 
		C^I_{J-e_i-e_d,K-e_d,l_1-1,l_2,m-1,n_1-1,n_2}
		\\
		&
		\nonumber
		+
		3
		C^I_{J-e_i-2e_d,K,l_1-1,l_2,m-1,n_1-2,n_2}
		+
		C^I_{J-e_i-2e_d,K,l_1-1,l_2,m,n_1-2,n_2}
		+
		\frac{1}{2} C^I_{J-e_i-2e_d,K,l_1-1,l_2-1,m,n_1-2,n_2-1}
		\Big].
	\end{align}
	
	Consider now the property \eqref{eq:radcomb3}.  Clearly \eqref{eq:radcomb3} holds for $\vert I\vert=1$.  Suppose now that \eqref{eq:radcomb3} holds for some $\vert I\vert\geq 1$.  If $C^{I+e_i}_{JKl_1l_2mn_1n_2} \neq 0$ then at least one of the terms of the right hand side of \eqref{eq:radcombrecur} must be non-vanishing.  By the inductive hypothesis
	\[
		\vert J \vert -3 \leq 3\vert I \vert,
		\quad
		\vert K \vert - 1 \leq \vert I \vert,
		\quad
		l_1 -1 \leq \vert I \vert,
		\quad
		l_2 - 1 \leq \vert I \vert,
		\quad
		m - 1 \leq \vert I \vert,
		\quad
		n_1 - 2 \leq 2\vert I \vert,
		\quad
		n_2 - 2 \leq 2\vert I \vert,
	\]
	i.\@e.\@ \eqref{eq:radcomb3} holds for each $\vert I\vert+1$.
	
	Note now that \eqref{eq:radcomb2} holds for $\vert I \vert=1$.  Suppose \eqref{eq:radcomb2} holds for some $\vert I \vert \geq1$.  It follows from \eqref{eq:radcombrecur}, and the property \eqref{eq:radcomb3}, that
	\[
		\vert C^{I+e_i}_{JKl_1l_2mn_1n_2} \vert
		\leq
		(51 \vert I \vert + 23) (51\vert I \vert)!
		\leq
		(51(\vert I \vert+1))!,
	\]
	i.\@e.\@ \eqref{eq:radcomb2} holds for $\vert I \vert+1$.
	
	Finally, the property \eqref{eq:radcomb4} follows from a similar induction argument using the relation \eqref{eq:radcombrecur}.  
\end{proof}

Again, combinations of the vector fields at $t=1$ can be related to the standard $\partial_x$ and $\partial_p$ derivatives in a similar manner.

\begin{proposition}[Vector fields at $t=1$] \label{prop:radcomb2}
	For any $k \geq 0$,
	\[
		\sum_{\vert I \vert \leq k} \big\vert L^I(f(t,x,p)) \vert_{t=1} \big\vert
		\lesssim
		(6 k)!
		\sum_{\vert I \vert + \vert J \vert \leq k}
		\Big( \log \big(
		1 + \sqrt{1 + \vert p \vert^2}
		\big) \Big)^{\vert I \vert}
		\big\vert \partial_x^I \partial_p^Jf_1(x,p) \big\vert.
	\]
\end{proposition}

\begin{proof}
	Note that
	\[
		L_k\vert_{t=1} = {A_k}^i(1,p) \partial_{x^i} + \partial_{p^k},
		\qquad
		{A_k}^i(1,p)
		=
		2 \log \Big(
		1 + \sqrt{1 + \vert p \vert^2}
		\Big) \delta_k^i
		+
		\frac{
		2 p^ip^k
		}{
		(1+ \sqrt{1+\vert p \vert^2})\sqrt{1+\vert p \vert^2}
		},
	\]
	Now, for any multi-index $I$, there are constants $C^I_{JKLabc} \in \mathbb{N}_0$ such that
	\begin{equation} \label{eq:radt1comb}
		L^I\vert_{t=1}
		=
		\sum
		C^I_{JKLabc} 
		\Big( \log \big(
		1 + \sqrt{1 + \vert p \vert^2}
		\big) \Big)^a
		\Big( 1+ (1+\vert p \vert^2)^{\frac{1}{2}} \Big)^{-b}
		\big( 1+\vert p \vert^2 \big)^{-\frac{c}{2}} p^J \partial_x^K \partial_p^L.
	\end{equation}
	Indeed, \eqref{eq:radt1comb} clearly holds for $I = e_k$ with constants
	\[
		C^{e_k}_{0,e_i,0,1,0,0} = 2 \delta^k_i,
		\qquad
		C^{e_k}_{e_{j_1}+e_{j_2},e_i,0,0,1,1} = 2 \delta^k_{j_1} \delta^i_{j_2},
		\qquad
		C^{e_k}_{0,0,e_i,0,0,0} = \delta^k_i.
	\]
	The fact that \eqref{eq:radt1comb} holds in general can be established by an induction argument, similar to that of the proof of Proposition \ref{prop:radcomb}, by applying $L_i$ to \eqref{eq:radt1comb} and checking that the form is preserved.  Moreover, the recursion relation
	\begin{align*}
		C^{I+e_i}_{JKLabc}
		=
		\
		&
		2 C^I_{J,K-e_i,L,a-1,b,c}
		+
		2 \sum_{k=1}^3 C^I_{J-e_i-e_k,K-e_k,L,a,b-1,c-1}
		+
		(a+1) C^I_{J-e_i,K,L,a+1,b-1,c-1}
		\\
		&
		-
		(b-1) C^I_{J-e_i,K,L,a,b-1,c-1}
		-
		(c-2) C^I_{J-e_i,K,L,a,b,c-2}
		+
		(j_i+1) C^I_{J+e_i,K,L,a,b,c}
		+
		C^I_{J,K,L-e_i,a,b,c},
	\end{align*}
	holds, where $C^{I}_{JKlabc} = 0$ if $j_1<0$, $j_2<0$, $j_3<0$, $a<0$, $b<0$, $c<0$, $K = \emptyset$, or $L = \emptyset$.  As in the proof of Proposition \ref{prop:radcomb}, it is easily inductively shown that each $C^I_{JKLabc}$ satisfies
	\[
		\vert C^I_{JKLabc} \vert \lesssim (6 \vert I \vert)!
		\qquad
		\text{for all}
		\quad
		a,b,c \in \mathbb{N}_0,
		\quad
		J \in (\mathbb{N}_0)^3,
		\quad
		\text{and multi-indices } K,L.
	\]
	and the non-vanishing $C^I_{JKLabc}$ satisfy
	\[
		C^I_{JKLabc} \neq 0
		\qquad
		\Rightarrow
		\quad
		0 \leq \vert J \vert \leq 2\vert I \vert,
		\quad
		0 \leq \vert K \vert + \vert L \vert \leq \vert I \vert,
		\quad
		a \leq \vert K \vert,
		\quad
		b \leq \vert I \vert,
		\quad
		c \leq 2 \vert I \vert,
		\quad
		\vert J \vert \leq b+c.
	\]
	The proof then follows.
\end{proof}

\subsubsection{Derivative relations}

Derivatives of $\rho$ can be related to combinations of the above vector fields applied to $f$ as follows.

\begin{proposition}[Derivatives of $\rho$ and vector fields] \label{prop:rhovectorsrad}
	For all $k \geq 0$,
	\[
		\sum_{\vert I \vert = k} 
		\big\vert \partial_x^I \rho(t,x) \vert
		\lesssim
		\frac{k^{11} (51k)!}{\big(\log (1+t) \big)^k}
		\sum_{\vert I \vert = 0}^k
		\int_{\mathbb{R}^3}
		\vert L^I f(t,x,p) \vert
		dp
		.
	\]
\end{proposition}

\begin{proof}
	Note first that, for any function $g$ and any $i=1,2,3$,
	\begin{align*}
		\partial_{x^i} \int g(t,x,p) dp
		&
		=
		\int
		{(A^{-1})_i}^j L_j g(t,x,p) - {(A^{-1})_i}^j \partial_{p^j} g(t,x,p) dp
		\\
		&
		=
		\int
		\Big(
		{(A^{-1})_i}^j L_j + \partial_{p^j} {(A^{-1})_i}^j
		\Big) g(t,x,p) dp.
	\end{align*}
	Recall the notation introduced in Section \ref{subsec:radbinom} and the facts \eqref{eq:radbinom1}--\eqref{eq:radbinom2}.  Applying the above repeatedly, it follows that
	\[
		(\log(1+t) \partial_x)^I \rho(t,x)
		=
		\int
		\big(
		M
		+
		\eta(t,r) r
		\big)^I f(t,x,p) dp,
	\]
	where the operators $M_i$ are defined by \eqref{eq:sfdefofM} and $\eta(t,r)$ is defined by \eqref{eq:radeta}.  Proposition \ref{prop:radcomb} implies that
	\[
		\Big\vert
		\big(
		M
		+
		\eta(t,r) r
		\big)^I f(t,x,p)
		\Big\vert
		\lesssim
		\sum
		\vert C^I_{JKl_1l_2mn_1n_2} \vert
		(\log (1+t))^{l_1} t^{\frac{l_2}{2}} h(t,r)^{-m} b(t,r)^{-n_1} c(t,r)^{-n_2}
		\vert r \vert^{\vert J\vert} \vert L^K f(t,x,p) \vert.
	\]
	Recall the property \eqref{eq:radcomb4} of Proposition \ref{prop:radcomb} and, accordingly, suppose that $l_2+\vert J \vert \leq 2n_1+n_2$ and $l_1+l_2+\vert J \vert \leq m+2n_1+n_2$.  By the former, there exists $0\leq a \leq 2n_1+n_2$ such that
	\[
		l_2+\vert J \vert = 2n_1+n_2-a,
	\]
	and thus, by the latter, 
	\[
		l_1 \leq m+a.
	\]
	Recall the definitions of $h(t,r)$, $b(t,r)$ and $c(t,r)$ and note that
	\[
		h(t,r)
		\geq
		\log (1+t),
		\qquad
		b(t,r)
		\geq
		c(t,r)^2
		=
		t + \vert r \vert^2,
	\]
	for all $t \geq 1$, $r \in \mathbb{R}^3$.  It in particular follows that,
	\[
		\frac{t^{\frac{l_2}{2}} \vert r \vert^{\vert J\vert}}{b(t,r)^{n_1} c(t,r)^{n_2}}
		\leq
		\frac{t^{\frac{l_2}{2}} \vert r \vert^{\vert J\vert}}{(\vert r \vert^2 + t)^{\frac{2n_1+n_2}{2}}}
		=
		\frac{t^{\frac{l_2}{2}} \vert r \vert^{\vert J\vert}}{(\vert r \vert^2 + t)^{\frac{l_2 + \vert J \vert}{2}}}
		\frac{1}{(\vert r \vert^2 + t)^{\frac{a}{2}}}
		\leq
		\frac{1}{t^{\frac{a}{2}}}.
	\]
	Thus
	\[
		\frac{(\log (1+t))^{l_1} t^{\frac{l_2}{2}} \vert r \vert^{\vert J\vert}}{h(t,r)^m b(t,r)^{n_1} c(t,r)^{n_2}}
		\leq
		(\log (1+t))^{l_1-m}
		t^{-\frac{a}{2}}
		\leq
		1.
	\]
	The property \eqref{eq:radcomb4} of Proposition \ref{prop:radcomb} therefore implies that
	\[
		\Big\vert
		\big(
		M
		+
		\eta(t,r) r
		\big)^I f(t,x,p)
		\Big\vert
		\lesssim
		\sum
		\vert C^I_{JKl_1l_2mn_1n_2} \vert
		\vert L^K f(t,x,p) \vert,
	\]
	If $\vert I\vert = k$, the property \eqref{eq:radcomb3} in particular implies that
	\[
		\# \{ (J,K,l_1,l_2,m,n_1,n_2) \mid C^I_{JKl_1l_2mn_1n_2} \neq 0 \}
		\lesssim
		(3k)^3 \cdot k^3 \cdot k \cdot k \cdot k \cdot (2k) \cdot (2k)
		\lesssim
		k^{11},
	\]
	and so the result follows from the property \eqref{eq:radcomb2}.
\end{proof}

\subsubsection{The proof of Theorem \ref{thm:main2}}

The proof of Theorem \ref{thm:main2} can now be given.

\begin{proof}[Proof of Theorem \ref{thm:main2}]
	Consider some $k \in \mathbb{N}$ and suppose $f_1\in H^k_{\log}$.  Proposition \ref{prop:rhovectorsrad} and Proposition \ref{prop:interp} give, for all $t \geq 1$,
	\begin{multline*}
		\sum_{\vert I \vert = k} \Vert \partial_x^I \rho(t,\cdot) \Vert_{L^2(\mathbb{T}^3)}
		\\
		\lesssim
		\frac{k^{11} (51k)!}{(\log (1+ t))^k}
		\sum_{\vert I \vert \leq k}
		\left( \int_{\mathbb{T}^3} \int_{\mathbb{R}^3} \vert p \vert^2 \vert L^I f(t,x,p) \vert^2 dp dx \right)^{\frac{1}{4}}
		\left( \int_{\mathbb{T}^3} \int_{\mathbb{R}^3} \vert p \vert^4 \vert L^I f(t,x,p) \vert^2 dp dx \right)^{\frac{1}{4}}
		.
	\end{multline*}
	By Proposition \ref{prop:sfcommutators} it follows that $L^If$ solves \eqref{eq:Vlasovradiation} for all multi-indices $I$, and so it follows from Proposition \ref{prop:conservationL2}, the Sobolev inequality of Proposition \ref{prop:Sobolev}, and the Cauchy--Schwarz inequality, along with Proposition \ref{prop:radcomb2}, that, for any $k \geq 2$,
	\begin{equation} \label{eq:mainradestimate}
		\sup_{x\in \mathbb{T}^3} \vert \rho(t,x) - \overline{\rho}(t) \vert
		\lesssim
		\frac{k^{11} (51k)! (6 k)!}{t^3(\log (1+ t))^k\sqrt{k}}
		\sum_{\vert I \vert + \vert J \vert \leq k}
		\left(
		\int_{\mathbb{T}^3} \int_{\mathbb{R}^3}
		(\vert p \vert^2 + \vert p \vert^4)
		(\log (2+\vert p \vert^2))^{\vert I \vert}
		\vert \partial_x^I \partial_p^J f_1(x,p) \vert^2 dp dx
		\right)^{\frac{1}{2}}.
	\end{equation}
	The proof of \eqref{eq:mainradiation2} follows.
	
	Suppose now that $f_1 \in H^{\omega}_{\log}(\mathbb{T}^3 \times \mathbb{R}^3)$.  Note that
	\[
		\sum_{\vert I \vert + \vert J \vert \leq k}
		\left(
		\int_{\mathbb{T}^3} \int_{\mathbb{R}^3}
		(\vert p \vert^2 + \vert p \vert^4)
		(\log (2+\vert p \vert^2))^{\vert I \vert}
		\vert \partial_x^I \partial_p^J f_1(x,p) \vert^2 dp dx
		\right)^{\frac{1}{2}}
		\lesssim
		\frac{k!}{\lambda^k}
		\Vert f_1 \Vert_{H^{\omega}_{\log} (\mathbb{T}^3\times \mathbb{R}^3)}
		,
	\]
	for all $k \geq 0$ and with $\lambda = \lambda(f_1)$.
	For all $t \geq 1$, $f$ satisfies \eqref{eq:mainradestimate} for all $k \geq2$, and so
	\begin{equation} \label{eq:mainradestimate2}
		\sup_{x\in \mathbb{T}^3} \vert \rho(t,x) - \overline{\rho}(t) \vert
		\lesssim
		\frac{1}{t^3} \frac{k^{11} (51k)! (6 k)! k!}{(\log (1+ t))^k \lambda^k\sqrt{k}} \Vert f_1 \Vert_{H^{\omega}_{\log}}
		\lesssim
		\frac{1}{t^3} \frac{(58k)!}{(\mu (\log (1+t))^{\frac{1}{58}})^{58k} \sqrt{58k}} \Vert f_1 \Vert_{H^{\omega}_{\log}}
		,
	\end{equation}
	for all $k \geq 2$, where $\mu = \lambda^{\frac{1}{58}} 2^{-\frac{1}{58}}$ and the fact that
	\[
		k^{11} (51k)! (6 k)! k! \lesssim 2^k (58k)!,
	\]
	has been used (see Proposition \ref{prop:factorials}).  In particular, for $\hat{t}(t):=\lfloor \frac{1}{58} \mu (\log t)^{\frac{1}{58}} \rfloor$, where $\lfloor \cdot \rfloor$ is the floor function, \eqref{eq:mainradestimate2} holds for $k=\hat{t}(t)$ and so (using Proposition \ref{prop:Stirling}),
	\[
		\sup_{x\in \mathbb{T}^3} \vert \rho(t,x) - \overline{\rho}(t) \vert
		\lesssim
		\frac{1}{t^3} \frac{(58 \, \hat{t}(t))!}{(58 \, \hat{t}(t))^{58 \, \hat{t}(t)+\frac{1}{2}}} \Vert f_1 \Vert_{H^{\omega}_{\log}}
		\lesssim
		\frac{1}{t^3} e^{-58 \, \hat{t}(t)} \Vert f_1 \Vert_{H^{\omega}_{\log}}
		\lesssim
		\frac{1}{t^3} e^{- \mu (\log t)^{\frac{1}{58}}} \Vert f_1 \Vert_{H^{\omega}_{\log}},
	\]
	which concludes the proof of \eqref{eq:mainradiation3}.
\end{proof}

\subsection{The proof of Theorem \ref{thm:main1}}
\label{subsec:generalproof}

In this section the proof of Theorem \ref{thm:main1} is given.  The proof is similar to that of Theorem \ref{thm:mainsf}, presented in Section \ref{subsec:sfproof}.  The commuting vector fields are not explicit, however, and their relevant properties are only exhibited for a certain range of $t$ and $p$.  The analytic properties of the functions $G_q$ and $H_q$, discussed in Section \ref{subsec:GandH}, are exploited.  It is assumed throughout this section that $0< q < \frac{1}{2}$, $q \neq \frac{1}{4}, \frac{1}{6}, \ldots$, is fixed.

\subsubsection{Commutation vector fields}

Recall the function $G_q\colon [0,\infty) \to [0,\infty)$ defined by \eqref{eq:Gqprime}.  For $k=1,2,3$, define vector fields
\begin{equation} \label{eq:vectorfields1}
	L_k = {A_k}^i(t,p) \partial_{x^i} + \frac{1}{t^{2q}} \partial_{p^k},
\end{equation}
where
\begin{align} \label{eq:vectorfields2}
	{A_k}^i(t,p)
	&
	=
	G_q \big( t \, \vert t^{2q} p \vert^{-\frac{1}{q}} \big)
	\vert t^{2q} p \vert^{\frac{1-2q}{q}}
	\bigg(
	\delta_k^i
	+
	\frac{1-2q}{q}
	\frac{t^{4q} p^i p^k}{\vert t^{2q} p \vert^2}
	\bigg)
	-
	\frac{1}{q}
	\frac{t^{1-q}}{\sqrt{t^{2q} + \vert t^{2q} p \vert^2}}
	\frac{t^{4q} p^i p^k}{\vert t^{2q} p \vert^2}
	\\
	&
	=
	t^{1-2q}
	\Big[
	H_q \big( t^{-2q} \vert t^{2q} p \vert^2 \big)
	\delta_k^i
	+
	2 t^{-2q}H_q'(t^{-2q} \vert t^{2q} p \vert^2)
	t^{4q} p^i p^k
	\Big]
	.
	\nonumber
\end{align}

Note that, for $q=\frac{1}{3}$, the function $G_{\frac{1}{3}}$ is explicit (see \eqref{eq:G13}) and the vector fields \eqref{eq:vectorfields1}--\eqref{eq:vectorfields2} reduce to the vector fields \eqref{eq:sfvectorfields1}--\eqref{eq:sfvectorfields2} introduced in Section \ref{subsec:sfproof}.  Unlike in the previous sections, however, the vector fields \eqref{eq:vectorfields1}--\eqref{eq:vectorfields2} are only considered when $t \, \vert t^{2q} p \vert^{-\frac{1}{q}}$ is sufficiently large and the behaviour of Proposition \ref{prop:analyticdivision} can be utilised.  For the remaining region, the standard coordinate derivatives $\partial_{x^k}$ are used.

\begin{proposition}[Commuting vector fields for equation \eqref{eq:Vlasovgeneral}] \label{prop:commutators}
	The vector fields $L_k$, for $k=1,2,3$, defined by \eqref{eq:vectorfields1}--\eqref{eq:vectorfields2}, satisfy
	\[
		\Big[
		\partial_t
		+
		\frac{p^i}{\sqrt{1+ t^{2q} \vert p \vert^2}} \partial_{x^i}
		-
		\frac{2q}{t} p^{i} \partial_{p^i}
		,
		L_k
		\Big]
		=
		0, \qquad k=1,2,3.
	\]
	Moreover, the standard coordinate derivatives $\partial_{x^k}$ also satisfy
	\[
		\Big[
		\partial_t
		+
		\frac{p^i}{\sqrt{1+ t^{2q} \vert p \vert^2}} \partial_{x^i}
		-
		\frac{2q}{t} p^{i} \partial_{p^i}
		,
		\partial_{x^k}
		\Big]
		=
		0, \qquad k=1,2,3.
	\]
\end{proposition}

\begin{proof}
	The proof is a direct computation.
\end{proof}

\begin{proposition}[Inverse of the matrix $A$] \label{prop:matrixA}
	There is a constant $\upsilon>1$ such that, if $t^q \geq \upsilon \vert t^{2q} p \vert$, the matrix ${A_k}^i(t,p)$, defined by \eqref{eq:vectorfields2}, is invertible.  The inverse takes the form
		\begin{align*}
		{(A^{-1})_k}^i(t,p)
		=
		\
		&
		\frac{1}{t^{1-2q}} \frac{1}{H_q(t^{-2q} \vert t^{2q}p \vert^2)}
		\Big(
		\delta_k^i
		-
		t^{-2q}B_q(t^{-2q} \vert t^{2q}p \vert^2) t^{4q} p^k p^i
		\Big)
		\qquad
		i,k=1,2,3,
	\end{align*}
	where $B_q$ is as in Proposition \ref{prop:analyticdivision}.
\end{proposition}

\begin{proof}
	Note that if
	\[
		{A_k}^i(t,r)
		=
		h(t,r)
		\Big(
		\delta_k^i
		+
		\frac{r^i r^k}{b(t,r) - \vert r \vert^2}
		\Big),
	\]
	for some functions $h(t,r)$, $b(t,r)$, then
	\[
		{(A^{-1})_k}^i(t,p)
		=
		\frac{
		1
		}{
		h(t,r)
		}
		\Big(
		\delta_k^i
		-
		\frac{r^i r^k}{b(t,r)}
		\Big),
		\qquad
		i,k=1,2,3.
	\]
	The proof follows from setting
	\[
		h(t,r) = t^{1-2q} H_q \big( t^{-2q} \vert t^{2q} p \vert^2 \big),
		\qquad
		b(t,r) = \vert r \vert^2 + \frac{t^{2q} H_q\big( t^{-2q} \vert r \vert^2 \big)}{2 H_q'\big( t^{-2q} \vert r \vert^2 \big)},
	\]
	noting that
	\[
		\frac{1}{b(t,r)}
		=
		t^{-2q}B_q(t^{-2q} \vert t^{2q}p \vert^2)
		,
	\]
	and recalling (see Proposition \ref{prop:analyticdivision}) that $1/H_q$ and $B_q/H_q$ are well defined (in fact real analytic) on $(-\upsilon^{-2},\upsilon^{-2})$ if $\upsilon$ is sufficiently large.
\end{proof}

\begin{remark}[Representation formula] \label{rmk:repgeneral}
	As in the $q= \frac{1}{3}$ case (see Remark \ref{rmk:repsf}), the expression for the vector fields of Proposition \ref{prop:commutators} can be obtained applying $\partial_{p^k}$ to the representation formula
	\[
		f(t,x,p)
		=
		F
		\Big(
		x
		-
		G_q \big( t \, \vert t^{2q} p \vert^{-\frac{1}{q}} \big)
		\vert t^{2q} p \vert^{\frac{1-2q}{q}} 
		t^{2q} p
		,
		t^{2q} p
		\Big).
	\]
	As in the $q= \frac{1}{3}$ case (see Remark \ref{rmk:repsf}), there is also a representation formula in terms of $f_1$, which takes the form
	\[
		f(t,x,p)
		=
		f_1
		\Big(
		x
		-
		\big( G_q \big( t \, \vert t^{2q} p \vert^{-\frac{1}{q}} \big) - G_q \big( \vert t^{2q} p \vert^{-\frac{1}{q}} \big) \big)
		\vert t^{2q} p \vert^{\frac{1-2q}{q}} 
		t^{2q} p
		,
		t^{2q} p
		\Big).
	\]
	A related collection of vector fields can be obtained using this representation formula, which again take the form \eqref{eq:vectorfields1}, where now
	\begin{multline*}
		{A_k}^i(t,p)
		=
		\Big( G_q \big( t \, \vert t^{2q} p \vert^{-\frac{1}{q}} \big) - G_q \big( \vert t^{2q} p \vert^{-\frac{1}{q}} \big) \Big)
		\vert t^{2q} p \vert^{\frac{1-2q}{q}}
		\bigg(
		\delta_k^i
		+
		\frac{1-2q}{q}
		\frac{t^{4q} p^i p^k}{\vert t^{2q} p \vert^2}
		\bigg)
		\\
		-
		\frac{1}{q}
		\bigg(
		\frac{t^{1-q}}{\sqrt{t^{2q} + \vert t^{2q} p \vert^2}}
		-
		\frac{1}{\sqrt{1 + \vert t^{2q} p \vert^2}}
		\bigg)
		\frac{t^{4q} p^i p^k}{\vert t^{2q} p \vert^2}
		.
	\end{multline*}
	These vector fields have the advantage that they are regular, even for the values $q = \frac{1}{4}, \frac{1}{6}, \frac{1}{8}, \ldots$, in contrast to \eqref{eq:vectorfields1}--\eqref{eq:vectorfields2}.  In order to simply the proof, however, only the vector fields \eqref{eq:vectorfields1}--\eqref{eq:vectorfields2} are used in the present work, and thus these exceptional values of $q$ are excluded.
\end{remark}

\subsubsection{A binomial theorem for the commutation vector fields}
\label{subsec:binom}

Consider some $t \geq1$ and $p\in \mathbb{R}^3$ such that $t^q \geq \upsilon \vert t^{2q} p \vert$, where $\upsilon>1$ is as in Proposition \ref{prop:matrixA}.  Define
\[
	r = t^{2q} p.
\]
and define the operators
\begin{equation*}
	M_i
	=
	t^{1-2q} {(A^{-1})_i}^k(t,p) L_k
	=
	\frac{1}{H_q(t^{-2q} \vert r \vert^2)}
	\Big(
	\delta^k_i
	-
	t^{-2q}B_q(t^{-2q} \vert r \vert^2) r^k r^i
	\Big)
	L_k,
	\qquad
	i=1,2,3.
\end{equation*}
A computation gives
\[
	\frac{t^{1-2q}}{t^{2q}} \partial_{p^i} {(A^{-1})_i}^k
	=
	\Phi_q(t^{-2q}\vert r \vert^2) t^{-2q} r^k,
\]
where $\Phi_q \colon (-\upsilon^{-2}, \upsilon^{-2}) \to \mathbb{R}$ is defined by \eqref{eq:analyticdivision}.
As will be seen in the proof of Proposition \ref{prop:rhovectors} below, it follows that, for any multi-index $I$,
\[
	(t^{1-2q} \partial_x)^I \rho(t,x)
	=
	\int
	\big(
	M
	+
	\Phi_q(t^{-2q}\vert r \vert^2) \, t^{-2q} \,r
	\big)^I f(t,x,p) dp.
\]
The following proposition is a Binomial Theorem-type result which relates the operator appearing on the right side to combinations of the commutation vector fields.

\begin{proposition}[A binomial theorem for the commutation vector fields] \label{prop:comb}
	If $\upsilon$ is sufficiently large then, for each multi-index $\vert I \vert \geq 1$, there are $d^I_{JKl} \in \mathbb{R}$, for $l\in \mathbb{N}_0$ and multi-indices $K$, and $J\in (\mathbb{N}_0)^3$, such that, for all $t^q \geq \upsilon \vert r \vert$,
	\begin{equation} \label{eq:comb}
		\big(
		M
		+
		\Phi_q(t^{-2q}\vert r \vert^2) \, t^{-2q} \,r
		\big)^I
		=
		\sum_{K,l}
		\sum_{\vert J \vert =0}^{\infty}
		d^I_{JKl} t^{-lq} (t^{-q}r)^J L^K.
	\end{equation}
	Moreover
	\begin{itemize}
		\item
			There are constants $C,\tilde{C}>0$ such that each $d^I_{JKl}$ satisfies
			\begin{equation} \label{eq:comb2}
				\vert d^I_{JKl} \vert \lesssim \tilde{C}^{\vert I \vert} C^{\vert J \vert} (1+\vert J \vert)^{4\vert I \vert}
				\qquad
				\text{for all}
				\quad
				l\in \mathbb{N}_0,
				\quad
				J \in (\mathbb{N}_0)^3,
				\quad
				\text{and multi-indices } K.
			\end{equation}
		\item
			The non-vanishing $d^I_{JKl}$ satisfy
			\begin{align} \label{eq:comb3}
				d^I_{JKl} \neq 0
				\qquad
				\Rightarrow
				\qquad
				&
				0 \leq \vert K \vert \leq \vert I \vert,
				\quad
				0 \leq l \leq \vert I \vert.
			\end{align}
	\end{itemize}
\end{proposition}

\begin{proof}
	By Proposition \ref{prop:analyticdivision}, $1/H_q$, $B_q/H_q$ and $\Phi_q$ are real analytic on $(-\upsilon^{-2}, \upsilon^{-2})$ and thus there are $h_N, b_N, \phi_N \in \mathbb{R}$, for each $N \in \mathbb{N}_0^3$,
	satisfying
	\begin{equation} \label{eq:hNbNphiN}
		\vert h_N \vert \lesssim C^{\vert N \vert},
		\qquad
		\vert b_N \vert \lesssim C^{\vert N \vert},
		\qquad
		\vert \phi_N \vert \lesssim C^{\vert N \vert},
	\end{equation}
	for some constant $C>0$, such that
	\[
		\frac{1}{H_q(t^{-2q} \vert r \vert^2)} = \sum_{\vert N \vert = 0}^{\infty} h_N (t^{-q} r)^N,
		\quad
		\frac{B_q(t^{-2q} \vert r \vert^2)}{H_q(t^{-2q} \vert r \vert^2)} = \sum_{\vert N \vert = 0}^{\infty} b_N (t^{-q} r)^N,
		\quad
		\Phi_q(t^{-2q}\vert r \vert^2) = \sum_{\vert N \vert = 0}^{\infty} \phi_N (t^{-q} r)^N,
	\]
	for all $t^q > \upsilon \vert r \vert$.  It follows that \eqref{eq:comb} holds for $I=e_i$ with
	\[
		d^{e_i}_{J,0,e_k} = h_J \delta^i_k,
		\qquad
		d^{e_i}_{J+e_{j_1}+e_{j_2},0,e_k} = - b_J \delta^i_{j_1} \delta_{k,j_2},
		\qquad
		d^{e_i}_{J+e_j,1,0} = \phi_J \delta^i_j,
		\qquad
		d^I_{JKl} = 0
		\text{ otherwise}
		.
	\]
	The estimates \eqref{eq:comb2}, for $I=e_i$, follow from \eqref{eq:hNbNphiN}.
	Suppose now that \eqref{eq:comb} holds for some $\vert I \vert \geq 1$ for $d^I_{JKl}$ satisfying \eqref{eq:comb2}--\eqref{eq:comb3}.  Then
	\begin{align*}
		&
		\big(
		M
		+
		\Phi_q(t^{-2q}\vert r \vert^2) \, t^{-2q} \,r
		\big)^{I+e_i}
		=
		\sum_{K,l}
		\sum_{\vert J \vert =0}^{\infty}
		d^I_{JKl}
		\Big[
		t^{-(l+1)q} \sum_{\vert N \vert = 0}^{\infty} h_N j_i (t^{-q}r)^{J+N-e_i} L^K
		\\
		&
		\qquad \qquad
		-
		\vert J \vert t^{-(l+1)q} \sum_{\vert N \vert = 0}^{\infty} b_N (t^{-q}r)^{J+N+e_i} L^K
		+
		t^{-lq} \sum_{\vert N \vert = 0}^{\infty} h_N j_i (t^{-q}r)^{J+N} L^{K+e_i}
		\\
		&
		\qquad \qquad
		-
		t^{-lq} \sum_{c=1}^3 \sum_{\vert N \vert = 0}^{\infty} b_N (t^{-q}r)^{J+N+e_i+e_c} L^{K+e_c}
		+
		t^{-(l+1)q} \sum_{\vert N \vert = 0}^{\infty} \phi_N (t^{-q}r)^{J+N+e_i} L^K
		\Big],
	\end{align*}
	and so $( M + \Phi_q(t^{-2q}\vert r \vert^2) t^{-2q} r)^{I+e_i}$ takes the form \eqref{eq:comb} with
	\begin{align*}
		d^{I+e_i}_{JKl}
		=
		\
		&
		(j_i+1) \sum_{J'=0}^{J+e_i} d^I_{J',K,l-1} h_{J-J'+e_i}
		-
		(\vert J \vert -1) \sum_{J'=0}^{J-e_i} d^I_{J',K,l-1} h_{J-J'-e_i}
		+
		\sum_{J'=0}^{J} d^I_{J',K-e_i,l} h_{J-J'}
		\\
		&
		+
		\sum_{c=1}^3 \sum_{J'=0}^{J-e_i-e_c} d^I_{J',K-e_c,l} b_{J-J'-e_i-e_c}
		+
		\sum_{J'=0}^{J-e_i} d^I_{J',K,l-1} \phi_{J-J'-e_i},
	\end{align*}
	where $\sum_{J'=0}^{J}$ is interpreted component-wise, and $\sum_{J'=0}^{J} = 0$ if $j_c < j_c'$ for some $c=1,2,3$.  Moreover the inductive assumption \eqref{eq:comb2}, together with the estimates \eqref{eq:hNbNphiN}, implies that
	\begin{align*}
		\vert d^{I+e_i}_{JKl} \vert
		&
		\leq
		\tilde{C}^{\vert I \vert}
		\Big(
		2 (1 + \vert J \vert) \sum_{J'=0}^{J+e_i} (1 + \vert J' \vert)^{4\vert I \vert} C^{\vert J' \vert} C^{\vert J \vert - \vert J'\vert +1}
		+
		5
		\sum_{J'=0}^{J} (1 + \vert J' \vert)^{4\vert I \vert} C^{\vert J' \vert} C^{\vert J \vert - \vert J'\vert}
		\Big)
		\\
		&
		\leq
		\tilde{C}^{\vert I \vert}C^{\vert J \vert}
		(1 + \vert J \vert)^{4\vert I \vert}
		\big(
		2C (1 + \vert J \vert)^{4}
		+
		5
		(1 + \vert J \vert)^{3}
		\Big)
		\leq
		\tilde{C}^{\vert I \vert+1}C^{\vert J \vert}
		(1 + \vert J \vert)^{4(\vert I \vert +1)}
		,
	\end{align*}
	where the latter follows if $\tilde{C} \geq 2C+5$.  Thus $d^{I+e_i}_{JKl}$ satisfies \eqref{eq:comb2}.  Finally, \eqref{eq:comb3} again follows by induction from the above recursion.
	
\end{proof}

\subsubsection{Dyadic decomposition of $f$}

Let $\chi \colon \mathbb{R} \to \mathbb{R}$ be a smooth cut off function satisfying $\vert \chi\vert \leq 1$,
\[
	\chi(z)
	=
	\begin{cases}
		1
		&
		\text{ if } \vert z \vert \leq \frac{1}{2},
		\\
		0
		&
		\text{ if } \vert z \vert \geq 1,
	\end{cases}
\]
and, for all $k \geq 0$,
\begin{equation} \label{eq:chigrowth}
	\sup_z \vert \partial_z^k \chi(z) \vert \lesssim (2k)!,
\end{equation}
and define, for each $n\geq 1$, and $z \geq 0$,
\[
	\chi_n(z)
	=
	\chi \Big( \frac{z}{2^{n+1}} \Big) - \chi \Big( \frac{z}{2^{n}} \Big),
	\qquad
	\chi_0(z) = \chi \Big( \frac{z}{2} \Big).
\]
Note that,
\[
	\supp(\chi_0) \subset \{ \vert z \vert \leq 2\},
\]
and, for all $n \in \mathbb{N}$ and $z\geq 0$,
\[
	\supp(\chi_n) \subset \{2^{n-1} \leq \vert z \vert \leq 2^{n+1}\},
	\qquad
	\sum_{m=0}^n \chi_m(z) = \chi \Big( \frac{z}{2^{n+1}} \Big)
\]
Define then
\[
	f^n (t,x,p) = f(t,x,p) \, \chi_n \big( \vert t^{2q} p\vert \big).
\]
Note that each $f^n$ solves the Vlasov equation \eqref{eq:Vlasovgeneral}, and, for $n \geq 1$,
\begin{equation} \label{eq:suppfn}
	\supp(f^n) \subset \{ (t,x,p) \mid 2^{n-1} \leq \vert t^{2q} p\vert \leq 2^{n+1} \},
	\qquad
	\supp(f^0) \subset \{ (t,x,p) \mid 0 \leq \vert t^{2q} p\vert \leq 2 \}.
\end{equation}
Define also the corresponding average,
\[
	\rho^n(t,x) = \int_{\mathbb{R}^3} f^n(t,x,p) dp,
\]
and initial condition
\[
	f^n_1 (x,p) = f^n(1,x,p) = f_1(x,p) \, \chi_n \big( \vert p\vert \big).
\]
Note that
\[
	\sum_{n=0}^N \rho^n \to \rho, \qquad \text{uniformly as } N\to \infty.
\]

\subsubsection{Derivative relations}

For each $n$, the $\partial_x$ derivatives of $\rho^n$ can be related to the above vector fields applied to $f^n$.

\begin{proposition}[Derivatives of $\rho$ and vector fields] \label{prop:rhovectors}
	For $\upsilon$ sufficiently large, there is a constant $C>1$ such that, for each $n \in \mathbb{N}$, for all $k \geq 0$, and for any $t^q \geq \upsilon \, 2^{n+1}$,
	\[
		\sum_{\vert I \vert = k} 
		\big\vert \partial_x^I \rho^n(t,x) \vert
		\lesssim
		\frac{C^k (4 k)!}{t^{(1-2q)k}}
		\sum_{\vert I \vert = 0}^k
		\int_{\mathbb{R}^3}
		\vert L^I f^n(t,x,p) \vert
		dp
		.
	\]
\end{proposition}

\begin{proof}
	From the assumption that $t^q \geq \upsilon \, 2^{n+1}$, and the support property \eqref{eq:suppfn}, it follows that $t^q \geq \upsilon \, \vert r \vert$ in $\supp(f^n)$.  The results of Proposition \ref{prop:comb} therefore apply.  Now, for any function $g$ and any $i=1,2,3$,
	\begin{align*}
		t^{1-2q} \partial_{x^i} \int g(t,x,p) dp
		&
		=
		\int
		t^{1-2q} {(A^{-1})_i}^j L_j g(t,x,p) - t^{1-4q} {(A^{-1})_i}^j \partial_{p^j} g(t,x,p) dp
		\\
		&
		=
		\int
		\Big(
		t^{1-2q} {(A^{-1})_i}^j L_j + t^{1-4q} \partial_{p^j} {(A^{-1})_i}^j
		\Big) g(t,x,p) dp
		\\
		&
		=
		\int
		\Big(
		M_i + \Phi_q(t^{-2q}\vert r \vert^2) \, t^{-2q} \,r^i
		\Big) g(t,x,p) dp.
	\end{align*}
	Applying the above repeatedly, it follows that, for any multi-index $I$,
	\[
		(t^{1-2q} \partial_x)^I \rho^n(t,x)
		=
		\int
		\big(
		M
		+
		\eta(t,r) \,r
		\big)^I f^n(t,x,p) dp.
	\]
	Proposition \ref{prop:comb} implies that
	\begin{align}
		\sum_{\vert I \vert = k}
		\big\vert
		\big(
		M
		+
		\eta(t,r) \,r
		\big)^I f^n
		\big\vert
		&
		\leq
		\sum_{K,l}
		\sum_{\vert J \vert =0}^{\infty}
		\vert d^I_{JKl} \vert t^{-lq} \upsilon^{-\vert J\vert} 
		\vert L^K f^n \vert
		\nonumber
		\\
		&
		\leq
		k \, \tilde{C}^k
		\sum_{\vert K \vert \leq k} \vert L^K f^n \vert
		\sum_{\vert J \vert =0}^{\infty}
		(C/\upsilon)^{\vert J \vert} (1+\vert J \vert)^{4\vert I \vert}
		\lesssim
		C^k (4k)! \sum_{\vert K \vert \leq k} \vert L^K f^n \vert,
		\label{eq:derivrhoproofmain}
	\end{align}
	where the latter follows from Lemma \ref{prop:summationlemma}, provided $C/\upsilon \leq e^{-1}$, and the proof follows.
\end{proof}

\subsubsection{Vector fields at $t=1$}

The main result of this section is the following proposition, which relates $L^2$ norms of the vector fields applied to $f^n$, at $t=1$, to the $H^k_q$ norms of $f_1$ (see Section \ref{subsec:functionspaces}).

\begin{proposition}[Vector fields at $t=1$] \label{prop:comb2}
	There is a constant $C>1$ such that, for any $k \geq 0$, and any $n \geq 0$,
	\begin{multline*}
		\sum_{\vert I \vert \leq k} 
		\Big(
		\int_{\mathbb{R}^3} (\vert p \vert^2 + \vert p \vert^4)
		\big\vert L^I(f^n(t,x,p)) \vert_{t=1} \big\vert^2
		dp
		\Big)^{\frac{1}{2}}
		\\
		\lesssim
		C^k (5 k)!
		\sum_{\vert I \vert + \vert J \vert \leq k}
		\Big(
		\int_{2^{n-1} \leq \vert p \vert \leq 2^{n+1}}
		(\vert p \vert^2 + \vert p \vert^4)
		(1+\vert p \vert)^{\vert I \vert (\frac{1-2q}{q})}
		\big\vert \partial_x^I \partial_p^Jf_1(x,p) \big\vert^2
		dp
		\Big)^{\frac{1}{2}}.
	\end{multline*}
\end{proposition}

The main parts of the proof of Proposition \ref{prop:comb2} is divided into the following two lemmas.  The first relates $L^I\vert_{t=1}$ to $\partial_x$ and $\partial_p$ on compact sets.

\begin{lemma}[$L^I\vert_{t=1}$ on compact sets] \label{prop:LIoncompact}
	For any $p_* \in \mathbb{R}^3$, there exists $\epsilon>0$, $B,C>0$, and constants $d^I_{JMN} \in \mathbb{R}$, satisfying
	\begin{equation} \label{eq:LIoncompact1}
		\vert d^{I}_{JNM} \vert
		\leq
		C^{\vert J \vert} (B+\vert J \vert)^{2\vert I\vert} \vert I \vert!,
	\end{equation}
	such that
	\begin{equation} \label{eq:LIoncompact2}
		L^I\vert_{t=1}
		=
		\sum_{\vert N\vert + \vert M\vert \leq \vert I\vert}
		\sum_{\vert J \vert = 0}^{\infty}
		(p-p_*)^J
		d^{I}_{JNM}
		\partial_x^N \partial_p^M,
		\qquad
		\text{for}
		\qquad
		\vert p - p_* \vert < \epsilon.
	\end{equation}
\end{lemma}

\begin{proof}
	By Proposition \ref{prop:hanalytic} the matrix components ${A_k}^i(1,p)$ are real analytic functions, and thus, for each $p_* \in \mathbb{R}^3$, there is a $C>0$, $\epsilon>0$ such that
	\[
		{A_k}^i(1,p) = \sum_{\vert J \vert = 0}^{\infty} a^{ki}_J (p-p_*)^J,
	\]
	where $a^{ki}_J$ satisfy
	\[
		\vert a^{ki}_J \vert \lesssim C^{\vert J \vert},
	\]
	for all $J$.  It follows that, for $I=e_i$, $L_i \vert_{t=1}$ has the form \eqref{eq:LIoncompact2} with
	\[
		d^{e_i}_{J,e_k,0} = a^{ki}_J,
		\qquad
		d^{e_i}_{J,0,e_l} = \delta_l^i,
		\qquad
		d^{e_i}_{JKL} = 0
		\quad
		\text{otherwise}.
	\]
	Suppose now that $L^I\vert_{t=1}$ takes the form \eqref{eq:LIoncompact2} for some $\vert I \vert \geq 1$, with $d^I_{JKL}$ satisfying
	\begin{equation} \label{eq:LIoncompact1again}
		\vert d^{I}_{JNM} \vert
		\leq
		C^{\vert J \vert} (B+\vert J \vert)^{2\vert I\vert} \frac{(\vert I \vert + \vert J \vert)!}{\vert J \vert!},
	\end{equation}
	for all $J$, $N$, $M$.  Then
	\begin{align*}
		L_i L^I\vert_{t=1}
		&
		=
		\sum_{\vert J \vert=0}^{\infty}
		\sum_{\vert N\vert + \vert M\vert \leq \vert I\vert}
		d^I_{JKM}
		\Big(
		a_{J'}^{i,k} (p-p_*)^{J+J'} \partial_x^{N+e_k} \partial_p^M
		+
		j_i (p-p_*)^{J-e_i} \partial_x^N \partial_p^M
		+
		(p-p_*)^J \partial_x^N \partial_p^{M+e_i}
		\Big)
		\\
		&
		=
		\sum_{\vert J \vert=0}^{\infty}
		\sum_{\vert N\vert + \vert M\vert \leq \vert I\vert}
		(p-p_*)^J
		\Big(
		\sum_{J' = 0}^J
		a_{J-J'}^{i,k} d^I_{J',N-e_k,M}
		+
		(j_i+1) d^I_{J+e_i,N,M}
		+
		d^I_{J,N,M-e_i}
		\Big)
		\partial_x^N \partial_p^M,
	\end{align*}
	where $0\leq J' \leq J$ is interpreted component wise.  Thus \eqref{eq:LIoncompact2} holds for $I+e_i$, with
	\[
		d^{I+e_i}_{J,N,M}
		=
		\sum_{J' = 0}^J
		a_{J-J'}^{i,k} d^I_{J',N-e_k,M}
		+
		(j_i+1) d^I_{J+e_i,N,M}
		+
		d^I_{J,N,M-e_i}.
	\]
	Suppose now inductively that, for some $I$, the estimate \eqref{eq:LIoncompact1again} holds for all $J,K,L$.  Then
	\begin{align*}
		\vert d^{I+e_i}_{J,N,M} \vert
		&
		\leq
		(B+\vert J \vert)^{2\vert I\vert}
		\Big(
		\sum_{J' =0}^{J}
		C^{\vert J \vert - \vert J' \vert} C^{\vert J' \vert} \frac{(\vert I \vert + \vert J' \vert)!}{\vert J' \vert!}
		+
		(\vert J \vert + 1) C^{\vert J \vert+1} \frac{(\vert I \vert + \vert J \vert + 1)!}{(\vert J \vert + 1)!}
		+
		C^{\vert J \vert} \frac{(\vert I \vert + \vert J \vert)!}{\vert J \vert!}
		\Big)
		\\
		&
		\leq
		(B+\vert J \vert)^{2\vert I\vert} C^{\vert J \vert} \frac{(\vert I \vert + 1 + \vert J \vert)!}{\vert J \vert!}
		\Big(
		\frac{(\vert J \vert +1)^3}{\vert I \vert + 1 + \vert J \vert}
		+
		C
		+
		\frac{1}{\vert I \vert + 1 + \vert J \vert}
		\Big)\\
		&
		\leq
		(B+\vert J \vert)^{2(\vert I\vert+1)} C^{\vert J \vert} \frac{(\vert I \vert + 1 + \vert J \vert)!}{\vert J \vert!},
	\end{align*}
	where the latter inequality holds provided $B \geq C+2$, and the fact that
	\[
		\frac{(\vert I \vert + \vert J' \vert)!}{\vert J' \vert!} \leq \frac{(\vert I \vert + \vert J \vert)!}{\vert J \vert!},
	\]
	for each $0\leq \vert J'\vert \leq \vert J \vert$ has been used in the previous step.  Thus \eqref{eq:LIoncompact1again} holds for $I+e_i$, and hence for all $I$ by induction.  Now \eqref{eq:LIoncompact1again} implies that \eqref{eq:LIoncompact1} holds, for new constants.  Indeed, for all $I,J$,
	\[
		\frac{(\vert I \vert + \vert J \vert)!}{\vert J \vert!}
		=
		\vert I\vert! \frac{(\vert I \vert + \vert J \vert)!}{\vert I\vert! \vert J \vert!}
		=
		\vert I \vert! \binom{\vert I \vert + \vert J \vert}{\vert I \vert}
		\leq
		2^{\vert I \vert + \vert J \vert} \vert I \vert!.
	\]
\end{proof}

The next lemma relates $L^I\vert_{t=1}$ to $\partial_x$ and $\partial_p$ for large $p$.

\begin{lemma}[$L^I\vert_{t=1}$ for large $p$] \label{prop:LIlargep}
	For any $R>1$ and all $p \in \mathbb{R}^3$ with $\vert p \vert \geq R$, there exists $C>0$, and constants $d^I_{klJMN} \in \mathbb{R}$, satisfying
	\begin{equation} \label{eq:LIlargep1}
		\vert d^{I}_{klJNM} \vert
		\leq
		C^{k+\vert I \vert} (1+k)^{\vert I \vert} \vert I \vert!,
	\end{equation}
	and
	\begin{equation} \label{eq:LIlargep2}
		d^{I}_{klJNM} \neq 0
		\quad
		\Rightarrow
		\quad
		\vert J \vert \leq k,
		\quad
		l \leq \vert N \vert,
		\quad
		\vert N \vert + \vert M \vert \leq \vert I \vert,
		\quad
		\vert J \vert \leq 2 \vert I\vert,
	\end{equation}
	such that
	\begin{equation} \label{eq:LIlargep3}
		L^I\vert_{t=1}
		=
		\sum_{l,J,N,M}
		\sum_{k = 0}^{\infty}
		p^J
		d^{I}_{klJNM}
		\vert p \vert^{l(\frac{1-2q}{q}) -k}
		\partial_x^N \partial_p^M.
	\end{equation}
\end{lemma}

\begin{proof}
	By Proposition \ref{prop:hanalyticlargep} there are $h_k$ and $h_k'$ such that
	\[
		L_i \vert_{t=1}
		=
		\Big(
		\Big(
		X_q \vert p \vert^{\frac{1-2q}{q}}
		+
		\sum_{k=1}^{\infty} h_k \vert p \vert^{-k}
		\Big)
		\delta_i^j
		+
		\Big(
		Y_q \vert p \vert^{\frac{1-4q}{q}}
		+
		\sum_{k=3}^{\infty} h_k' \vert p \vert^{-k}
		\Big)
		p^ip^j
		\Big)
		\partial_{x^j}
		+
		\partial_{p^i}.
	\]
	It follows that, for $I=e_i$, $L_i \vert_{t=1}$ has the form \eqref{eq:LIlargep3} with
	\[
		d^{e_i}_{0,1,0,e_n,0} = X_q \delta^i_n,
		\qquad
		d^{e_i}_{k,0,0,e_n,0} = h_k \delta^i_n,
		\qquad
		d^{e_i}_{2,1,e_{j_1}+e_{j_2},e_n,0} = Y_q \delta^i_{j_1} \delta_{n,j_2},
	\]
	\[
		d^{e_i}_{k,0,e_{j_1}+e_{j_2},e_n,0} = h_k' \delta^i_{j_1} \delta_{n,j_2},
		\qquad
		d^{e_i}_{0,0,0,0,e_m} = \delta^i_m,
		\qquad
		d^{e_i}_{klJNM} = 0
		\quad
		\text{otherwise},
	\]
	where $h_0=h_0'=h_1'=h_2'=0$.
	(Note that this is not the unique way to express $L_i \vert_{t=1}$ in this form.)  Suppose now that $L^I\vert_{t=1}$ takes the form \eqref{eq:LIlargep3} for some $\vert I \vert \geq 1$, with $d^{I}_{klJNM}$ satisfying \eqref{eq:LIlargep1}--\eqref{eq:LIlargep2}.  Then
	\begin{align*}
		&
		L^{I+e_i}\vert_{t=1}
		=
		\sum_{l,J,N,M}
		\sum_{k = 0}^{\infty}
		d^{I}_{klJNM}
		\Big[
		X_q
		p^J
		\vert p \vert^{(l+1)(\frac{1-2q}{q}) -k}
		\partial_x^{N+e_i} \partial_p^M
		+
		\sum_{k'=1}^{\infty}
		h_{k'}
		p^J
		\vert p \vert^{l(\frac{1-2q}{q}) -(k+k')}
		\partial_x^{N+e_i} \partial_p^M
		\\
		&
		\qquad
		+
		\sum_{c=1}^3
		Y_q
		p^{J+e_i+e_c}
		\vert p \vert^{(l+1)(\frac{1-2q}{q}) -(k+2)}
		\partial_x^{N+e_c} \partial_p^M
		+
		\sum_{c=1}^3
		\sum_{k'=3}^{\infty}
		h_{k'}'
		p^{J+e_i+e_c}
		\vert p \vert^{l(\frac{1-2q}{q}) -(k+k')}
		\partial_x^{N+e_c} \partial_p^M
		\\
		&
		\qquad
		+
		j_i
		p^{J-e_i}
		\vert p \vert^{l(\frac{1-2q}{q}) -k}
		\partial_x^{N} \partial_p^M
		+
		\Big(
		l \Big( \frac{1-2q}{q} \Big) -k
		\Big)
		p^{J+e_i}
		\vert p \vert^{l(\frac{1-2q}{q}) -(k+2)}
		\partial_x^{N} \partial_p^M
		+
		p^{J}
		\vert p \vert^{l(\frac{1-2q}{q}) -k}
		\partial_x^{N} \partial_p^{M+e_i}
		\Big],
	\end{align*}
	and thus $L^{I+e_i}\vert_{t=1}$ also takes the form \eqref{eq:LIlargep3} with
	\begin{align}
		&
		d^{I+e_i}_{klJNM}
		=
		(j_i+1) d^{I}_{k,l,J+e_i,N,M}
		+
		\Big( l \frac{1-2q}{q} - (k-2) \Big) d^{I}_{k-2,l,J-e_i,N,M}
		+
		d^{I}_{k,l,J,N,M-e_i}
		+
		X_q d^{I}_{k,l-1,J,N-e_i,M}
		\nonumber
		\\
		&
		\quad
		+
		\sum_{k'=0}^{k-1}
		d^{I}_{k',l,J,N-e_i,M} h_{k-k'}
		+
		\sum_{c=1}^3
		\Big(
		Y_q d^{I}_{k-2,l-1,J-e_i-e_c,N-e_c,M}
		+
		\sum_{k'=0}^{k-3} d^{I}_{k',l,J-e_i-e_c,N-e_c,M} h_{k-k'}'
		\Big),
		\label{eq:LIlargepproof}
	\end{align}
	where $\sum_{k'=0}^K :=0$ if $K \leq -1$ and $d^I_{klJNM} := 0$ if $k < 0$, $l < 0$, $j_1,j_2,j_3 <0$, or $N, M = \emptyset$.
	
	The properties \eqref{eq:LIlargep2} easily follow by induction from the recursion \eqref{eq:LIlargepproof}.  Suppose now that $\vert I \vert \geq 1$ is such that the estimate \eqref{eq:LIlargep1} holds for all $k,l,J,N,M$.  The recursion \eqref{eq:LIlargepproof}, along with \eqref{eq:LIlargep2} (which implies that $j_i \leq k$), implies that
	\begin{align*}
		\vert d^{I+e_i}_{klJNM} \vert
		\leq
		\
		&
		C^{\vert I \vert}
		(1+k)^{\vert I \vert} \vert I \vert!
		\Big(
		C^k (1+k)
		+
		C^{\vert k-2 \vert} \vert I \vert \frac{1-2q}{q}
		+
		C^{\vert k-2 \vert} \vert k - 2 \vert
		+
		C^k
		+
		C^k \vert X_q \vert
		\\
		&
		+
		\sum_{k'=0}^{k-1}
		C^{k'} C^{k-k'}
		+
		3C^{\vert k-2 \vert} \vert Y_q \vert
		+
		3\sum_{k'=0}^{k-3} C^{k'} C^{k-k'}
		\Big)
		\\
		\leq
		\
		&
		C^{k+\vert I \vert} (1+k)^{\vert I \vert+1} \vert I \vert!
		\Big(
		7
		+
		\vert I \vert \frac{1-2q}{q}
		+
		\vert X_q \vert
		+
		3 \vert Y_q \vert
		\Big)
		\\
		\leq
		\
		&
		C^{k+\vert I \vert+1} (1+k)^{\vert I \vert+1} (\vert I \vert +1)!,
	\end{align*}
	where the final inequality holds if $C \geq \frac{1-2q}{q} + 7 + \vert X_q \vert + 3 \vert Y_q \vert$.  Thus \eqref{eq:LIlargep1} holds for $I+e_i$ and hence, by induction, for all $I$.
\end{proof}

The proof of Proposition \ref{prop:comb2} can now be given.

\begin{proof}[Proof of Proposition \ref{prop:comb2}]
	Note that
	\[
		L_k\vert_{t=1} = {A_k}^i(1,p) \partial_{x^i} + \partial_{p^k},
		\qquad
		{A_k}^i(1,p)
		=
		H_q \big(\vert p \vert^2 \big)
		\delta_k^i
		+
		2 H_q '(\vert p \vert^2)
		p^i p^k.
	\]
	We first show that, for any $k \geq 0$, and any $n \geq 0$,
	\begin{equation} \label{prop:comb2proof1}
		\sum_{\vert I \vert \leq k}
		\big\vert L^I(f^n(t,x,p)) \vert_{t=1} \big\vert
		\lesssim
		C^k (3k)!
		\sum_{\vert I \vert + \vert J \vert \leq k}
		(1+\vert p \vert)^{\vert I \vert (\frac{1-2q}{q})}
		\big\vert \partial_x^I \partial_p^Jf_1(x,p) \big\vert.
	\end{equation}
	
	Given $v\in \mathbb{R}^3$ and $R>0$, define
	\[
		B(v,R) = \{ w\in \mathbb{R}^3 \mid \vert v - w \vert < R\},
		\qquad
		\overline{B}(v,R) = \{ w\in \mathbb{R}^3 \mid \vert v - w \vert \leq R\}.
	\]
	Consider some $R >1$.  For any $p_* \in \overline{B}(0,R)$, let $\epsilon = \epsilon(p_*)$ be as in Lemma \ref{prop:LIoncompact}.  By compactness, there are finitely many $p_*^1,\ldots,p_*^N \in \overline{B}(0,R)$ such that
	\[
		\overline{B}(0,R) \subset B(p_*^1,\epsilon(p_*^1)) \cup \ldots \cup B(p_*^N,\epsilon(p_*^N)).
	\]
	Thus, for any $p \in \overline{B}(0,R)$, there is $1 \leq n \leq N$ such that $p \in B(p_*^n,\epsilon(p_*^n))$.  Lemma \ref{prop:LIoncompact} then implies that
	\[
		\big\vert L^I(f^n(t,x,p)) \vert_{t=1} \big\vert
		\lesssim
		\vert I \vert!
		\sum_{\vert J \vert = 0}^{\infty}
		(\epsilon(p_*^n) C)^{\vert J \vert} (B+\vert J \vert)^{2\vert I\vert}
		\sum_{\vert N\vert + \vert M\vert \leq \vert I\vert}
		\vert \partial_x^N \partial_p^M f^n_1(x,p) \vert,
	\]
	and so, by Lemma \ref{prop:summationlemma}, the estimate \eqref{prop:comb2proof1} holds for all $p \in \overline{B}(0,R)$ (since $\gamma_q \geq 2$) if $\max\{ \epsilon(p_*^1), \ldots, \epsilon(p_*^N)\}$ is sufficiently small.
	
	Suppose now that $\vert p \vert > R$.  By Lemma \ref{prop:LIlargep}, for fixed $l,N,M$,
	\begin{align*}
		&
		\sum_{J}
		\sum_{k = 0}^{\infty}
		\vert p \vert^{\vert J \vert -k}
		\vert d^{I}_{klJNM} \vert
		=
		\sum_{\vert J \vert =0}^{2 \vert I \vert}
		\sum_{k = \vert J \vert}^{\infty}
		\vert p \vert^{\vert J \vert -k}
		\vert d^{I}_{klJNM} \vert
		=
		\sum_{k = 0}^{\infty}
		\sum_{\vert J \vert =0}^{2 \vert I \vert}
		\vert d^{I}_{k+\vert J \vert, l, J,N,M} \vert
		\vert p \vert^{-k}
		\\
		&
		\qquad
		\lesssim
		(2 \vert I \vert+1)
		C^{\vert I \vert} \vert I \vert!
		\sum_{k = 0}^{\infty}
		(C/R)^k (1+k+ 2 \vert I \vert)^{\vert I \vert}
		\lesssim
		C^{\vert I \vert} \vert I \vert!
		\sum_{k = 0}^{\infty}
		(C/R)^k \big( (1+k)^{\vert I \vert} + \frac{(2 \vert I \vert)!}{\vert I \vert!} \big)
		\lesssim
		C^{\vert I \vert} (2\vert I \vert)!
		,
	\end{align*}
	where $C$ varies from line to line, the fact that $\vert I \vert^{\vert I \vert} \leq (2 \vert I \vert)!/\vert I \vert!$ has been used, and the latter holds by Lemma \ref{prop:summationlemma}, provided $C/R \leq 1/e$.
	Thus, by Lemma \ref{prop:LIlargep},
	\begin{align*}
		\big\vert L^I(f^n(t,x,p)) \vert_{t=1} \big\vert
		&
		\lesssim
		\sum_{l \leq 2 \vert I \vert}
		\sum_{\vert N\vert + \vert M\vert \leq \vert I\vert}
		\vert p \vert^{\vert N\vert (\frac{1-2q}{q})}
		\sum_{J}
		\sum_{k = 0}^{\infty}
		\vert p \vert^{\vert J \vert -k}
		\vert d^{I}_{klJNM} \vert
		\vert \partial_x^N \partial_p^M f^n_1(x,p) \vert
		\\
		&
		\lesssim
		C^{\vert I \vert} (2\vert I \vert)!
		\sum_{\vert N\vert + \vert M\vert \leq \vert I\vert}
		\vert p \vert^{\vert N\vert (\frac{1-2q}{q})}
		\vert \partial_x^N \partial_p^M f^n_1(x,p) \vert.
	\end{align*}
	In particular, \eqref{prop:comb2proof1} holds.
	
	Note that, for any $I$, $J$ with $\vert I \vert + \vert J \vert \leq k$,
	\begin{align*}
		\big\vert \partial_x^I \partial_p^Jf_1^n(x,p) \big\vert
		&
		\lesssim
		\sum_{\vert J_1\vert + \vert J_2\vert = \vert J \vert}
		\big\vert \partial_p^{J_1} \chi_n \big\vert
		\big\vert \partial_x^I \partial_p^{J_2} f_1(x,p) \big\vert
	\\
		&
		\lesssim
		(2\vert J \vert)!
		\,
		2^{\vert J \vert}
		\mathds{1}_{2^{n-1} \leq \vert p \vert \leq 2^{n+1}}
		\sum_{\vert J' \vert \leq \vert J \vert}
		\vert p \vert^{-(\vert J \vert - \vert J'\vert)}
		\big\vert \partial_x^I \partial_p^{J'} f_1(x,p) \big\vert
	\\
		&
		\lesssim
		2^k (2k)!
		\mathds{1}_{2^{n-1} \leq \vert p \vert \leq 2^{n+1}}
		\sum_{\vert J' \vert \leq \vert J \vert}
		\big\vert \partial_x^I \partial_p^{J'} f_1(x,p) \big\vert,
	\end{align*}
	using the property \eqref{eq:chigrowth}. Together with \eqref{prop:comb2proof1}, by Proposition \ref{prop:factorials}, this completes the proof.
\end{proof}

\subsubsection{The proof of Theorem \ref{thm:main1}}

The proof of Theorem \ref{thm:main1} can now be given.

\begin{proof}[Proof of Theorem \ref{thm:main1}]
	For $n \in \mathbb{N}$, define
	\[
		\mathcal{A}_n = \{ 2^{n-1} \leq \vert p\vert \leq 2^{n+1}\},
		\qquad
		\mathcal{A}_0 = \{ \vert p\vert \leq 2 \}.
	\]
	Suppose $n \in \mathbb{N}_0$, and consider first some $t^q \leq \upsilon \, 2^{n+1}$.  Now, for any multi index $I$,
	\[
		\partial_x^I \rho^n(t,x) = \int_{\mathcal{A}_n} \partial_x^I f^n(t,x,p) dp,
	\]
	and so, by Proposition \ref{prop:interp}, for any $k \geq 0$,
	\begin{align*}
		\sum_{\vert I \vert = k} \Vert \partial_x^I \rho^n(t,\cdot) \Vert_{L^2(\mathbb{T}^3)}
		\lesssim
		&
		\sum_{\vert I \vert = k}
		\left( \int_{\mathbb{T}^3} \int_{\mathbb{R}^3} \vert p \vert^2 \vert \partial_x^I f^n(t,x,p) \vert^2 dp dx \right)^{\frac{1}{4}}
		\left( \int_{\mathbb{T}^3} \int_{\mathbb{R}^3} \vert p \vert^4 \vert \partial_x^I f^n(t,x,p) \vert^2 dp dx \right)^{\frac{1}{4}}
		\!\!\! 
		\\
		\lesssim
		&
		\frac{1}{t^{6q}}
		\sum_{\vert I \vert = k}
		\left(
		\int_{\mathbb{T}^3} \int_{\mathcal{A}_n}
		(\vert p \vert^2+ \vert p \vert^4) \vert \partial_x^I f_1(x,p) \vert^2
		dp dx
		\right)^{\frac{1}{2}}.
	\end{align*}
	where the latter follows from the fact that $\partial_x^I$ commutes with the Vlasov equation (see Proposition \ref{prop:commutators}).  Now, for $t^q \leq \upsilon \, 2^{n+1}$ and $2^{n-1} \leq \vert p\vert \leq 2^{n+1}$,
	\[
		1
		\leq
		\frac{(\upsilon \, 2^{n+1})^{\frac{(1-2q)k}{q}}}{t^{(1-2q)k}}
		\leq
		\frac{(4 \upsilon)^{\frac{(1-2q)k}{q}}}{t^{(1-2q)k}} \vert p \vert^{\frac{(1-2q)k}{q}}.
	\]
	(When $n=0$, and thus $t^q \leq 2 \upsilon$ and $0 \leq \vert p\vert \leq 2$, then trivially $1 \leq (2 \upsilon)^{\frac{(1-2q)k}{q}} t^{-(1-2q)k}$.)
	Thus, by the Sobolev inequality of Proposition \ref{prop:Sobolev}, for any $k \geq 2$,
	\begin{equation} \label{eq:smalltest}
		\sup_{x\in \mathbb{T}^3} \vert \rho^n(t,x) - \overline{\rho^n}(t) \vert
		\lesssim
		\frac{(4 \upsilon)^{\frac{(1-2q)k}{q}}}{t^{6q+(1-2q)k}\sqrt{k}}
		\sum_{\vert I \vert = k}
		\left(
		\int_{\mathbb{T}^3} \int_{\mathcal{A}_n}
		(\vert p \vert^2+ \vert p \vert^4) (1+\vert p \vert)^{\frac{(1-2q)k}{q}} \vert \partial_x^I f_1(x,p) \vert^2
		dp dx
		\right)^{\frac{1}{2}}.
	\end{equation}
	
	Consider now some $t^q \geq \upsilon \, 2^{n+1}$.  The results of Proposition \ref{prop:rhovectors} then hold and so, using again Proposition \ref{prop:interp},
	\[
		\sum_{\vert I \vert = k} \Vert \partial_x^I \rho^n(t,\cdot) \Vert_{L^2(\mathbb{T}^3)}
		\lesssim
		\frac{C^k (4 k)!}{t^{6q+(1-2q)k}}
		\sum_{\vert I \vert \leq k}
		\left(
		\int_{\mathbb{T}^3} \int_{\mathcal{A}_n}
		(\vert p \vert^2+ \vert p \vert^4)
		\big\vert L^I(f^n(t,x,p)) \vert_{t=1} \big\vert^2
		dp dx
		\right)^{\frac{1}{2}}.
	\]
	Thus Proposition \ref{prop:Sobolev} and Proposition \ref{prop:comb2} imply that, for any $k \geq 2$,
	\begin{equation} \label{eq:largetest}
		\sup_{x\in \mathbb{T}^3} \vert \rho^n(t,x) - \overline{\rho^n}(t) \vert
		\lesssim
		\frac{C^k (9 k)!}{t^{6q+(1-2q)k}\sqrt{k}}
		\sum_{\vert I \vert + \vert J \vert \leq k}
		\left(
		\int_{\mathbb{T}^3} \int_{\mathcal{A}_n}
		(\vert p \vert^2+ \vert p \vert^4) (1+\vert p \vert)^{\vert I \vert \frac{1-2q}{q}}
		\vert \partial_x^I \partial_p^J f_1(x,p) \vert^2
		dp dx
		\right)^{\frac{1}{2}}.
	\end{equation}
	Combining \eqref{eq:smalltest} and \eqref{eq:largetest} and summing over all $n$ then gives
	\begin{equation} \label{eq:proofmainestimate}
		\sup_{x\in \mathbb{T}^3} \vert \rho(t,x) - \overline{\rho}(t) \vert
		\lesssim
		\frac{C^k (9 k)!}{t^{6q+(1-2q)k}\sqrt{k}}
		\Vert f_1 \Vert_{H^k_q},
	\end{equation}
	for all $k \geq 2$, which completes the proof of \eqref{eq:main2}.
	
	Suppose now that $f_1 \in H^{\omega}_q(\mathbb{T}^3 \times \mathbb{R}^3)$.  Note that
	\[
		\Vert f_1 \Vert_{H^k_q}
		\lesssim
		\frac{k!}{\lambda^k}
		\Vert f_1 \Vert_{H^{\omega}_q}
		,
	\]
	for all $k \geq 0$ and with $\lambda = \lambda(f_1)$.
	For all $t \geq 1$, $f$ satisfies \eqref{eq:proofmainestimate} for all $k \geq2$, and so
	\begin{equation} \label{eq:proofmainestimate2}
		\sup_{x\in \mathbb{T}^3} \vert \rho(t,x) - \overline{\rho}(t) \vert
		\lesssim
		\frac{1}{t^{6q}} \frac{C^k (10 k)!}{t^{(1-2q)k} \lambda^k \sqrt{k}} \Vert f_1 \Vert_{H^{\omega}_q}
		=
		\frac{1}{t^{6q}}
		\frac{(10k)!}{(\mu t^{\frac{1-2q}{10}})^{10k}\sqrt{10k}} \Vert f_1 \Vert_{H^{\omega}_q},
	\end{equation}
	for all $k \geq 2$, where
	\[
		\mu = \Big( \frac{\lambda(f_1)}{C} \Big)^{\frac{1}{10}}.
	\]
	In particular, for $\hat{t}(t):=\lfloor \frac{1}{10} \mu t^{\frac{1-2q}{10}} \rfloor$, where $\lfloor \cdot \rfloor$ is the floor function, \eqref{eq:proofmainestimate2} holds for $k=\hat{t}(t)$ and so (since $n! \,e^n \lesssim \sqrt{n} n^n$ for all $n$ --- see Proposition \ref{prop:Stirling} --- and $e^{-\lfloor s \rfloor} \leq e e^{-s}$ for all $s\geq 1$),
	\[
		\sup_{x\in \mathbb{T}^3} \vert \rho(t,x) - \overline{\rho}(t) \vert
		\lesssim
		\frac{1}{t^{6q}} \frac{(10 \, \hat{t}(t))!}{(10 \, \hat{t}(t))^{10 \, \hat{t}(t)+\frac{1}{2}}} \Vert f_1 \Vert_{H^{\omega}_q}
		\lesssim
		\frac{1}{t^{6q}} e^{-10 \, \hat{t}(t)} \Vert f_1 \Vert_{H^{\omega}_q}
		\lesssim
		\frac{1}{t^{6q}} e^{- \mu t^{\frac{1-2q}{10}}} \Vert f_1 \Vert_{H^{\omega}_q},
	\]
	which concludes the proof of \eqref{eq:main3}.
\end{proof}

\begin{remark}[Weak convergence to spatial average] \label{rmk:weakconvagain}
	Recall the weak convergence statement of Remark \ref{rmk:weakconv}.  For any $k \geq 0$, any smooth compactly supported function $\phi \colon \mathbb{R}^3 \to \mathbb{R}$, and any multi-index $I$ with $\vert I \vert = k$, it follows as in the proof of Theorem \ref{thm:main1} above that
	\begin{align*}
		&
		\Big\vert \partial_x^I \int_{\mathbb{R}^3} f(t,x, t^{-2q} p) \phi(p) dp \Big\vert
		=
		\Big\vert t^{6q} \partial_x^I \int_{\mathbb{R}^3} f(t,x, p) \phi(t^{2q}p) dp \Big\vert
		\\
		&
		\qquad \qquad
		=
		\Big\vert
		\frac{t^{6q}}{t^{k(1-2q)}}
		\sum_{I_1+I_2=I}
		\int
		\big(
		M+\Phi_q(t^{-2q}\vert r \vert^2) \, t^{-2q} \,r 
		\big)^{I_1} f (t,x,p) 
		\big(
		M+\Phi_q(t^{-2q}\vert r \vert^2) \, t^{-2q} \,r
		\big)^{I_2} (\phi(t^{2q}p)) dp
		\Big\vert
		\\
		&
		\qquad \qquad
		\lesssim
		\frac{t^{6q}}{t^{k(1-2q)}}
		\sum_{\vert I_1\vert, \vert I_2 \vert \leq k}
		\int
		\big\vert
		\big(
		M+\Phi_q(t^{-2q}\vert r \vert^2) \, t^{-2q} \,r 
		\big)^{I_1} f (t,x,p) 
		\big\vert
		\big\vert
		(\partial^{I_2} \phi) (t^{2q}p))
		\big\vert
		dp
		,
	\end{align*}
	(for the final step recall the inequality \eqref{eq:derivrhoproofmain} and note that $L_i(\phi(t^{2q}p)) = (\partial_i\phi)(t^{2q}p)$ for $i=1,2,3$) so that, for $k \geq 2$,
	\[
		\sup_{x\in \mathbb{T}^3}
		\Big\vert
		\int_{\mathbb{R}^3} \big( f(t,x, t^{-2q} p) - \overline{f_1}(p) \big) \phi(p) dp
		\Big\vert
		\lesssim
		\sum_{\vert I \vert = k}
		\Big\Vert
		\partial_x^I \int_{\mathbb{R}^3} f(t,\cdot, t^{-2q} p) \phi(p) dp
		\Big\Vert_{L^2_x}
		\leq
		C_k
		\frac{\Vert f_1 \Vert_{H^k_q}}{t^{k(1-2q)}}
		\sum_{\vert I \vert \leq k} \Vert \partial^I \phi \Vert_{L^{\infty}}.
	\]
	This statement can be used to obtain the weak convergence statement of Remark \ref{rmk:weakconv}.
\end{remark}

\appendix

\section{FLRW solutions of the Einstein equations}
\label{subsec:FLRW}

There are many examples of matter models for which the Einstein equations \eqref{eq:Einstein} admit solutions of the form \eqref{eq:FLRWgeneral} --- often with $q$ in the range covered by Theorem \ref{thm:main1} and Theorem \ref{thm:main2}.  Some notable examples are listed here.

\subsection*{Einstein--Euler: $\frac{1}{3} \leq q \leq \frac{2}{3}$}

The Euler equations on a given spacetime $(\mathcal{M},g)$ concern a fluid four velocity --- a vector field $u \in \mathfrak{X}(\mathcal{M})$, normalised so that $g_{\mu \nu} u^{\mu} u^{\nu} = -1$ --- and fluid pressure and density $p, \rho : \mathcal{M} \to [0,\infty)$.  The pressure and density are typically related by an equation of state, which we assume here to be linear:
\begin{equation} \label{eq:eqnstate}
	p = c_s^2 \rho,
\end{equation}
for some $c_s^2 \in [0,1]$.  The constant $c_s$ is known as the \emph{speed of sound}.  The Euler equations result from equating to zero the spacetime divergence of the energy momentum tensor
\begin{equation} \label{eq:TEuler}
	T^{\mu \nu}
	=
	(\rho + p) u^{\mu} u^{\nu} + p g^{\mu \nu}.
\end{equation}

The \emph{Einstein--Euler system} --- namely the Euler equations coupled to the Einstein equations \eqref{eq:Einstein} with this definition of $T$ --- admits the FLRW solution,
\begin{equation} \label{eq:EulerFLRW}
	g = - dt^2 + t^{\frac{4}{3\gamma}} \big( (dx^1)^2 + (dx^2)^2 + (dx^3)^2 \big),
	\qquad
	u^{\mu} = \delta^{\mu}_0,
	\qquad
	\rho = \frac{4}{3\gamma^2} t^{-2},
\end{equation}
on $\mathcal{M} = (0,\infty) \times \mathbb{T}^3$, where $\gamma = 1+c_s^2 \in [1,2]$.

\subsection*{Einstein--scalar field: $q = \frac{1}{3}$}

The Einstein--scalar field system on a spacetime $(\mathcal{M},g)$ consists of the scalar wave equation
\begin{equation} \label{eq:wave}
	\Box_g \psi = 0,
\end{equation}
coupled to the Einstein equations \eqref{eq:Einstein} via
\begin{equation} \label{eq:emtensorsf}
	T_{\mu \nu} = \partial_{\mu} \psi \partial_{\nu} \psi - \frac{1}{2} g_{\mu \nu} g^{\alpha \beta} \partial_{\alpha} \psi \partial_{\beta} \psi.
\end{equation}
The system admits an FLRW solution of the form
\begin{equation} \label{eq:FLRWsf}
	g = - dt^2 + t^{\frac{2}{3}} \big( (dx^1)^2 + (dx^2)^2 + (dx^3)^2 \big),
	\qquad
	\psi = \sqrt{\frac{2}{3}}\log t.
\end{equation}
The Einstein--scalar field system arises, as an appropriate special case, from the above Einstein--Euler system when the fluid is \emph{irrotational}, and \emph{stiff}, so that $c_s = 1$.

\subsection*{Einstein--nonlinear scalar field: $q>0$}

Given $V\colon \mathbb{R} \to \mathbb{R}$, consider the scalar field equation with potential on a given spacetime $(\mathcal{M},g)$,
\begin{equation} \label{eq:waveV}
	\Box_g \psi - V'(\psi) = 0.
\end{equation}
Define
\begin{equation} \label{eq:emtensorsfexp}
	T_{\mu \nu}[\psi]
	=
	\partial_{\mu} \psi \partial_{\nu} \psi
	-
	\Big(
	\frac{1}{2} g^{\alpha \beta} \partial_{\alpha} \psi \partial_{\beta} \psi
	+
	V(\psi)
	\Big) g_{\mu \nu}.
\end{equation}
Note that $T$ is divergence free if $\psi$ satisfies \eqref{eq:waveV}.  The Einstein--nonlinear scalar field system then consists of the Einstein equations \eqref{eq:Einstein} coupled to \eqref{eq:waveV}.  Given $q>0$, if $V$ takes the exponential form
\begin{equation} \label{eq:expV}
	V(\psi) = q(3q-1) e^{- \psi \sqrt{\frac{2}{q}}},
\end{equation}
then system admits an FLRW solution of the form
\begin{equation} \label{eq:FLRWsfexp}
	g = - dt^2 + t^{2q} \big( (dx^1)^2 + (dx^2)^2 + (dx^3)^2 \big),
	\qquad
	\psi = \sqrt{2q} \log t.
\end{equation}
Such solutions are often known as \emph{power law inflationary} spacetimes.  See, for example, \cite{Hal, HeRe, Rin, Ber}.  Note that, when $q=\frac{1}{3}$, the potential \eqref{eq:expV} vanishes, equation \eqref{eq:waveV}--\eqref{eq:emtensorsfexp} reduces to the previous \eqref{eq:wave}--\eqref{eq:emtensorsf}, and the solution \eqref{eq:FLRWsfexp} is equal to the previous solution \eqref{eq:FLRWsf}.

\subsection*{Einstein--massless Vlasov: $q = \frac{1}{2}$}

For a fixed particle mass $m \geq 0$, the Einstein--Vlasov system consists of the Vlasov equation
\begin{equation} \label{eq:EVVlasov}
	p^0 \partial_t f + p^i \partial_{x^i} f - p^{\mu} p^{\nu} \Gamma_{\mu \nu}^i \partial_{p^i} f = 0,
	\qquad
	g_{\mu \nu} p^{\mu} p^{\nu} = -m^2,
	\qquad
	\Gamma_{\mu \nu}^{\alpha} = \frac{g^{\alpha \beta}}{2} \big( \partial_{\mu} g_{\beta \nu} + \partial_{\nu} g_{\mu \beta} -  \partial_{\beta} g_{\mu \nu} \big),
\end{equation}
for $f\colon P \to [0,\infty)$, where $P = \{ (t,x,p) \in T\mathcal{M} \mid g(p,p) = -m^2\}$, coupled to the Einstein equations \eqref{eq:Einstein} with
\begin{equation} \label{eq:EVVlasov2}
	T^{\mu \nu}(t,x)
	=
	\int_{P_{(t,x)}}
	f(t,x,p) p^{\mu} p^{\nu}
	\frac{\sqrt{- \det g }}{- p_0}
	dp^1 dp^2 dp^3,
\end{equation}

When the particle mass $m =0$, the system is known as the Einstein--massless Vlasov system.  In this case, there is an infinite dimensional family of explicit FLRW solutions, where the metric $g$ takes the form \eqref{eq:FLRWgeneral} with $q=\frac{1}{2}$.
Indeed, for any smooth, sufficiently decaying function $F : \mathbb{R}^3 \to [0,\infty)$ satisfying the integral conditions
\[
	\int_{\mathbb{R}^3} F(p) p^i dp = 0, \quad i=1,2,3,	
	\qquad
	\text{and}
	\qquad
	\int_{\mathbb{R}^3} F(p) \frac{p^i p^j}{\vert p \vert} dp = \frac{\varrho}{3} \delta^{ij}, \quad i,j=1,2,3.
\]
such that $F \not\equiv 0$, the metric and particle density
\[
	g = -dt^2 + a(t)^2 \big( (dx^1)^2 + (dx^2)^2 + (dx^3)^2 \big),
	\qquad
	f(t,x,p)
	=
	F (a(t)^2 p)
	,
\]
where
\[
	a(t) = t^{\frac{1}{2}} \left( \frac{4 \varrho}{3} \right)^{\frac{1}{4}},
	\qquad
	\varrho = \int \vert p \vert F(p) dp,
\]
solve the Einstein--massless Vlasov system.

\subsection*{Einstein--massless Boltzmann: $q = \frac{1}{2}$}

The Einstein--massless Boltzmann system consists of the Boltzmann equation
\[
	p^0 \partial_t f + p^i \partial_{x^i} f - p^{\mu} p^{\nu} \Gamma_{\mu \nu}^i \partial_{p^i} f = Q(f,f),
	\qquad
	g_{\mu \nu} p^{\mu} p^{\nu} = 0,
\]
coupled to the Einstein equations \eqref{eq:Einstein} with $T^{\mu \nu}$ defined again by \eqref{eq:EVVlasov2}.  Here
\[
	Q(f,f)
	=
	\int_{P_{(t,x)}(q)}
	\int_{\Sigma_{p,q}}
	\big(f(p^{\prime })f(q^{\prime})-f(p)f(q)\big)A(p,q,p',q') d\mathrm{Vol},
\]
is a \emph{collision operator}, involving an integral on the space $\Sigma_{p,q} = \{(p',q')\in P_{(t,x)} \times P_{(t,x)} \mid p'+q'=p+q \}$ equipped with a suitable volume form, and $A(p,q,p',q')$ is a suitable \emph{cross-section} (for more details see, for example, the textbook of Choquet-Bruhat \cite{ChBr}).  This system admits an FLRW solution, with $f$ a Maxwell--J\"{u}ttner distribution,
\[
	g = -dt^2 + 2t \big( (dx^1)^2 + (dx^2)^2 + (dx^3)^2 \big),
	\qquad
	f(t,x,p)
	=
	\frac{1}{8\pi}e^{-\vert 2t p \vert}.
\]

\subsection*{Einstein--massive Vlasov: $q \sim \frac{2}{3}$}

The system obtained from coupling \eqref{eq:Einstein} to \eqref{eq:EVVlasov}--\eqref{eq:EVVlasov2} when $m>0$ is known as the Einstein--massive Vlasov system.  There is again an infinite dimensional family of FLRW solutions, but they are no longer explicit.  Indeed, for any suitably decaying function $\mu \colon [0,\infty) \to [0,\infty)$ such that $\mu \not\equiv 0$, there exists an FLRW solution
\[
	g = -dt^2 + a(t)^2 \big( (dx^1)^2 + (dx^2)^2 + (dx^3)^2 \big),
	\qquad
	f(t,x,p)
	=
	\mu (\vert a(t)^2 p \vert^2)
	,
\]
on $(0,\infty) \times \mathbb{T}^3$, where $a(t)$ is no longer explicit but solves the ordinary differential equation
\begin{equation} \label{eq:VlasovODE}
	a'(t)
	=
	a(t) \varrho(a(t))^{\frac{1}{2}} \sqrt{\frac{1}{3}},
\qquad
	\varrho(a)
	:=
	\int \sqrt{m^2+\vert p \vert^2} \mu( \vert a p \vert^2) dp^1 dp^2 dp^3.
\end{equation}
Solutions $a(t)$ of \eqref{eq:VlasovODE} behave like $t^{\frac{2}{3}}$ as $t \to \infty$ (and like $t^{\frac{1}{2}}$ as $t \to 0$).

\bibliography{mixingFLRW}{}
\bibliographystyle{plain}

\end{document}